\newif\ifabstract

\abstractfalse 
\newif\iffull
\ifabstract \fullfalse \else \fulltrue \fi

\ifabstract
\documentclass[11pt]{article}
\usepackage{fullpage}
\fi

\iffull
\documentclass[11pt]{article}
\usepackage{fullpage}
\fi

\usepackage{amssymb,amsmath}
\usepackage{amsthm}
\usepackage{graphicx}
\usepackage{subfigure}
\usepackage{graphics}
\usepackage{url}
\usepackage{hyperref}
\usepackage{color}

\newtheorem{theorem}{Theorem}[section]
\newtheorem{lemma}[theorem]{Lemma}
\newtheorem{proposition}[theorem]{Proposition}

\newtheorem{corollary}[theorem]{Corollary}
\newtheorem{definition}[theorem]{Definition}

\newcommand{\patch}{\mathsf{patch}}
\newcommand{\irrelevant}{\mathsf{irrelevant}}
\newcommand{\myapprox}{\mathsf{approx}}
\newcommand{\approxfactor}{\ensuremath{\alpha_{\myapprox}}}

\newcommand{\pruning}{\mathsf{pruning}}
\newcommand{\AlgPatch}{\ensuremath{\mathsf{Patch}}}
\newcommand{\AlgIrrelevant}{\ensuremath{\mathsf{Irrelevant}}}
\newcommand{\AlgPruning}{\ensuremath{\mathsf{Pruning}}}
\newcommand{\AlgPruningDecision}{\ensuremath{\mathsf{Pruning\text{-}Decision}}}
\newcommand{\AlgApprox}{\ensuremath{\mathsf{Approx}}}

\newcommand{\inner}{\mathsf{inner}}
\newcommand{\apex}{\mathsf{apex}}
\newcommand{\myin}{\mathsf{in}}
\newcommand{\myout}{\mathsf{out}}
\newcommand{\father}{\mathsf{father}}
\newcommand{\cFHL}{c_{\text{FHL}}}
\newcommand{\XXX}{\textcolor{red}{XXX}}

\newcommand{\len}{\mathsf{len}}

\newcommand{\crossingnumber}{\mathsf{cr}}

\newcommand{\eg}{\mathsf{eg}}
\newcommand{\mvp}{\mathsf{mvp}}
\newcommand{\tw}{\mathsf{tw}}
\newcommand{\pw}{\mathsf{pw}}

\newcommand{\dmax}{\Delta}
\newcommand{\OPT}{\mathsf{OPT}}

\newcommand{\framed}{\mathsf{framed}}

\newcommand{\nil}{\mathsf{nil}}
\newcommand{\row}{\mathsf{row}}
\newcommand{\col}{\mathsf{col}}
\newcommand{\cost}{\mathsf{cost}}
\newcommand{\width}{\mathsf{width}}
\newcommand{\cover}{\mathsf{cover}}
\newcommand{\OPTMWC}{\mathsf{OPT}_{\textsc{MWC}}}
\newcommand{\OPTNWMWC}{\mathsf{OPT}_{\textsc{NWMWC}}}

\newcommand{\NWMWC}{\textsc{Node Weighted Multiway-Cut}}

\DeclareMathOperator*{\mymod}{mod}

\newcommand{\eps}{\varepsilon}

\newcommand{\calI}{{\cal I}}
\newcommand{\hatF}{\widehat{F}}
\newcommand{\comment}[1]{}

\ifabstract

\fi
\iffull

\fi

\title{Polylogarithmic approximation for minimum planarization (almost)}

\author{Ken-ichi Kawarabayashi\thanks{National Institute of Informatics, 2-1-2, Hitotsubashi, Chiyoda-ku, Tokyo, Japan.
    \texttt{k\_keniti@nii.ac.jp}.
    Supported by JST ERATO Kawarabayashi Large Graph Project.}
\and
Anastasios Sidiropoulos\thanks{
Dept.~of Computer Science and Engineering, and Dept.~of Mathematics, The Ohio State University. Columbus, OH, 43210.
\texttt{sidiropoulos.1@osu.edu}.
Supported by NSF grant CCF 1423230 and by CAREER 1453472.
}
}

\date{}

\begin{document}

\maketitle

\begin{abstract}
In the \emph{minimum planarization} problem, given some $n$-vertex graph, the goal is to find a set of vertices of minimum cardinality whose removal leaves a planar graph.
This is a fundamental problem in topological graph theory.
We present a $\log^{O(1)} n$-approximation algorithm for this problem on general graphs with running time $n^{O(\log n/\log\log n)}$.
We also obtain a $O(n^\eps)$-approximation with running time $n^{O(1/\eps)}$ for any arbitrarily small constant $\eps > 0$.
Prior to our work, no non-trivial algorithm was known for this problem on general graphs, and the best known result even on graphs of bounded degree was a $n^{\Omega(1)}$-approximation \cite{chekuri2013approximation}.

As an immediate corollary, we also obtain improved approximation algorithms for the crossing number problem on graphs of bounded degree.
Specifically, we obtain $O(n^{1/2+\eps})$-approximation and $n^{1/2} \log^{O(1)} n$-approximation algorithms in time $n^{O(1/\eps)}$ and $n^{O(\log n/\log\log n)}$ respectively.
The previously best-known result was a polynomial-time $n^{9/10}\log^{O(1)} n$-approximation algorithm \cite{DBLP:conf/stoc/Chuzhoy11}.

Our algorithm introduces several new tools including an efficient grid-minor construction for apex graphs, and a new method for computing irrelevant vertices. Analogues of these tools were previously available only for exact algorithms.
Our work gives efficient implementations of these ideas in the setting of approximation algorithms, which could be of independent interest.
\end{abstract}

\section{Introduction}

In the \emph{minimum planarization} problem, given a graph $G$, the goal is to find a set of vertices of minimum cardinality whose removal leaves a planar graph.
This is a fundamental problem in topological graph theory, which been extensively studied over the past 40 years.
It generalizes planarity, and has connections to several other problems, such as crossing number and Euler genus.
The problem is known to be fixed-parameter tractable \cite{DBLP:conf/focs/Kawarabayashi09,DBLP:journals/algorithmica/MarxS12,DBLP:conf/soda/JansenLS14}, but very little is known about its approximability.



\subsection{Our contribution}

Prior to our work, no non-trivial approximation algorithm for minimum planarization was known for general graphs.
The only prior result was a $n^{\Omega(1)}$-approximation for graphs of bounded degree \cite{chekuri2013approximation}.
We present the first non-trivial approximation algorithms for this problem on general graphs.
Our main results can be summarized as follows: 

\begin{theorem}\label{thm:polylog}
There exists a $O(\log^{32} n)$-approximation algorithm for the minimum vertex planarization problem with running time $n^{O(\log n / \log\log n)}$.
\end{theorem}

\begin{theorem}\label{thm:eps}
For any arbitrarily small constant $\eps > 0$, 
there exists a $O(n^{\eps})$-approximation algorithm for the minimum vertex planarization problem with running time $n^{O(1/\eps)}$.
\end{theorem}


\paragraph{Applications to crossing number.}

The \emph{crossing number} of a graph $G$, denoted $\crossingnumber(G)$, is the minimum number of crossings in any drawing of $G$ into the plane (see \cite{DBLP:conf/stoc/Chuzhoy11}).
Prior to our work, the best-known approximation for the crossing number of bounded-degree graphs was due to Chuzhoy \cite{DBLP:conf/stoc/Chuzhoy11}.
Given a bounded-degree graph, her algorithm computes a drawing with $(\crossingnumber(G))^{10} \log^{O(1)} n$ crossings, which implies a $n^{9/10} \log^{O(1)} n$-approximation.
We now explain how our result on minimum planarization implies an improved approximation algorithm for crossing number on bounded-degree graphs.
It is easy to show that for any graph $G$, $\mvp(G)\leq \crossingnumber(G)$, simply by removing one endpoint of one edge involved in each crossing in some optimal drawing.
Thus, using our $\alpha$-approximation algorithm for minimum planarization, we can compute a planarizing set of size at most $\alpha \cdot \crossingnumber(G)$.
Thus, in graphs of maximum degree $\dmax$, we can compute some $F\subset E(G)$, with $|F|\leq \alpha \dmax \crossingnumber(G)$, such that $G\setminus F$ is planar.
Chimani and Hlinen\'{y} \cite{DBLP:journals/jco/ChimaniH17} (see also \cite{DBLP:conf/soda/ChuzhoyMS11}) have given a polynomial-time algorithm which given some graph $G$ and some $F\subset E(G)$, such that $G\setminus F$ is planar, computes a drawing of $G$ with at most $O(\dmax^3 \cdot |F| \cdot \crossingnumber(G) + \dmax^3 \cdot |F|^2)$ crossings.
Combining this with our result we immediately obtain an algorithm with running time $n^{O(\log n / \log\log n)}$, which given a graph $G$ of bounded degree, computes a drawing of $G$ with at most $(\crossingnumber(G))^2 \log^{O(1)} n$ crossings.
Similarly, we obtain an algorithm with running time $n^{O(1/\eps)}$, which given a graph $G$ of bounded degree, computes a drawing of $G$ with at most $(\crossingnumber(G))^2 n^{\eps} \log^{O(1)} n$ crossings, for any fixed $\eps>0$.
Combining this with existing approximation algorithms for crossing number of graphs of bounded degree that are based on balanced separators, we obtain the following (see \cite{DBLP:conf/stoc/Chuzhoy11} for details).

\begin{theorem}
There exists a $n^{1/2}\log^{O(1)} n$-approximation algorithm for the crossing number of graphs of bounded degree, with running time $n^{O(\log n / \log\log n)}$.
Furthermore there exists a $n^{1/2+\eps}$-approximation algorithm for the crossing number of graphs of bounded degree, with running time $n^{O(1/\eps)}$, for any fixed $\eps>0$.
\end{theorem}

\subsection{Related work}

 In the \emph{$\cal{F}$-deletion problem},
 the goal is to compute a minimum vertex set $S$ in an input graph $G$ such that $G-S$ is $\cal{F}$-minor-free. Characterizing graph properties for which the corresponding vertex deletion problem can be approximated within a constant factor or a polylogarithmic factor is a long standing open problem in approximation algorithms \cite{yan2,DBLP:conf/icalp/LundY93}. In spite of a long history of research, we are still far from resolving the status of this problem. Constant-factor approximation algorithms for the vertex Cover problem (i.e., ${\cal F}=C_3$) are known since 1970s \cite{vertexcover1, vertexcover2}.

  Yannakakis \cite{yan1} showed that approximating the minimum vertex set that needs to be deleted in order to obtain a connected graph with some property $P$ within factor $n^{1-\eps}$ is NP-hard, for a very broad class of properties (see \cite{yan1}). There was not much progress on approximability/non-approximability of vertex deletion problems until quite recently. Fomin et al.~\cite{fomin1} showed that for every graph property $P$ expressible by a finite set of forbidden minors ${\cal F}$ containing at least one planar graph, the vertex deletion problem for property $P$ admits a constant factor approximation algorithm. They explicitly mentioned that the most interesting case is when $\cal{F}$ contains a non-planar graph (they said that perhaps the most interesting case is when ${\cal F} =\{K_{3,3}, K_5\}$), because there is no poly-logarithmic factor approximation algorithm so far. Indeed, the planar graph case and the non-planar case for the family $\cal{F}$ may be quite different, as the graph minor theory
  suggests.
 The main result of this paper almost settles the most interesting case. We believe that our techniques can lead to further results on approximation algorithms for minor-free properties.


\section{High level description of the algorithm}

We now give a brief overview of our approach and highlight some of the main challenges.
Our approximation algorithm is inspired by fixed-parameter algorithms for the minimum planarization problem, where one assumes that the size of the optimum planarizing set is some fixed constant $k$ (see \cite{DBLP:conf/focs/Kawarabayashi09,DBLP:journals/algorithmica/MarxS12,DBLP:conf/soda/JansenLS14}).

\subsection{Overview of previous fixed-parameter algorithms.}
The known fixed-parameter algorithms for minimum planarization work as follows:
If the treewidth of the input graph $G$ is large enough (say, $k^c$, for some constant $c>0$), then 
one can efficiently compute a large grid minor $H$ in $G$ (that is, a minor of $G$ that is isomorphic to some grid).
A subgraph $J$ of $G$ is called \emph{flat} if it admits some planar drawing with outer face $F$, such that all edges between $J$ and $G\setminus J$ have one endpoint in $F$.
If some vertex $v\in H$ is surrounded by a flat subgrid of $H$ of size $\Omega(k)$, then it is \emph{irrelevant}; this means that by removing $v$, we do not change any optimal solution.
Thus, if such an irrelevant vertex $v$ exists, we recurse on $G\setminus \{v\}$, and return the optimum solution found.
We define the \emph{face cover} of some set of vertices $U$ to be the minimum number of faces of $H$ that are needed to cover $U$.
If there exists some vertex $u$ such that the neighborhood of $u$ has face cover of size $\Omega(k)$, then $u$ is \emph{universal}; that is, removing $u$ decreases the size of some optimum planarizing set by 1.
Thus, if such a universal vertex $u$ exists, we recurse on $G\setminus \{u\}$, and return the optimum solution found, together with $u$.
If the grid $H$ is large enough, then we can always find either an irrelevant or a universal vertex.
Thus, by repeatedly removing such vertices, we arrive at a graph of bounded treewidth, where the problem can be solved using standard dynamic programming techniques.

\subsection{Obtaining an approximation algorithm.}
We now discuss the main challenges towards extending the above approach to the approximate setting.
In order to simplify the exposition, we discuss the $\log^{O(1)}$-approximation algorithm.
The $n^\eps$-approximation is essentially identical, after changing some parameters.

\begin{description}
\item{\textbf{1. The small treewidth case.}}
In the above fixed-parameter algorithms, the problem is eventually reduced to the bounded-treewidth case.
That is, one has to solve the problem on a graph of treewidth $f(k)$, for some function $f$.
Since the optimum $k$ is assumed to be constant, this can be done in polynomial time (in fact, linear time).
However, in our setting, $k$ can be as large as $\Omega(n)$, and thus this approach is not applicable.
Instead, we try to find some small balanced vertex separator $S$.
If the treewidth is at most $k\log^{O(1)} n$, then we can find some separator of size $k\log^{O(1)} n$.
In this case, we recurse on all non-planar connected components of $G\setminus S$, and we add $S$ to the final solution.
It can be shown that $|S|$ can be charged to the optimum solution, so that the total increase in the cost of the solution is $k\log^{O(1)}$.

\item{\textbf{2. The large treewidth case.}}
We say that a graph is \emph{$k$-apex} if it can be made planar by the removal of at most $k$ vertices.
Since any planar graph of treewidth $t$ contains a grid minor of size $\Omega(t)\times \Omega(t)$, it easily follows that any $k$-apex graph of treewidth $t>ck$, for some universal constant $c>0$, also contains a grid minor of size $\Omega(t)\times \Omega(t)$.
To see that, first delete some planarizing set of size $k$, and then find a grid minor in the resulting planar graph, which has treewidth at least $t-k$.
However, even thought it is trivial to prove the existence of such a large grid minor, computing it in polynomial time when $k$ is not fixed turns out to be a significant challenge.
We remark that it is known how to compute a grid minor of size $\Omega(k)\times \Omega(k)$ when $t=\Omega(k^{2})$ \cite{chekuri2013approximation,DBLP:conf/stoc/KawarabayashiS15}, and this is enough to obtain a $k^{O(1)}$-approximation algorithm.
However, in order to obtain a $\log^{O(1)} n$-approximation, we need to find a grid minor when $t= k \log^{O(1)} n$.

\item{\textbf{3. Doubly-well-linked sets}}
The first main technical contribution of this work is an algorithm for computing a large grid minor in $k$-apex graphs, when $k$ is not fixed.
Suppose that the treewidth of $G$ is $t>2k$.
As a first step, we compute some separation $(U,U')$ of order $t \log^{O(1)}n$ (that is, some $U,U'\subset V(G)$ with $V(G)=U\cup U'$ and $|U\cap U'|=t\log^{O(1)} n$), and some $Y\subseteq U\cap U'$ such that $Y$ is \emph{well-linked} in both sides of the separation.
Intuitively, for a set $Y$ to be well-linked in some graph $G'$ means that $G'$ does not have any sparse cuts, w.r.t.~$Y$; in other words, contracting $G'$ into $Y$ results in an ``expander-like'' graph (see Section \ref{sec:prelim} for a formal definition).
We refer to such a set $Y$ as \emph{doubly-well-linked}.
We remark that the notion of doubly-well-linked set considered here is similar to, and inspired by, the \emph{well-linked bipartitions} introduced by Chuzhoy \cite{DBLP:conf/stoc/Chuzhoy11} in her work on the crossing number problem.
It is well-known that in any graph, such a separation can be found so that $Y$ is well-linked in at least one of the two sides.
However, as we explain below, we need $Y$ to be well-linked in both sides of the separation.

\item{\textbf{4. From a doubly-well-linked set to a grid minor.}}
There are several algorithms for computing large grid minors in planar graphs \cite{chekuri2004edge,DBLP:journals/jct/RobertsonST94}.
A key ingredient in these algorithms is the duality between cuts and cycles in embedded planar graphs.
That is, any cut of a planar graph corresponds to a collection of cycles in its dual.
The algorithms for the planar case exploit this duality by first computing a well-linked set $Z$, and then finding some disk ${\cal D}$ in the plane, that contains $Z$ and has a large fraction of $Z$ on its boundary.
Then, one can find two sets of paths ${\cal P}$ and ${\cal Q}$, with endpoints on the boundary of ${\cal D}$, such that every path in ${\cal P}$ intersects every path in ${\cal Q}$.
By planarity, this yields a grid minor.

In our case we cannot apply this idea since we don't have a planar drawing of the graph (indeed, this is precisely what we want to compute).
However, it turns out that, intuitively, any doubly-well-linked set $Y$ behaves as a Jordan curve.
That is, if we remove any optimal solution from $G$ (that is, any planarizing set of minimum cardinality), then there exists a planar drawing of $G[U]$ such that most of the vertices in $Y$ are close to the outer face.
Since $Y$ is well-linked in $G[U]$, we can route in $G[U]$, with low congestion, a multicommodity flow that routes a unit demand between every pair of vertices in $Y$.
We then sample $c$ paths from this flow, for some sufficiently large constant $c>0$.
The fact that the congestion is low, can be used to deduce that the resulting paths will avoid all vertices in some optimal solution, with some constant probability.
Thus, the union of the sampled paths admits a planar drawing.
Furthermore, since $Y$ is doubly-well-linked, we can show that, with some constant probability, the union of the sampled paths can be drawn so that their endpoints are all in the outer face.
We thus use the union of the sampled paths to construct a \emph{skeleton} graph.
Specifically, we find a suitable subgraph of $G$ which does not intersect some optimal solution.
We then sample $t/\log^{O(1)}n$ more paths from the flow, and partition them into two sets ${\cal P}$ and ${\cal Q}$, depending on the structure of their intersection with the skeleton graph.
Conditioned on the event that the skeleton graph does not intersect some optimal solution, we can find sets of paths ${\cal P}$ and ${\cal Q}$ such that every path in ${\cal P}$ intersects every path in ${\cal Q}$.

\item{\textbf{5. Computing a partially triangulated grid minor.}}
Having computed a large grid minor $H$, we wish to use $H$ to find either universal or irrelevant vertices.
To that end, we need to ensure that there are no edges between different faces of $H$.
We first compute some $X\subset V(G)$, such that $G\setminus X$ can be contracted into some grid $H'$ of size $t/\log^{O(1)} n \times t/\log^{O(1)} n$, where $H'$ is obtained by ``eliminating'' some rows and columns on $H$.

\item{\textbf{6. Computing a semi-universal set.}}
The next main technical challenge in our algorithm is the computation of universal vertices.
In the fixed-parameter algorithms described above, in order to compute a universal vertex, one needs a grid of size at least $\Omega(k^{2})$.
However, in our case, we only have a grid of size $k/\log^{O(1)} n \times k/\log^{O(1)} n$.
Thus, we cannot always find a universal vertex.
We overcome this obstacle by introducing the notion of a \emph{semi-universal} vertex:
We say that a set $A$ of vertices is semi-universal if deleting $A$ from $G$ decreases the cost of the optimum by at least $(2/3)|A|$.
We can prove that if the size of the neighborhood of $X$ in $H'$ is at least $k\log^{\Omega(1)} n$, then we can find some $A\subseteq X$ that is semi-universal.
Intuitively, the algorithm finds some $A\subseteq X$ that behaves as an expander: for every $A'\subseteq A$, the size of the face cover of $A'$ is at least $|A|/\log^{O(1)} n$.

\item{\textbf{7. Computing an irrelevant vertex.}}
The next technical difficulty is computing an irrelevant vertex, when the size of the neighborhood of $X$ in $H'$ is at most $k\log^{O(1)} n$.
As in the computation of universal vertices, this is a fairly easy task when the treewidth is at least $t=\Omega(k^{2})$.
However, here we can only find a grid minor of size $k\log^{O(1)} n \times k \log^{O(1)} n$.
We overcome this difficulty as follows:
We first partition the grid minor into subgrids, in a fashion similar to a quad-tree decomposition: 
There are $O(\log n)$ partitions of $H'$, each into subgrids of size $2^i\times 2^i$, for all $i\in \{0,\ldots,\log n\}$.
Then, for each subgrid in this collection of partitions, we compute an upper estimate on its planarization number:
If the number of vertices of $G$ in this subgrid is at least $n/\log n$, then we set the estimate to be $k$, and otherwise we recursive approximate its minimum planarization number.
Finally, we add to this estimate the number of neighbors of $X$ in this subgrid.
We say that some subgrid of $H'$ of size $2^i \times 2^i$ is \emph{active} if its upper estimate is at least $2^i/c$, for some constant $c>0$.
Since the size of the neighborhood of $X$ is small, we can show that there exists some $v \in V(H')$ that lies outsize all active subgrids; we can then show that $v$ must be irrelevant.

\item{\textbf{8. Embedding into a higher genus surface.}} 
Given the above algorithms for computing grid minors, semi-universal vertices, and irrelevant vertices, the algorithm proceeds as follows.
We iteratively compute one of the following (1) a small balanced separator, (2) a semi-universal set, or (3) a set of irrelevant vertices.
We refer to such a sequence of reductions as a \emph{pruning sequence}.
The graph obtained at the end of some pruning sequence is planar.
Let $X'$ be the set of vertices removed in Cases (1) and (2), throughout the pruning sequence.
We have $|X'|\leq k\log^{O(1)} n$.
We say that the \emph{cost} of the pruning sequence is $|X'|$.
However, the number of irrelevant vertices removed can be as large as $\Omega(n)$.
We need to add these vertices to the planar drawing of the resulting graph.
It turns out that this is not always possible. The reason is that the irrelevant vertices removed are only guaranteed to be irrelevant w.r.t.~any \emph{optimal} planarizing set; in contrast, the set $X'$ does not form an optimal planarizing set (indeed, we remove $k\log^{O(1)} n$ vertices).
We overcome this obstacle using a technique that was first introduced in \cite{chekuri2013approximation}:
When deleting a set of irrelevant vertices, we add a grid of width 3, referred to as a \emph{frame},  around the ``hole'' that is created.
The point of adding this grid is that we can inductively add the irrelevant vertices back to the graph as follows:
If $X'$ does not intersect the frame, then we can simply add the irrelevant vertices back to the graph without violating planarity.
Otherwise, if the frame intersets $m$ vertices from $X'$, then we can extend the current drawing to the irrelevant vertices corresponding to that frame by adding at most $\ell$ handles or antihandles.
This leads to an embedding into a new non-planar surface.
Repeating this process over all frames, we obtain an embedding of $G\setminus X'$ into some surface of Euler genus $k \log^{O(1)}$.

\item{\textbf{9. The final alorithm.}}
The last remaining step is to compute a planarizing set for $G\setminus X'$.
It turns out that this can be done by exploiting the embedding of $G\setminus X'$
into the surface ${\cal S}$ of Euler genus $g=k \log^{O(1)} n$, that was computed above.
Using tools from the theory of graphs of surfaces, we show that we can decrease the Euler genus of ${\cal S}$ by one, while deleting at most $O((1+k/g) \log n)$ vertices.
Repeating this process $g$ times, we obtain a planar graph after deleting a set $X''$ of at most $O(g \log n + k\log^2 n)$ vertices.
The final output of the algorithm is $X'\cup X''$, which is a planarizing set for $G$ of size $k\log^{O(1)}$.
\end{description}

\subsection{Organization}
The rest of the paper is organized as follows.
Section \ref{sec:prelim} introduces some basic definitions and results that are used throughout the paper.
Section \ref{sec:main} presents the main algorithm, by putting together the main ingredients of our approach.
Section \ref{sec:doubly} presents our algorithm for computing a doulby-well-linked set.
Section \ref{sec:pseudogrids} introduces the notion of a \emph{pseudogrid}, which are used to construct grid minors.
Section \ref{sec:grid} presents the algorithm for computing a grid minor.
Section \ref{sec:pt_grid} gives the algorithm for contracting the graph into a partially triangulated grid, with a small number of apices.
Section \ref{sec:semi-universal} shows how to compute a semi-universal set, given a partially-triangulated grid contraction, such that the apex set has a large neighborhood.
Section \ref{sec:irrelevant} shows how to compute irrelevant vertices, for the case where the apex set of the partially triangulated grid contraction has a small neighborhood.
Section \ref{sec:patch} shows how, given an algorithm for computing irrelevant vertices, we can compute a patch.
Section \ref{sec:pruning} combines the above algorithms for computing grid minors, semi-universal vertex sets, and patches, to obtain an algorithm for computing a pruning sequence of low cost.
Section \ref{sec:embedding} presents an algorithm which given given a pruning sequence of low cost, computes an embedding into a surface of low Euler genus.
Finally, Section \ref{sec:surface_planarization} gives the algorithm for planarizing a graph embedded into some non-planar surface.

\section{Definitions and preliminaries}\label{sec:prelim}

This section provides some basic notations needed in this paper.

\paragraph{Graph notation.}
Let $G$ be a graph and $X\subseteq V(G)$.
We use $d_G$ to denote the shortest path metric on $G$.
For any $v\in V(G)$ we write $d_G(X,v)=d_G(v,X)=\min_{u\in X} d_G(v,u)$.
For any $r\geq 0$ we define
$N_G(X,r)=\{v\in V(G) : d_G(v,X)\leq r\}$.
For a simple path $P$ and $u,v\in V(P)$ we denote by $P[u,v]$ the subpath of $P$ between $u$ and $v$.
We define the $(r\times \ell)$ grid to be the Cartesian product $P_r\times P_\ell$, where $P_i$ denote the path with $i$ vertices.
Let $H$ be the $(r\times \ell)$-grid.
For each $v\in V(H)$ we denote by $\row(v)\in \{1,\ldots,r\}$ and by $\col(v)\in \{1,\ldots,\ell\}$ the indexes of the row and column of $v$ respectively.
We denote by $\partial H$ the boundary cycle of $H$.
Every face in $H$ other than $\partial H$ is a cycle of length 4.
We say that some graph $H'$ is the \emph{partially triangulated $(r\times \ell)$-grid} if $H'$ is obtained from $H$ by adding for every face $F$ of $H$, with $F\neq \partial H$, at most one diagonal edge.

For some graph $G$, and some $X\subseteq V(G)$ such that $G\setminus X$ is planar, we say that $X$ is \emph{planarizing} (for $G$).
We denote by $\mvp(G)$ the minimum vertex planarization number of $G$, i.e.
\[
\mvp(G) = \min\{|X| : X\subseteq V(G) \text{ is planarizing for } G\}.
\]
We remark that deciding whether $\mvp(G)=0$ is precisely the problem of deciding whether $G$ is planar, which can be solved in linear time \cite{DBLP:journals/jacm/HopcroftT74}.
\begin{lemma}\label{lem:additivity}
Let $G$ be a graph and let $G_1,\ldots,G_{\ell}$ be a collection of pairwise vertex-disjoint subgraphs of $G$.
Then $\sum_{i=1}^{\ell} \mvp(G_i) \leq \mvp(G)$.
\end{lemma}
\begin{proof}
Let $S \subseteq V(G)$ with $|S|=\mvp(G)$ and such that $G\setminus S$ is planar.
For each $i\in \{1,\ldots,\ell\}$ let $S_i = S\cap V(G_i)$.
Then for any $i\in \{1,\ldots,\ell\}$, we have that $G_i\setminus S_i$ is a subgraph of $G\setminus S$, and thus it is planar.
Thus $S_i\geq \mvp(G_i)$.
It follows that $\sum_{i=1}^{\ell} \mvp(G_i) \leq \sum_{i=1}^{\ell} |S_i| \leq |S| = \mvp(G)$.
\end{proof}

\paragraph{Minors and contractions.}
A graph $H$ obtained via a sequence of zero or more edge contractions on a graph $G$ is called a \emph{contraction} of $G$.
Unless stated otherwise, we will replace a set of parallel edges in a contraction of a graph by a single edge.
The induced mapping $V(H) \to 2^{V(G)}$ is called the \emph{contraction mapping} (w.r.t.~$H$ and $G$).
Similarly, we say that $H$ is a \emph{minor} of $G$ if it can be  obtained via a sequence of zero or more edge contractions, edge deletions, and vertex deletions, performed on $G$.
The induced mapping $V(H) \to 2^{V(G)}$ is called the \emph{minor mapping} (w.r.t.~$H$ and $G$).
For any $X\subseteq V(G)$ we also write $\mu(X) = \bigcup_{v\in X} \mu(v)$.
For any $v\in V(G)$ we will abuse notation slightly and write $\mu^{-1}(v)$ to denote the unique vertex $v'\in H$ with $v\in \mu(v')$.

\paragraph{Treewidth, pathwidth, and grid minors.}

A \emph{tree decomposition} of a graph $G$ is a pair $(T,R)$, where
$T$ is a tree and $R$ is a family $\{R_t \mid t \in V(T)\}$ of
vertex sets $R_t\subseteq V(G)$, such that the following two
properties hold:

\begin{enumerate}
\item[(W1)] $\bigcup_{t \in V(T)} R_t = V(G)$, and every edge of $G$ has
both ends in some $R_t$.
\item[(W2)] If  $t,t',t''\in V(T)$ and $t'$ lies on the path in $T$
between $t$ and $t''$, then $R_t \cap R_{t''} \subseteq R_{t'}$.
\end{enumerate}
The \emph{width} of a tree decomposition $(T,R)$ is $\max\{|R_t|\mid
t\in V(T)\}-1$, and the \emph{treewidth} of $G$ is defined as
the minimum width taken over all tree decompositions of $G$.
If $T$ is a path, then we can define the \emph{pathwidth} of $G$ as
the minimum width taken over all path decompositions of $G$.
We use $\tw(G)$ and $\pw(G)$ to denote the treewidth and pathwidth of $G$ respectively.

We recall the following result on the treewidth of planar graphs.

\begin{lemma}[Robertson, Seymour and Thomas \cite{robertson1994quickly}]\label{lem:grid_planar}
Any planar graph of treewidth $t$ contains some $(r\times r)$-grid minor for some $r=\Omega(t)$.
\end{lemma}

From Lemma \ref{lem:grid_planar} we obtain the following.

\begin{lemma}[Existence of a large grid minor in a $k$-apex graphs]\label{lem:grid_existence}
Let $G$ be a $k$-apex graph of treewidht $t$.
Then $G$ contains some $(r\times r)$-grid minor for some $r=\Omega(t-k)$.
\end{lemma}
\begin{proof}
Let $X\subset V(G)$ with $|X|=k$ such that $H=G\setminus X$ is planar.
We have $\tw(H)\geq \tw(G)-k=t-k$.
By Lemma \ref{lem:grid_planar} we have that $H$ contains the $(r\times r)$-grid as a minor for some $r=\Omega(\tw(H)) = \Omega(t-k)$.
\end{proof}

The following is a related result needed in this paper.

\begin{lemma}[Eppstein \cite{Eppstein_grids}]\label{lem:grids_persistency}
  Let $r,f \geq 1$.  Let $G$ be the $(r \times r)$-grid, and $X\subset
  V(G)$, with $|X| = f$.  Then, $G\setminus X$ contains the $(r'
  \times r')$-grid as a minor, where $r' = \Theta(\min\{r, r^2/f\})$.
\end{lemma}

\paragraph{Well-linked sets.}

Our proof needs to handle a graph of large tree-width. It is well-known that such a graph must have a ``highly-connected'' subset, which
is often referred to as ``well-linked''.

\begin{definition}[Cut-linked set]
Let $G$ be a graph, let $X\subseteq V(G)$ and $\alpha>0$.
We say that $X$ is \emph{$\alpha$-cut-linked} (in $G$), iff for any partition of $V(G)$ into $\{A,B\}$, we have
\[
|E(A,B)|\geq \alpha \cdot \min\{|A\cap X|, |B\cap X|\}.
\]
\end{definition}

\paragraph{Graph on surfaces.}

A \emph{drawing} of a graph $G$ into a surface ${\cal S}$ is a mapping $\phi$
that sends every vertex $v\in V(G)$ into a point $\phi(v)\in {\cal
  S}$ and every edge into a simple curve connecting its endpoints, so
that the images of different edges are allowed to intersect only at
their endpoints. The \emph{Euler genus} of a surface ${\cal S}$, denoted by $\eg({\cal S})$, is defined to be $2-\chi({\cal S})$, where $\chi({\cal S})$ is the Euler characteristic of ${\cal S}$. This parameter coincides with the usual notion of genus, except that it is twice as large if the surface is orientable.
For a graph $G$, the Euler genus of $G$, denoted by $\eg(G)$, is defined to be the minimum Euler genus of a surface ${\cal S}$, such that $G$ can be embedded into ${\cal S}$.

Let $G$ be a graph and let $\phi$ be an embedding of $G$ into some surface ${\cal S}$.
A simple non-contractible loop $\gamma$ in ${\cal S}$ is called a \emph{$\phi$-noose} if it intersects the image of $G$ only on vertices.
We define the \emph{length} of $\gamma$ to be
\[
\len(\gamma) = |\{v\in V(G) : \phi(v) \in \gamma\}|.
\]

\paragraph{Sparsest-cut and the multi-commodity flow-cut gap.}

\begin{definition}[Sparsest Cut]
Consider a graph $G = (V,E)$. The \emph{sparsity} of a cut $(S, \overline{S})$ equals
\[
\phi(S) = \frac{E(S, \overline{S})}{\min\{|S|, |\overline{S}|\}}
\]
where $\overline{S} = V-S$ and $E(S, \overline{S})$ is the number of cut edges, that is, the number of edges from
$S$ to $\overline{S}$. The \emph{Sparsest Cut problem} asks to find a cut $(S, \overline{S})$ with smallest possible sparsity
$\phi(S)$.
\end{definition}

Let $\alpha\in (0,1)$.
For a graph $G$, we say that some $X\subseteq V(G)$ is a \emph{$\alpha$-balanced vertex separator} if every connected component of $G\setminus X$ contains at most $\alpha |V(G)|$ vertices.
We also say that some $F\subseteq E(G)$ is a \emph{$\alpha$-balanced edge separator} if every connected component of $G\setminus F$ contains at most $\alpha |V(G)|$ vertices.
The following is the well-known for approximating the sparsest cut. The \emph{uniform sparsest-cut} means the sparsest-cut problem with uniform weight on every edge.

\begin{theorem}[Arora, Rao and Vazirani \cite{DBLP:journals/jacm/AroraRV09}]\label{thm:ARV}
There exists a polynomial-time $O(\log^{1/2})$-approximation for uniform sparsest-cut.
There is also a polynomial-time algorithm which given some graph that contains a $2/3$-balanced edge separator of size $\ell$, computes a $3/4$-balanced edge separator of size $O(\ell \log^{1/2} n)$. 
\end{theorem}

\begin{theorem}[Feige, Hajiaghayi and Lee \cite{feige2008improved}]\label{thm:vertex-separators-approx}
There exists a polynomial-time which given a graph $G$
 outputs a $3/4$-balanced vertex separator of $G$ of size at most $c_{FHL} s \log^{1/2}$, where $s$ is the size of the minimum $1/2$-balanced vertex separator, for some universal constant $c_{FHL}>0$.
\end{theorem}

We recall that for any graph $G$, we have $\pw(G) = O(\tw(G) \log n)$.
Combined with Theorem \ref{thm:vertex-separators-approx} one obtains the following.

\begin{corollary}[Feige, Hajiaghayi and Lee \cite{feige2008improved}]\label{thm:tw-pw-approx}
There exists a polynomial-time algorithm which given a graph $G$ of treewidth $t$ outputs a tree decomposition of $G$ of width $O(t \log^{1/2} n)$ and a path decomposition of $G$ of width  $O(t \log^{3/2} n)$.
\end{corollary}

Finally, we need the following well-known theorem about the
flow-cut gap for multicommodity flows.

\begin{theorem}[Leighton and Rao \cite{leighton1999multicommodity}]\label{thm:LR}
The multi-commodity flow cut gap for product demands in $n$-vertex graphs is $O(\log n)$.
\end{theorem}


\paragraph{Randomization.}
Some of the algorithms presented in this paper are randomized.
In order to simplify notation, we say that an algorithm succeeds \emph{with high probability} when the failure probability is at most $n^{-c}$.
When the target running time is polynomial, we need $c$ to be some sufficiently large constant.
Similarly, when the target running time is $n^{O(\log n /\log \log n)}$, we need $c=\Omega(\log n /\log\log n)$.
Both guarantees can be achieved by repeating some algorithm that succeeds with constant probability, either at most $O(\log n)$ or $O(\log^2 n / \log \log n)$ times respectively.

\section{The main algorithm}\label{sec:main}

In this Section we present the main algorithm for approximating minimum planarization.
We first introduce some definitions.
Intuitively, a patch is a small irrelevant subgraph, that is contained inside some disk in any optimal solution.
The framing of a patch is a new graph, that does not contain the interior of the patch, and instead contains grid of constant width attached to the boundary of the patch.
Computing a framing in this manner will allow us to extend an approximate solution to the whole graph via an embedding into a higher genus surface.

\begin{definition}[Patch]
Let $G$ be a graph, let $\Gamma\subseteq G$, and let $C\subseteq \Gamma$ be a cycle.
Suppose that there exists a planar drawing of $\Gamma$ having $C$ as the outer face.
Then we say that the ordered pair $(\Gamma,C)$ is a \emph{patch} (of $G$).
\end{definition}

\begin{definition}[Framing]
Let $G$ be a graph and let $(\Gamma,C)$ be a patch of $G$.
Suppose that $V(C)=\{v_0^0,\ldots,v_{t-1}^0\}$, and $E(C)=\{\{v_0^0,v_1^0\},\ldots,\{v_{t-2}^0,v_{t-1}^0\}, \{v_{t-1}^0,v_0^0\}\}$.
Let $G^\framed$ be the graph with
\[
V(G^\framed) = (V(G) \setminus (V(\Gamma)\setminus V(G))) \cup \bigcup_{i=1}^{t-1} \bigcup_{j=1}^{3}\{v_i^j\},
\]
where $v_0,\ldots,v_{t-1}'$ are new vertices,
and
\[
E(G^\framed) = (E(G) \setminus (E(\Gamma)\setminus E(C))) \cup \left(\bigcup_{i=0}^{t-1} \bigcup_{j=0}^2\{\{v_i^j,v_i^{j+1}\}, \{v_i^j,v_{i+1 \mymod t}^j\}\} \right).
\]
We refer to the graph $G^\framed$ as the \emph{$(\Gamma,C)$-framing of $G$} (see Figure \ref{fig:framing} for an example).
\end{definition}

\begin{figure}
\begin{center}
\scalebox{1.1}{\includegraphics{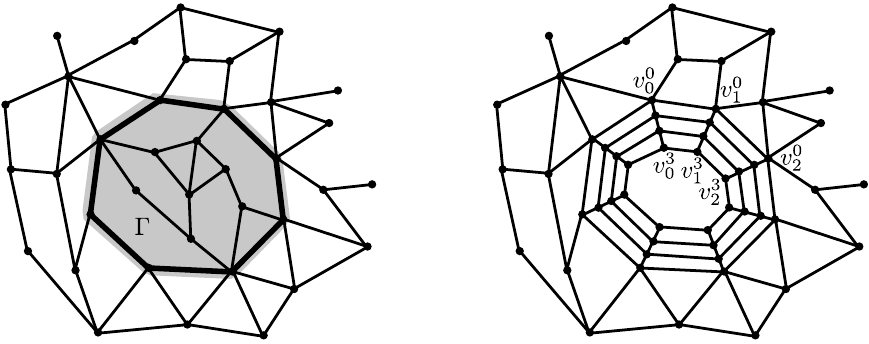}}
\caption{A patch $(\Gamma, C)$ in a graph $G$, where the shaded area contains the graph $\Gamma$, and the cycle $C$ is drawn in bold (left), and the $(\Gamma,C)$-framing of $G$ (right).\label{fig:framing}}
\end{center}
\end{figure}

Using the above definitions, we can now define the concept of a pruning sequence.
Intuitively, this consists of a sequence of operations that inductively simplify the graph until it becomes planar.

\begin{definition}[Pruning sequence]
Let $G$ be a graph.
Let ${\cal G}=(G_0,A_0),\ldots,(G_\ell,A_\ell)$ be a sequence satisfying the following properties:
\begin{description}
\item{(1)}
For all $i\in \{0,\ldots,\ell\}$, $G_i$ is a graph.
Moreover $G_0=G$ and $G_\ell$ is planar.

\item{(2)}
For all $i\in \{1,\ldots,\ell\}$, exactly one the following holds:
\begin{description}
\item{(2.1)}
$A_{i-1} \subset V(G_{i-1})$ and $G_i=G_{i-1} \setminus A_{i-1}$.
We say that $i-1$ is a \emph{deletion step} (of ${\cal G}$).
\item{(2.2)}
 $A_{i-1}=(\Gamma_{i-1}, C_{i-1})$, where $(\Gamma_{i-1}, C_{i-1})$ is a patch in $G_{i-1}$, and $G_i$ is the $(\Gamma_{i-1}, C_{i-1})$-framing of $G_{i-1}$.
We say that $i-1$ is a \emph{framing step} (of ${\cal G}$).
\end{description}
We also let $A_\ell=\emptyset$.
\end{description}
We say that ${\cal G}$ is a \emph{pruning sequence} (for $G$).
We also define the \emph{cost} of ${\cal G}$ to be
\[
\cost({\cal G}) = \sum_{i\in \{0,\ldots,\ell-1\} : i \text{ is a deletion step}} |A_{i}|
\]
\end{definition}

The next Lemma shows how to compute a pruning sequence of low cost.
We remark that the algorithm for computing a pruning sequence calls recursively the whole approximation algorithm on graphs of smaller size.
By controlling the size of these subgraphs, we obtain a trade-off between the running time and the approximation guarantee.
We remark that this is the main technical contribution of this paper.
The proof of the next Lemma uses several other results, and spans the majority of the rest of the paper.

\begin{lemma}[Computing a pruning sequence]\label{lem:pruning}
Let $G$ be an $n$-vertex graph, and let $\rho>0$.
Suppose that there exists an algorithm $\AlgApprox$ which for all $n'\in \mathbb{N}$, with $n'<n$, given an $n'$-vertex graph $G'$, outputs some $S'\subset V(G')$, such that $G'\setminus S'$ is planar, with $|S'|\leq \approxfactor \cdot \mvp(G')$,
for some $\approxfactor \geq 2\rho$,
in time $T_\myapprox(n')$, where $T_\myapprox:\mathbb{N}\to\mathbb{N}$ is increasing and convex.
Then the algorithm $\AlgPruning(G)$ returns some pruning sequence for $G$ of cost at most $O(\mvp(G) \sqrt{\approxfactor} \log^{15} n)$.
Moreover
\[
T_\pruning(n) \leq n^{O(1)} + T_\pruning(n,k) \log n,
\]
and
\[
T_\pruning(n,k) \leq n^{O(1)} + \max\{T_\pruning(n/4)+T_\pruning(3n/4), T_\myapprox(n/\rho)2\rho \log n + T_\pruning(n-1,k)\},
\]
where $T_\pruning(n)$ and $T_\pruning(n,k)$ denote the worst-case running time of $\AlgPruning(G)$ and $\AlgPruning(G,k)$ respectively on a graph of $n$ vertices.
\end{lemma}

The next Lemma shows that given a pruning sequence of low cost, we can efficiently compute an embedding into a surface of low Euler genus, after deleting a small number of vertices.

\begin{lemma}[Embedding into a higher genus surface]\label{lem:embedding}
Let $G$ be a graph and let ${\cal G}$ be a pruning sequence for $G$.
Then there exists a polynomial-time algorithm which given $G$ and ${\cal G}$ outputs some $X\subseteq V(G)$, with $|X|\leq \cost({\cal G})$, and an embedding of $G\setminus X$ into some surface of Euler genus $O(\cost({\cal G}))$.
\end{lemma}

The next Lemma gives an efficient algorithm for planarizing a graph embedded into some surface of low Euler genus.

\begin{lemma}[Planarizing a surface-embedded graph]\label{lem:planarize_surface}
Let $G$ be a graph, and let $\phi$ be an embedding of $G$ into some surface of Euler genus $g>0$.
Then there exists a polynomial-time algorithm which given $G$ and $\phi$, outputs some planarizing set $X$ for $G$, with $|X| = O(g \log n + \mvp(G) \log^2 n)$.
\end{lemma}

We are now ready to present our main results in this paper, which combines the above results.

\begin{theorem}\label{thm:main}
Let $\rho>0$, and $\approxfactor =\max\{2\rho, O(\log^{32} n)\}$.
Then there exists a $\approxfactor$-approximation algorithm for the minimum vertex planarization problem with running time
\[
T_{\myapprox}(n)\leq n^{O(1)} + \max\{T_{\myapprox}(n/4)+T_{\myapprox}(3n/4), T_{\myapprox}(n/\rho)2\rho\log n + T_{\myapprox}(n-1)\}.
\]
\end{theorem}

\begin{proof}
Let $\rho>0$ be the parameter from Lemma \ref{lem:pruning}, with $\approxfactor\geq 2\rho$.
Using the algorithm from Lemma \ref{lem:pruning} we compute a pruning sequence ${\cal F}$ of $G$ with $\cost({\cal G}) = O(\mvp(G) \sqrt{\approxfactor} \log^{15} n)$.
Using the algorithm from Lemma \ref{lem:embedding}, in polynomial time, we compute some $X_1\subset V(G)$, with $|X_1|\leq \cost({\cal G}) \leq O(\mvp(G) \sqrt{\approxfactor} \log^{15} n)$, and an embedding $\phi$ of $G\setminus X_1$ into some surface of genus $g=O(\mvp(G) \sqrt{\approxfactor} \log^{15} n)$.
Using the algorithm from Lemma \ref{lem:planarize_surface} we compute some $X_2 \subset V(G)\setminus X_1$, such that $G\setminus (X_1 \cup X_2)$ is planar, with $|X_2| = O(g \log n + \mvp(G\setminus X_1) \log^2 n) = O(g \log n + \mvp(G) \log^2 n) = O(\mvp(G) \sqrt{\approxfactor} \log^{16} n)$.
It follows that $X=X_1\cup X_2$ is a planarizing set of $G$, with
$|X| = |X_1| + |X_2| = O(\mvp(G) \sqrt{\approxfactor} \log^{16} n)$.
Thus the resulting approximation factor is at most
$\approxfactor = \max\{2\rho, O(\log^{32} n)\}$.

The total running time is dominated by the running time of the algorithm from Lemma \ref{lem:pruning}.
Thus we have
$T_{\myapprox}(n)\leq n^{O(1)} + \max\{T_\pruning(n/4)+T_\pruning(3n/4), T_\myapprox(n/\rho)2\rho\log n + T_\pruning(n-1,n-1)\} \leq n^{O(1)} + \max\{T_\pruning(n/4)+T_\pruning(3n/4), T_\myapprox(n/\rho)2\rho\log n + T_\pruning(n-1)\} \leq n^{O(1)} + \max\{T_\myapprox(n/4)+T_\myapprox(3n/4), T_\myapprox(n/\rho)2\rho\log n + T_\myapprox(n-1)\}$,
which concludes the proof.
\end{proof}

By Theorem \ref{thm:main}, we easily obtain the following results.
 

\begin{proof}[Proof of Theorem \ref{thm:polylog}]
It follows from Theorem \ref{thm:main} by setting $\rho=\log n$.
\end{proof}


\begin{proof}[Proof of Theorem \ref{thm:eps}]
It follows from Theorem \ref{thm:main} by setting $\rho= n^\eps$.
\end{proof}

\section{Computing a doubly cut-linked set}\label{sec:doubly}

In this Section we show that when the treewidth is sufficiently large, we can efficiently compute a separation $(U,U')$, and some large $Y\subseteq U\cap U'$, such that $Y$ is well-linked on both sides of the separation.
This result will form the basis for our algorithm for computing a grid minor in the subsequent Sections.

Let $G$ be a $n$-vertex $k$-apex graph of treewidth $t > 2 k$.
In this section, we shall, in polynomial time, compute some separation $(U,U')$ of $G$ of order at most $O(t \log^{3/2} n)$ and some $X\subseteq U\cap U'$ with $|X|= \Omega(t/\log^{9/2} n)$ that is $\Omega(1)$-cut-linked in both $G[U]$ and $G[U']$.
Moreover let $G_U$ be the graph obtained from $G[U]$ by removing all edges in $G[U\cap U']$, that is $G_U=G[U]\setminus E(G[U\cap U'])$. Then $X$ is also $\Omega(1)$-cut-linked in $G_U$.

To this end, we begin with the following lemma, which merges two disjoint cut-linked sets.

\begin{lemma}[Merging two cut-linked sets]\label{lem:merging_cut-linked}
Let $\alpha>0$ and $\eps \in (0,1/2]$.
Let $G$ be a graph and let $W_1, W_2, R \subseteq V(G)$ that are all $\alpha$-cut-linked in $G$, with $W_1\cap W_2=\emptyset$.
Suppose further that for any $i\in \{1,2\}$ there exist a collection ${\cal P}_i$ of pairwise edge-disjoint paths in $G$ between $W_i$ and $R$, with $|{\cal P}_i|\geq \eps\cdot |W_i|$, such that all paths in ${\cal P}_i$ have distinct endpoints.
Then $W_1\cup W_2$ is $\eps \alpha /8$-cut-linked in $G$.
\end{lemma}

\begin{proof}
Let $U\subseteq V(G)$ and $\bar{U}=V(G)\setminus U$.
For any $i\in \{1,2\}$ let $U_i=W_i\cap U$, $\bar{U}_i=W_i\setminus U$, and let $U_R = R\cap U$, $\bar{U}_R=R\setminus U$.
Let also
$M=\min\{|U_1|+|U_2|, |\bar{U}_1|+|\bar{U}_2|\}$.
Let $\delta=\eps/2$ and $\gamma=\eps/8$.
We distinguish between the following cases:
\begin{description}
\item{Case 1: Suppose that $\delta M\leq |U_1| \leq (1-\delta)M$ or $\delta M\leq |\bar{U}_1| \leq (1-\delta)M$.}
Then
\begin{align*}
|E(U, \bar{U})| &\geq |E(U_1,\bar{U}_1)| \\
 &\geq \delta  \alpha  M.
\end{align*}

\item{Case 2: Suppose that $\delta M\leq |U_2| \leq (1-\delta)M$ or $\delta M\leq |\bar{U}_2| \leq (1-\delta)M$.}
This is similar to Case 1.

\item{Case 3: Suppose that $|U_1|,|\bar{U}_1|,|U_2|,|\bar{U}_2| \notin [\delta M, (1-\delta)M]$.}

\begin{description}
\item{Case 3.1: Suppose that $\gamma M\leq |U_R| \leq (1-\gamma) M$.}
Then
\begin{align*}
|E(U,\bar{U})| &\geq |E(U_R, \bar{U}_R)| \\
 &\geq \gamma \alpha M.
\end{align*}

\item{Case 3.2: Suppose that $|U_R| < \gamma M$.}
We may assume w.l.o.g.~that $|U_1|+|U_2|\leq |\bar{U}_1|+|\bar{U}_2|$ and thus $M=|U_1|+|U_2|$; the case $M=|\bar{U}_1|+|\bar{U}_2|$ is identical by swapping $U$ and $\bar{U}$.
We may also assume w.l.o.g.~that $|U_1|\leq |U_2|$ because the case $|U_2|\leq |U_1|$ is identical by swapping $W_1$ and $W_2$.
Since $|U_1|+|U_2|=M$, it follows that $|U_1|< \delta M$ and $|U_2|>(1-\delta)M$.

\begin{description}
\item{Case 3.2.1: Suppose that $|\bar{U}_1|<\delta M$.}
Since $|\bar{U}_1|+|\bar{U}_2|\geq M$, we get $|\bar{U}_2|>(1-\delta)M$.
Thus
\begin{align*}
|E(U,\bar{U})| &\geq |E(U_2, \bar{U}_2)| \\
 &\geq \alpha \cdot \min\{|U_2|, |\bar{U}_2|\} \\
 &\geq \alpha (1-\delta)  M \\
 &\geq (\alpha/2) M.
\end{align*}

\item{Case 3.2.2: Suppose that $|\bar{U}_1|>(1-\delta)M$.}
Then $|{\cal P}_1| \geq \eps  |W_1| \geq \eps |\bar{U}_1| > \eps(1-\delta) M$.
There exist at most $|U_1|$ paths in ${\cal P}_1$ with both endpoints in $U$.
Thus there exist at least $|{\cal P}_1|-|U_1|$ paths in ${\cal P}_1$ with at least one endpoint in $\bar{U}$.
Moreover there exist at most $|\bar{U}_R|$ paths in ${\cal P}_1$ with both endpoints in $\bar{U}$.
Therefore there exist at least $|{\cal P}_1|-|U_1|-|\bar{U}_R|$ paths in ${\cal P}_1$ with one endpoint in $U$ and one endpoint in $\bar{U}$.
Each such path must contain an edge in $E(U,\bar{U})$.
Therefore we have
\begin{align*}
|E(U,\bar{U}) &\geq |{\cal P}_1|-|U_1|-|\bar{U}_R| \\
 &> \eps(1-\delta) M - \delta M - \gamma M \\
 &\geq (\eps-\eps^2/2-\eps/2-\eps/8) M\\
 &\geq (\eps/8) M.
\end{align*}
\end{description}
\item{Case 3.3: Suppose that $|U_R| > (1-\gamma) M$.}
This is similar to Case 3.2.
\end{description}
\end{description}
We conclude that in either case we have $|E(U,\bar{U})|\geq (\eps \alpha/8) M$, as required.
\end{proof}

The next two lemmas are concerning results for the $(r\times r)$-grid. These two lemmas
are needed in the rest of this section.

\begin{lemma}\label{lem:grid-row-cut-linked}
Let $r\geq 1$ and let $H$ be the $(r\times r)$-grid.
Let $X$ be the set of vertices in the first row of $H$.
Then $X$ is $1$-cut-linked in $H$.
\end{lemma}

\begin{proof}
Let $U\subseteq V(H)$ and let $\bar{U}=V(H)\setminus U$.
Let $A=X\cap U$ and $B=X\cap \bar{U}$.
Assume w.l.o.g.~that $|A|\leq |B|$.
Suppose that for every $v\in A$ the column of $H$ containing $v$ intersects both $U$ and $\bar{U}$.
Then we get $|E(U,\bar{U})|\geq |A|$.
Similarly, if for every $v\in B$ the column of $H$ containing $v$ intersects both $U$ and $\bar{U}$ then $|E(U,\bar{U})|\geq |A|\geq |B|$.
It remains to consider the case where there exist $v\in A$ and $v'\in B$ such that the column $C$ (resp.~$C'$) of $H$ containing $v$ (resp.~$v'$) does not intersect both $U$ and $\bar{U}$.
Since $v\in U$ and $v'\in \bar{U}$ we get $C\subseteq U$ and $C'\subseteq U'$.
Since every row in $H$ intersects both $U$ and $\bar{U}$ (because there exist $v\in A$ and $v'\in B$ such that the column $C$ (resp.~$C'$) of $H$ containing $v$ (resp.~$v'$) does not intersect both $U$ and $\bar{U}$ ),
it follows that every row in $H$ contains an edge in $E(U,\bar{U})$, and thus $|E(U,\bar{U})|\geq r > |A|$.
In either case we have obtained that $|E(U,\bar{U})|\geq |A|$, concluding the proof.
\end{proof}

\begin{lemma}[Linking two grid minors in a grid]\label{lem:grid_isoperimetry_boundaries}
Let $\Gamma$ be a grid and let $r_1\leq r_2$.
For each $i\in \{1,2\}$ let $Z_i$ be some $(r_i\times r_i)$-grid and suppose that $Z_i$ is a minor of $\Gamma$ with minor mapping $\mu_i$.
Assume that $\mu_1(V(Z_1)) \cap \mu_2(V(Z_2)) = \emptyset$.
For each $i\in \{1,2\}$ let $Y_i$ be the set of vertices in the first row of $Z_i$,
and
for each $v\in Y_i$ pick some $q(v)\in \mu_i(v)$.
For each $i\in \{1,2\}$ let $W_i=\bigcup_{v\in Z_i}\{q(v)\}$.
Then $W_1\cup W_2$ is $1/384$-cut-linked in $\Gamma$.
\end{lemma}

\begin{proof}
Let $i\in \{1,2\}$.
Let $(U,\bar{U})$ be a cut in $\Gamma$.
Define a cut $(U', \bar{U}')$ in $\Gamma$ as follows:
Let $\{x,y\}\in E(\Gamma)$ be an edge crossing $(U,\bar{U})$, and let $x'\in V(Z_i)$ with $x\in \mu_i(x')$.
We cut all the edges in $E(Z_i)$ that are incident to $x'$.
Since $Z_i$ has maximum degree 4, it follows that
$|E(U', \bar{U}')| \leq 4\cdot |E(U,\bar{U})|$.
Since $Y_i$ is 1-cut-linked in $Z_i$, we have
$|E(U', \bar{U}')| \geq \min\{|U'\cap Y_i|, |\bar{U}'\cap Y_i|\}$.
By construction, if a pair of vertices $a,b\in W_i$ is separated by the cut $(U,\bar{U})$ in $\Gamma$, then the pair of vertices $q^{-1}(a)$, $q^{-1}(b)$ is also separated by the cut $(U',\bar{U}')$ in $Z_i$.
Thus
$\min\{|U'\cap Y_i|, |\bar{U}'\cap Y_i|\} \geq \min\{|U\cap W_i|, |\bar{U}\cap W_i|\}$.
Putting everything together we get
$|E(U,\bar{U})|\geq (1/4)\min\{|U\cap W_i|, |\bar{U}\cap W_i|\}$, and thus $W_i$ is $1/4$-cut-linked in $\Gamma$.

Let $R$ be the set of vertices in the right-most column in $\Gamma$.
By Lemma \ref{lem:grid-row-cut-linked} we have that $R$ is 1-cut-linked in $\Gamma$.
We first argue that there exists some collection ${\cal P}_i$ of pairwise edge-disjoint paths between $W_i$ and $R$ such that all paths in ${\cal P}_i$ have distinct endpoints.
We consider the following cases:

\begin{description}
\item{Case 1:} There exists a column $C$ of $\Gamma$ such that $|C\cap W_i| \geq r_i/3$.
Let $W_i'=C\cap W_i$.
For each $v\in W_i'$ we add to ${\cal P}_i$ the horizontal path between $v$ and the right-most vertex in $V(H)$ that lies in the same row as $v$.
It is immediate that the paths in ${\cal P}_i$ are pairwise edge-disjoint and that $|{\cal P}_i| \geq r_i/3$.

\item{Case 2:} Any column of $\Gamma$ contains less than $r_i/3$ vertices in $W_i$.
There exists a column $C'$ such that each connected component of $\Gamma \setminus C'$ contains at least $r_i/3$ vertices in $W_i$.
Let $S$ (resp.~$T$) be the set of vertices in $W_i$ that are in the left (resp.~right) of $C'$.
Let $\Gamma'$ be the graph obtained from $\Gamma$ as follows.
We add new vertices $s$ and $t$ and we connect $s$ (resp.~$t$) to every vertex in $S$ (resp.~$T$) via edges of capacity $|S|+|T|$.
All other edges in $\Gamma'$ have unit capacity.
Since $W_i$ is $(1/4)$-cut-linked it follows that the value of the minimum $s$-$t$ cut in $\Gamma'$ is at least $\min\{|S|,|T|\}/4\geq r_i/12$.
Let $F_i$ be a maximum $s$-$t$ flow between $s$ and $t$ in $\Gamma'$ of value $f_i\geq r_i/12$.
Since all capacities in $\Gamma'$ are integers, it follows that $F_i$ can be decomposed into a collection ${\cal Q}_i$ of pairwise edge-disjoint $s$-$t$ paths in $\Gamma'$.
For each path $Q\in {\cal Q}_i$ let $Q'$ be the subpath of $Q$ between the first vertex that $Q$ visits in $S$ and the first vertex it visits in $C'$.
It follows that the last edge traversed by $Q'$ is some horizontal edge $\{z,w\}$ where $w\in V(C')$ and $z$ is to the left of $w$.
Let $Q''$ be the path obtained by concatenating $Q'$ with the horizontal path in $H$ between $w$ and the right-most column in $H$.
Let ${\cal P}_i$ be the collection of all the paths $Q''$ obtained in this way.
It is immediate that the paths in ${\cal P}_i$ are pairwise edge-disjoint, and $|{\cal P}_i|=|{\cal Q}_i| \geq r_i/12$.
\end{description}

In either case we have obtained for each $i\in \{1,2\}$ a collection ${\cal P}_i$ of pairwise edge-disjoint paths between $W_i$ and $C$, such that all paths in ${\cal P}_i$ have distinct endpoints, and with $|{\cal P}_i|\geq r_i/12$.
Since $W_1$ and $W_2$ are both $1/4$-cut-linked in $\Gamma$ and $R$ is $1$-cut-linked in $\Gamma$, we can apply Lemma \ref{lem:merging_cut-linked} with $\alpha=1/4$ and $\eps=1/12$.
We obtain that $W_1\cup W_2$ is $1/384$-cut-linked in $\Gamma$.
\end{proof}

The following is immediate by the definition of a cut-linked set.

\begin{proposition}\label{prop:cut-linked_degree}
Let $G$ be a graph and let $H$ be a minor of $G$ of maximum degree $\Delta$ with minor mapping $\mu:V(H)\to 2^{V(G)}$.
Let $X\subseteq V(H)$ such that $X$ is $\alpha$-cut-linked in $H$.
For each $x\in X$ pick some $q(x)\in \mu(x)$.
Then $\bigcup_{x\in X} \{q(x)\}$ is $\alpha/\Delta$-cut-lined in $G$.
\end{proposition}

Before presenting the main result in this section we need the following auxiliary result.
\begin{lemma}\label{lem:linking_paths}
Let $G$ be a graph and let $X,Y\subset V(G)$.
Suppose that $X$ is $\alpha$-cut-linked in $G$ for some $\alpha> 0$.
Suppose that there exists a collection ${\cal Q}$ of pairwise vertex-disjoint paths in $G$ between $X$ and $Y$ with $|{\cal Q}|=|Y|$.
Then $Y$ is $\min\{1/2,\alpha/2\}$-cut-linked in $G$.
\end{lemma}
\begin{proof}
Let $U\subseteq V(G)$ and $\bar{U}=V(G)\setminus U$.
Let ${\cal Q}'$ be the set of all paths in ${\cal Q}$ that contain at least one edge crossing the cut $(U,\bar{U})$.
Then $|U\cap X|\geq |U\cap Y|-|{\cal Q}'|$ and
$|\bar{U}\cap X|\geq |\bar{U}\cap Y|-|{\cal Q}'|$.
Thus $\min\{|U\cap X|, |\bar{U}\cap X|\} \geq \min\{|U\cap Y|, |\bar{U}\cap Y|\}-|{\cal Q}'|$.
It follows that
\begin{align*}
|E(U,\bar{U})| &\geq \max\{|Q'|, \alpha \cdot (\min\{|U\cap Y|, |\bar{U}\cap Y|\}-|Q'|)\} \\
 &\geq \min\{1/2, \alpha/2\} \cdot \min\{|U\cap Y|, |\bar{U}\cap Y|\}-|Q'|)\},
\end{align*}
as required.
\end{proof}

We are now ready to present the main result in this section.
The following lemma is the main ``existential'' result with respect to
computing a doubly-cut-linked set. Actually, as stated below, we first compute, in polynomial time,
some separation $(U,U')$ of $G$ such that there exists some $X\subseteq U\cap U'$
that is $\Omega(1)$-cut-linked in both $G[U]$ and $G[U']$ (in fact, this set $X$ has even
a stronger property as stated below). In the next lemma,
we shall show how to find such a vertex set $X$ in polynomial time.

\begin{lemma}[Existence of a large doubly-cut-linked set]\label{lem:doubly-cut-linked-existence}
Let $G$ be a $n$-vertex $k$-apex graph of treewidth $t > 2 k$.
Then there exists a polynomial-time algorithm which given $G$ outputs some separation $(U,U')$ of $G$ of order at most $O(t \log^{3/2} n)$ such that there exists some $X\subseteq U\cap U'$ with $|X|= \Omega(t/\log^{9/2} n)$ that is $\Omega(1)$-cut-linked in both $G[U]$ and $G[U']$.
Moreover let $G_U$ be the graph obtained from $G[U]$ by removing all edges in $G[U\cap U']$, that is $G_U=G[U]\setminus E(G[U\cap U'])$.
Then $X$ is also $\Omega(1)$-cut-linked in $G_U$.
\end{lemma}

\begin{proof}
By Lemma \ref{lem:grid_existence} there exists some $r\geq c_1 (t-k)$, for some universal constant $c_1>0$, and some minor mapping $\mu:V(\Gamma) \to 2^{V(G)}$ where $\Gamma$ is the $(r\times r)$ grid.

Let $({\cal P}, {\cal X})$ be a path decomposition of $G$ of width $w\leq c_2 t \log^{3/2} n$ computed by the algorithm in Theorem \ref{thm:tw-pw-approx}, for some universal constant $c_2>0$.
Let ${\cal X}=\{B_1,\ldots,B_{\ell}\}$ where $E({\cal P})=\{\{B_1,B_2\},\ldots,\{B_{\ell-1},B_{\ell}\}\}$.

For any $i\in\{1,\ldots,\ell\}$ let
\[
G_i^- = G\setminus \bigcup_{j=i}^{\ell} B_j
\]
and
\[
G_i^+ = G\setminus \bigcup_{j=1}^{i} B_j.
\]
Let also $\Gamma_i^-$ be the graph obtained as follows:
We start with $\Gamma$ and we delete all $v\in V(\Gamma)$ such that $\mu(v)$ intersects $B_{i}\cup \ldots \cup B_{\ell}$.
Formally we define
\[
\Gamma_i^- = \Gamma\left[\{v\in V(\Gamma) : \mu(v)\subseteq V(G_i^-)\}\right].
\]
Similarly we define
\[
\Gamma_i^+ = \Gamma\left[\{v\in V(\Gamma) : \mu(v)\subseteq V(G_i^+)\}\right].
\]
Observe that for any $i\in \{1,\ldots,\ell\}$ we have that $\Gamma_i^-$ and $\Gamma_i^+$ are vertex-disjoint subgraphs of $\Gamma$.
Moreover $\Gamma_i^-$ is a minor of $G_i^-$, and $\Gamma_i^+$ is a minor of $G_i^+$.

Let $i^*\in \{1,\ldots,\ell\}$ be maximum such that $\Gamma_i^+$ contains some $(r'\times r')$-grid minor where $r'=\gamma \cdot r$, for some $\gamma>0$ to be determined later.
Since $G_\ell^+$ is the empty graph it follows that $i^*<\ell$.
Since $G_{i^*}^+$ and $G_{i^*+1}^+$ differ by at most $w$ vertices it follows that $\Gamma_{i^*+1}^+$ is obtained from $\Gamma_{i^*}^+$ by removing at most $w$ vertices.
By Lemma \ref{lem:grids_persistency} it follows that $\Gamma_{i^*+1}^+$ contains some $(r''\times r'')$-grid $Z$ as a minor for some $r''=\Omega((r')^2/w)=\Omega(\gamma^2 r^2/(t \log^{3/2} n))=\Omega(\gamma^2 (t-k)^2/(t \log^{3/2} n))=\Omega(\gamma^2 t/\log^{3/2} n)$.
It follows by composition of minor mappings that $Z$ is also a minor of $G_{i^*+1}^+$.

By Lemma \ref{lem:grids_persistency} we have that for any $i\in \{1,\ldots,\ell\}$, there exists a universal constant $c_3>0$ and a $(r'''\times r''')$-grid $J_i$ that is a minor of $\Gamma_i^+\cup \Gamma_i^-$ for some $r'''\geq c_3 r^2/w\geq \frac{c_3 c_1^2}{c_2} (t-k)^2/(t\log^{3/2}n)\geq \frac{c_3 c_1^2}{4 c_2} t/\log^{3/2}n$.
By setting $\gamma = \frac{c_1 c_3}{ 4 c_2 \log^{3/2} n}$ we get that $r'''\geq r'$.
By the choice of $i^*$ it follows that $G_{i^*+1}^+$ does not contain a $(r'\times r')$-grid minor.
Therefore $J_{i^*+1}$ must be a minor of $\Gamma_{i^*+1}^-$ and thus also a minor of $G_{i^*+1}^-$.

We have that $Z$ and $J_{i^*+1}$ are vertex-disjoint minors of $\Gamma$.
Let $\mu_1$ be the minor mapping for $Z$ and let $\mu_2$ be the minor mapping for $J_{i^*+1}$.
Let $Y_1$ and $Y_2$ be the sets of vertices in the first row of $Z$ and $J_{i^*+1}$ respectively.
For each $y\in Y_1$ pick some $q(y)\in \mu_1(y)\subseteq V(G)$ and for each $y\in Y_2$ pick some $q(y)\in \mu_2(y)\subseteq V(G)$.
Let $Q_1 = \bigcup_{y\in Y_1} \{q(y)\}$ and $Q_2 = \bigcup_{y\in Y_2} \{q(y)\}$.
By Lemma \ref{lem:grid_isoperimetry_boundaries} we have that $Q_1\cup Q_2$ is $1/384$-cut-linked in $\Gamma$.
Since $\Gamma$ has maximum degree 4 it follows that there exists a collection ${\cal Q}$ of at least $\min\{|Q_1|, |Q_2|\}/(384 \cdot 4) = \min\{r',r''\}/1536 = \Omega(t/\log^{9/2} n)$ vertex-disjoint paths between $Q_1$ and $Q_2$ in $\Gamma$.
For each $x\in Q_1\cup Q_2$ pick some $q(x)\in \mu(x)$.
For each $i\in \{1,2\}$ let $Q_i'=\bigcup_{x\in Q_i} \{q(x)\}$.
Then there exists a collection ${\cal Q}'$ of at least $|{\cal Q}|$ vertex-disjoint paths between $Q_1'$ and $Q_2'$ in $G$.
Each path $P\in {\cal Q}'$ must intersect some vertex in $B_{i^*+1}$.
Let $v_P$ be the first vertex in $B_{i^*+1}$ that is visited when traversing $P$ starting from $Q_1'$.
Let $X=\{v_P\}_{P\in {\cal Q}'}$,
$U=G[B_1\cup \ldots\cup B_{i^*+1}]$, and
$U'=G[B_{i^*+1}\cup \ldots\cup B_{\ell}]$.
We have $|X|=|{\cal Q}'|=|{\cal Q}| = \Omega(t/\log^{9/2} n)$.
By Proposition \ref{prop:cut-linked_degree} we have that $Q'_1$ is $1/1536$-cut-linked in $G[U]$ and $G_U$, and $Q'_2$ is $1/1536$-cut-linked in $G[U']$.
It follows by Lemma \ref{lem:linking_paths} that $X$ is $1/3072$-cut-linked in $G[U]$, $G_U$, and $G[U']$, which concludes the proof.
\end{proof}

Finally, we show how to compute a vertex set $X$ as in Lemma \ref{lem:doubly-cut-linked-existence} in polynomial time, which completes the proof of this section.

\begin{lemma}[Computing a doubly-cut-linked set]\label{lem:doubly-linked_compute}
Let $G$ be a $n$-vertex $k$-apex graph of treewidth $t > 2 k$.
Then there exists a polynomial-time algorithm which computes some separation $(U,U')$ of $G$ of order at most $O(t \log^{3/2} n)$ and some $Y\subseteq U\cap U'$ with $|Y| = \Omega(t/\log^{9/2} n)$ such that $Y$ is $\Omega\left(\frac{1}{\log^{13/2} n}\right)$-cut-linked in both $G[U]$ and $G[U']$.
Moreover let $G_U$ be the graph obtained from $G[U]$ by removing all edges in $G[U\cap U']$, that is $G_U=G[U]\setminus E(G[U\cap U'])$.
Then $Y$ is also $\Omega\left(\frac{1}{\log^{13/2} n}\right)$-cut-linked in $G_U$.
\end{lemma}

\begin{proof}
Let $(U,U')$ be the separation of order at most $c_1 t \log^{3/2} n$ computed by the algorithm in Lemma \ref{lem:doubly-cut-linked-existence}, for some universal constant $c_1>0$.
Lemma \ref{lem:doubly-cut-linked-existence} guarantees that there exists some $X\subseteq U\cap U'$ with $|X| \geq c_2 t \log^{9/2} n$ such that $X$ is $c_3$-well-lined in both $G_U$ and $G[U']$, for some universal constants $c_2,c_3>0$


We compute a sequence ${\cal P}=\{P_i\}_{i=0}^{\ell}$ of partitions of $U\cap U'$.
We set $P_0=\{U\cap U'\}$, i.e.~with a single cluster containing all vertices in $U\cap U'$.
Given $P_i$ for some $i\geq 0$, we proceed as follows.
For any $C\in P_i$ and for any $S\subseteq C$ let
$\mu_C(S)$ (resp.~$\mu_C'(S)$) denote the minimum number of edges in a cut in $G_U$ (resp.~$G[U']$) that separates $S$ from $C\setminus S$.
Let $\alpha>0$ be a parameter to be determined later.
By Theorem \ref{thm:ARV} there exists some universal constant $c_4>0$ such that in polynomial time we can decide whether one of the following two conditions holds:
\begin{description}
\item{(1)}
There exists $C\in P_i$ and some non-empty $S\subsetneq C$ such that either
\[
\mu_C(S) \leq \alpha \cdot \min\{|S|, |C\setminus S|\}
\]
or
\[
\mu'_C(S) \leq \alpha \cdot \min\{|S|, |C\setminus S|\}
\]

\item{(2)}
For all $C\in P_i$ and for all non-empty $S\subsetneq C$ we have
\[
\mu_C(S) \geq \min\{|S|, |C\setminus S|\} \cdot \frac{c_4 \alpha}{\log^{1/2}n}
\]
and
\[
\mu'_C(S) \geq \min\{|S|, |C\setminus S|\} \cdot \frac{c_4 \alpha}{\log^{1/2}n}.
\]
\end{description}
If condition (1) holds for some $C\in P_i$ and some $S\subsetneq C$ then we obtain $P_{i+1}$ as a refinement of $P_i$ by partitioning $C$ into $S$ and $C\setminus S$.
Formally, we set $P_{i+1} = (P_i \setminus \{C\})\cup \{S,C\setminus S\}$.
Otherwise, if condition (2) holds then all clusters in $P_i$ are $\frac{c_4 \alpha}{\log^{1/2} n}$-cut-linked in both $G_U$ and $G[U']$;
we terminate the sequence ${\cal P}$ at $P_i$, and we output the largest cluster in ${\cal P}$.
This concludes the description of the algorithm.

It is immediate that the length of ${\cal P}$ is at most $n-1$ and thus the algorithm runs in polynomial time.
It remains to show that the cluster that the algorithm outputs is sufficiently large.
To that end we define $X_0\supseteq X_1 \supseteq \ldots \supseteq X_{\ell}$ where $X_0=X$ and for each $i\in \{1,\ldots,\ell\}$ we set $X_i$ to be the largest subset of $X_{i-1}$ that is contained in some cluster in $P_i$, breaking ties arbitrarily.
For any $i\in \{0,\ldots,\ell\}$, let $C_i$ be the unique cluster in $P_i$ that contains $X_i$; clearly, we have $C_0=U\cap U'$ and $C_0\supseteq C_1\supseteq\ldots\supseteq C_{\ell}$.

Consider some $i\in \{0,\ldots,\ell-1\}$.
If $P_{i+1}$ is obtained from $P_i$ by partitioning some cluster other than $C_i$, then we have $X_{i+1}=X_i$ and $C_{i+1}=C_i$.
Otherwise, suppose that $P_{i+1}$ is obtained by partitioning $C_i$ into $C_{i+1}$ and $C_i\setminus C_{i+1}$.
By construction we have that either
\begin{align}
\mu_{C_i}(C_{i+1}) &\leq \alpha \cdot \min\{|C_{i+1}|, |C_i| - |C_{i+1}|\} \label{eq:srink_mu}
\end{align}
or
\begin{align}
\mu'_{C_i}(C_{i+1}) &\leq \alpha \cdot \min\{|C_{i+1}|, |C_i| - |C_{i+1}|\}  \label{eq:srink_mu_prime}.
\end{align}
We may assume w.l.o.g.~that \eqref{eq:srink_mu} holds.
Since $X$ is $c_3$-cut-linked in $G[U]$ it follows that
\begin{align}
\mu_{C_1}(C_{i+1}) &\geq c_3 (|X_i| - |X_{i+1}|). \label{eq:srink_X}
\end{align}
By \eqref{eq:srink_mu_prime} and \eqref{eq:srink_X} we get
\begin{align*}
\min\{|C_{i+1}|, |C_i|-|C_{i+1}|\} &\geq \frac{c_3}{\alpha} \cdot (|X_i| - |X_{i+1}|).
\end{align*}
Thus we obtain that for all $i\in \{0,\ldots,\ell-1\}$ (regardless or whether $P_{i+1}$ is obtained from $P_i$ by partitioning $C_i$ or not), we have
\[
|C_i|-|C_{i+1}| \geq \frac{c_3}{\alpha} \cdot (|X_i| - |X_{i+1}|).
\]
Thus
\begin{align*}
|C_0| - |C_t| &= \sum_{i=0}^{\ell-1} |C_i|-|C_{i+1}|\\
 &\geq \frac{c_3}{\alpha} \sum_{i=0}^{\ell-1} |X_i| - |X_{i+1}|\\
 &= \frac{c_3}{\alpha} (|X_0| - |X_t|).
\end{align*}
Thus
$|X_t| \geq |X_0| - \frac{\alpha}{c_3} |C_0| = |X| - \frac{\alpha}{c_3} |U\cap U'| \geq \frac{c_2}{\log^{9/2}} t - \frac{\alpha}{c_3} c_1 t \log^{3/2} n$.
By setting $\alpha=\frac{c_2 c_3}{2 c_1 \log^{6} n}$, we obtain that $|X_t| \geq \frac{c_2}{2 \log^{9/2}} t$
and $X_t$ is
$\frac{c_2 c_3 c_4}{2 c_1 \log^{13/2} n}$-cut-linked in both $G_U$ and $G[U']$,
concluding the proof.
\end{proof}

\section{Pseudogrids}\label{sec:pseudogrids}

In this section we introduce the concept of a \emph{pseudogrid}, and prove some preliminary results.
This concept will be used in the construction of a grid minor.
Specifically, we are given two sets ${\cal P}$, ${\cal Q}$ of disjoint paths such that each path in ${\cal P}$ intersects all paths in ${\cal Q}$.

\begin{definition}[Pseudogrid]
Let ${\cal P}$, ${\cal Q}$ be families of paths in some graph.
Suppose that all paths in ${\cal P}$ are vertex disjoint and all paths in ${\cal Q}$ are vertex disjoint.
Suppose further that for all $P\in {\cal P}$ and $Q\in {\cal Q}$ we have $V(P)\cap V(Q)\neq\emptyset$.
Then we say that $({\cal P}, {\cal Q})$ is a \emph{pseudogrid}.
\end{definition}

\begin{lemma}[Finding a large complete bipartite graph]\label{lem:compute_bipartite}
Let $\Gamma$ be a bipartite graph with left and right sides $L$ and $R$.
Suppose there exists $X\subseteq V(\Gamma)$ such that $\Gamma\setminus X$ is complete bipartite with left and right sides $L\setminus X$ and $R\setminus X$.
Then there exists a polynomial time algorithm which given $\Gamma$, $L$, and $R$, outputs some $X'\subseteq V(\Gamma)$ with $|X'|\leq 2|X|$ such that $\Gamma\setminus X'$ is complete bipartite with left and right sides $L\setminus X'$ and $R\setminus X'$.
\end{lemma}

\begin{proof}
We begin with $X'=\emptyset$.
If $\Gamma\setminus X'$ does not satisfy the assertion then there must exist $x\in L\setminus X'$ and $y\in R\setminus X'$ such that $\{x,y\}\notin E(\Gamma)$.
We add both $x$ and $y$ to $X'$.
We repeat until the assertion is satisfied.
It is immediate that for each $\{x,y\}$ considered, at least one of $x$ and $y$ must be in $X$.
Thus $|X'|\leq 2|X|$ as required.
\end{proof}

\begin{definition}[Compatibility]
Let ${\cal Q}$ be a collection of vertex-disjoint paths and let $P$ be a directed path in some graph.
Suppose that $P$ intersects all paths in ${\cal Q}$.
We define $\sigma^-_P({\cal Q})$ (resp.~$\sigma^+_P({\cal Q})$) to be the total ordering of ${\cal Q}$ induced by $P$ as follows: we traverse $P$ and we order the paths in ${\cal Q}$ according to the first (resp.~last) time they are visited.
We say that $P$ is \emph{compatible} with ${\cal Q}$ if $\sigma^-_P({\cal Q})=\sigma^+_P({\cal Q})$.
If $P$ is compatible with ${\cal Q}$ then we also write $\sigma_P({\cal Q})=\sigma^-_P({\cal Q})=\sigma^+_P({\cal Q})$.
\end{definition}

\begin{definition}[Combed pseudogrids]
Let $({\cal P}, {\cal Q})$ be a pseudogrid.
Suppose that all $P\in {\cal P}$ are compatible with ${\cal Q}$ and all $Q\in {\cal Q}$ are compatible with ${\cal P}$.
Suppose further that
for all $P,P'\in {\cal P}$ we have $\sigma_P({\cal Q})=\sigma_{P'}({\cal Q})$
and
for all $Q,Q'\in {\cal Q}$ we have $\sigma_Q({\cal P})=\sigma_{Q'}({\cal P})$.
Then we say that $({\cal P}, {\cal Q})$ is \emph{combed}.
\end{definition}

\begin{lemma}[Combing a pseudogrid]\label{lem:combing}
Let ${\cal P}$, ${\cal Q}$ be collections of paths in some graph $\Gamma$ such that $({\cal P}, {\cal Q})$ is a pseudogrid.
Suppose that there exist ${\cal X}\subseteq {\cal P}$ and ${\cal Y}\subseteq {\cal Q}$ such that $({\cal P}\setminus {\cal X}, {\cal Q}\setminus {\cal Y})$ is combed.
Then there exists a polynomial-time algorithm which given $\Gamma$, ${\cal P}$, and ${\cal Q}$ outputs some ${\cal X}'\subseteq {\cal P}$ and ${\cal Y}'\subseteq {\cal Q}$ such that
$({\cal P}\setminus {\cal X}', {\cal Q}\setminus {\cal Y}')$ is combed,
with
$|{\cal X}'|+|{\cal Y}'| \leq 3 (|{\cal X}|+|{\cal Y}|)$.
\end{lemma}

\begin{proof}
We compute the desired sets ${\cal X}'$ and ${\cal Y}'$ inductively starting with ${\cal X}'={\cal Y}'=\emptyset$.
If $({\cal P}\setminus {\cal X}', {\cal Q}\setminus {\cal Y}')$ is combed then we are done.
Otherwise at least one of the following conditions holds:
(i) There exists $P\in {\cal P}$ and $Q_1,Q_2\in {\cal Q}$ such that $P$ is not compatible with $\{Q_1,Q_2\}$.
(ii) There exists $Q\in {\cal Q}$ and $P_1,P_2\in {\cal P}$ such that $Q$ is not compatible with $\{P_1,P_2\}$.
In either case at least one of the three involved paths must be in ${\cal X}\cup {\cal Y}$.
In case (i) we add $P$ to ${\cal X}'$ and we add $Q_1$ and $Q_2$ to ${\cal Y}'$.
In case (ii) we add $P_1$ and $P_2$ to ${\cal X}'$ and we add $Q$ to ${\cal Y}'$.
We repeat this process until $({\cal P}\setminus {\cal X}', {\cal Q}\setminus {\cal Y}')$ is combed.
It is immediate that for any three paths that we add to ${\cal X}'\cup {\cal Y}'$ there exists at least one path in ${\cal X}\cup {\cal Y}$ and thus $|{\cal X}'|+|{\cal Y}'| \leq 3(|{\cal X}+|{\cal Y}|)$, as required.
\end{proof}

The following is implicit in \cite{chekuri2004edge,robertson1994quickly}.

\begin{lemma}[From combed pseudogrids to grid minors]\label{lem:from_combed_to_grid}
Let $({\cal P}, {\cal Q})$ be some combed pseudogrid in some graph $\Gamma$.
Then there exists a polynomial-time algorithm which given $\Gamma$, ${\cal P}$, and ${\cal Q}$ outputs some $(|{\cal P}|\times |{\cal Q}|)$-grid minor in $\Gamma$.
\end{lemma}

\section{Computing a grid minor}\label{sec:grid}

In this section, we present one of the key procedures used by our algorithm. Namely, we give an algorithm for computing a grid minor in graphs of large treewidth.
To this end, we need some notations.

\begin{definition}[Sampling a path from a flow]
Let $Y\subseteq V(G)$ and let $F$ be a product multi-commodity flow instance in $G$ that routes a unit of demand between any pair of vertices in $Y$.
We define the following probability distribution over paths in $G$.
We first sample an unordered pair of nodes $\{u,v\}\in {Y\choose 2}$ uniformly at random.
Let $F_{u,v}$ be the single-commodity flow induced by $F$ that routes a unit of demand between $u$ and $v$.
We decompose the flow $F_{u,v}$ into a collection ${\cal F}_{u,v}$ of flows such that each $F'\in {\cal F}_{u,v}$ routes $p(F')$ total flow along a single $u$--$v$ path $P(F')$.
We may assume w.l.o.g.~that all paths $P(F')$ are simple; otherwise we may shortcut them without increasing the congestion.
We choose a flow $F'\in {\cal F}_{u,v}$ with probability $p(F')$ and we let $Q$ be the path $P(F')$.
We say that the random path $Q$ is obtained by \emph{sampling from $F$}.
\end{definition}

For the remainder of this section let $G$ be an $n$-vertex graph and let $A\subseteq V(G)$ with $|A|=k$ such that $H=G\setminus A$ is planar. For the shake of our algorithm, 
we assume for a moment that maximum degree of $G$ is $\dmax$ (this assumption will be eliminated at the end of this section). 
 
We first describe the whole algorithm, and then present analysis and correctness of the algorithm.

\subsection{The algorithm for computing a grid minor}

\begin{description}
\item{\textbf{Step 1: Computing a doubly cut-linked set.}}
Using Lemma \ref{lem:doubly-linked_compute} we compute a separation $(U,U')$ of $G$ of order $O(t \log^{3/2} n)$ and some $Y\subseteq U\cap U'$ with $|Y|=\Omega(t/\log^{9/2} n)$ such that $Y$ is $\Omega(1/\log^{13/2} n)$-cut-linked in both $G_U$ and $G[U']$, where $G_U=G[U]\setminus E(G[U\cap U'])$.
Using Theorem \ref{thm:LR} we compute a multi-commodity flow $F$ in $G_U$ that routes a unit of demand between any pair of vertices in $Y$ and with congestion $c=O(|Y|\log^{15/2} n)$.
For each $\{u,v\}\in {Y\choose 2}$ let $F_{u,v}$ denote the single-commodity flow induced by $F$ that routes a unit of demand between $u$ and $v$.

\item{\textbf{Step 2: Sampling a tree.}}
Let $\ell=12$.
We sample $\{u_1,v_1\},\ldots,\{u_{\ell},v_{\ell}\} \in {Y\choose 2}$, uniformly and independently with repetition.
For each $i\in \{1,\ldots,\ell\}$ we sample a path $P_i$ from $F_{u_i,v_i}$ independently (see Figure \ref{fig:skeleton_paths}).
We remark that all $P_i$ are simple.
We proceed to construct a sequence of trees $T_1,\ldots,T_{\ell}$ with $T_1=P_1$.
For any $i\in \{2,\ldots,\ell\}$ we construct $T_i$ as follows.
If $V(P_i)\cap V(T_{i-1})=\emptyset$ then the algorithm returns $\nil$.
Otherwise let $Q$ (resp.~$Q'$) be the maximal subpath of $P_i$ that contains $u_i$ (resp.~$v_i$) and at most one vertex $u'$ (resp.~$v'$) in $V(T_{i-1})$.
We let $T_i=T_{i-1} \cup Q \cup Q'$.
This completes the construction of the sequence $T_1,\ldots,T_{\ell}$ (see Figure \ref{fig:skeleton_trees}).

\begin{figure}
\begin{center}
\subfigure[The set of paths $P_1,\ldots,P_{12}$.]{
  \scalebox{1.1}{\includegraphics{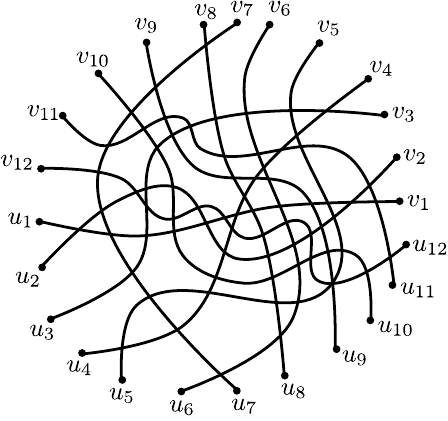}}
  \label{fig:skeleton_paths}
}
\subfigure[The trees $T_1$, $T_2$, and $T_{12}$.]{
  \scalebox{1.1}{\includegraphics{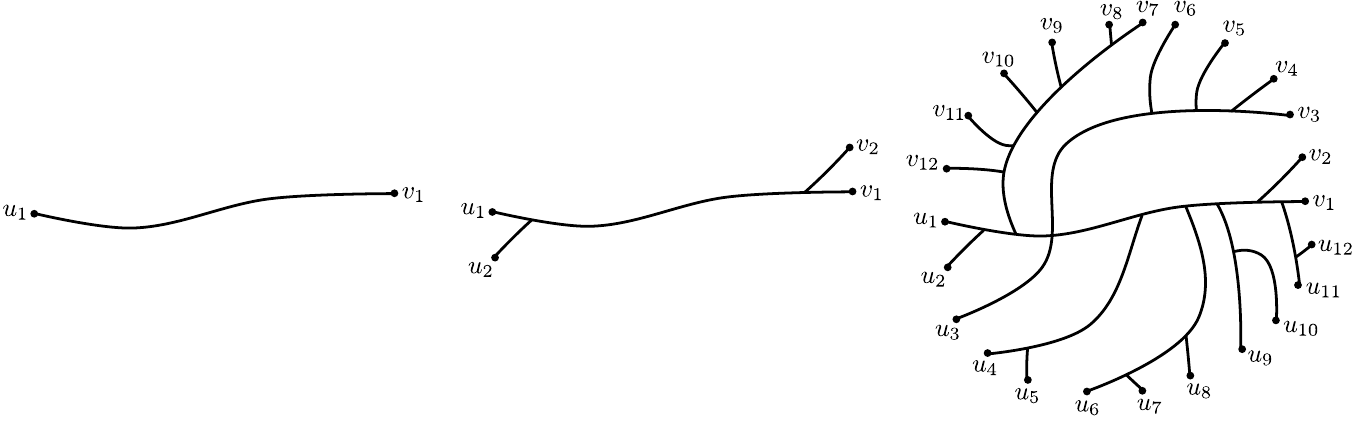}}
  \label{fig:skeleton_trees}
}
\subfigure[The graph $S$.]{
  \scalebox{1.1}{\includegraphics{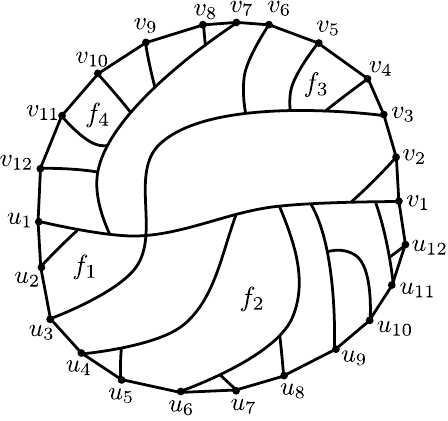}}
  \label{fig:skeleton}
}
\end{center}
\caption{Constructing a skeleton.\label{fig:skeleton_construction}}
\end{figure}


\item{\textbf{Step 3: Constructing a skeleton.}}
Let $T=T_{\ell}$.
We construct a graph $S$ as follows.
We start with $T$ and we add a cycle $C$ that visits all the leaves of $T$ in the order $u_1,\ldots,u_{\ell},v_1,\ldots,v_{\ell}$.
We check whether $S$ admits a planar drawing with $C$ being the outer face.
If not then the algorithm returns $\nil$.
Otherwise let $\psi$ be such a drawing of $S$ (see Figure \ref{fig:skeleton}).
Let $S^*$ be the dual of $S$.
Let $v^*\in V(S^*)$ be the vertex that is dual to the outer face of $S$ in $\psi$.
The graph $S^*-v^*$ is outerplanar and thus $3$-colorable.
Since $S$ has $\ell$ faces other than the outer face, it follows that there exists a collection of $\ell/3=4$ faces $f_1,\ldots,f_{4}$ in $\psi$ that are edge-disjoint.
Let us assume (after possibly reordering the indices) that $f_1,\ldots,f_4$ intersect $C$ in this order along a clockwise traversal of $C$.
For each $i\in \{1,\ldots,4\}$ let $C_i=f_i\setminus C$.
We remark that each $C_i$ is a simple path in $T$.
We set $B_i=V(C_i) \setminus \bigcup_{j\in \{1,\ldots,4\} \setminus\{i\}} V(C_j)$.

\item{\textbf{Step 4: Sampling paths between boundary components.}}
We sample a matching $M\subseteq {Y\choose 2}$ of size $|Y|/2$ uniformly at random.
For each $\{u,v\}\in M$ we sample a path $Q_{u,v}$ from $F_{u,v}$.
For any $i\in \{1,\ldots,4\}$ let $M_i$ be the set of pairs $\{u,v\}\in M$ such that
there exists a subpath $R$ of $Q_{u,v}$ containing an endpoint of $Q_{u,v}$ and with $V(R)\cap V(T)\subseteq V(B_i)$.
Let $M_{1,3}=M_1\cap M_3$ and $M_{2,4}=M_2\cap M_4$.

\item{\textbf{Step 5: Finding vertex-disjoint paths.}}
For any $i\in \{1,2\}$ let
\[
H_{i,i+2} = \bigcup_{\{u,v\}\in M_{i,i+2}} Q_{u,v}.
\]
Using Menger's Theorem we find a maximum collection ${\cal R}_{i,i+2}$ of vertex-disjoint paths in $H_{i,i+1}$ such that each path in ${\cal R}_{i,i+2}$ has one endpoint in $B_i$ and one in $B_{i+2}$.

\item{\textbf{Step 6: Finding a pseudogrid.}}
Let ${\cal I}$ be the intersection graph of the paths in ${\cal R}_{1,3} \cup {\cal R}_{2,4}$, that is
 $V({\cal I})={\cal R}_{1,3} \cup {\cal R}_{2,4}$ and
 $E({\cal I}) = \{\{P,Q\} : V(P)\cap V(Q) \neq\emptyset\}$.
Note that ${\cal I}$ is bipartite with left and right partitions ${\cal R}_{1,3}$ and ${\cal R}_{2,4}$.
Using Lemma \ref{lem:compute_bipartite} we compute some ${\cal W}_{1,3}\subseteq {\cal R}_{1,3}$ and ${\cal W}_{2,4}\subseteq {\cal R}_{2,4}$ such that ${\cal I}[{\cal W}_{1,3}\cup {\cal W}_{2,4}]$ is complete bipartite.
It follows that $({\cal W}_{1,3}, {\cal W}_{2,4})$ is a pseudogrid.

\item{\textbf{Step 7: Combing the pseudogrid.}}
Using Lemma \ref{lem:combing} we compute some ${\cal P}_{1,3}\subseteq {\cal W}_{1,3}$ and ${\cal P}_{2,4}\subseteq {\cal W}_{2,4}$ such that $({\cal P}_{1,3}, {\cal P}_{2,4})$ is combed.

\item{\textbf{Step 8: Computing a grid minor.}}
Using Lemma \ref{lem:from_combed_to_grid} we compute a $(|{\cal P}_{1,3}| \times |{\cal P}_{2,4}|)$-grid minor in $G$.

\end{description}

\subsection{Analysis}

We now give analysis of our algorithm. 
Fix a drawing $\phi$ of $H$ into the sphere ${\cal S}$ (it is convenient to assume that all faces are bounded).
Let $H_{U}$ be the subgraph of $H[U]$ obtained after removing all edges in $H[U\cap U']$, that is
\[
H_{U}=H[U]\setminus E(H[U\cap U']).
\]
Note that $H_{U} \subseteq G_U$. 

We need the following lemma. 

\begin{lemma}\label{lem:Y_prime}
There exists some universal constant $c>0$, such that the following holds.
Suppose that $t>c  k \dmax \log^{11} n$.
Then there exists some $Y'\subseteq Y\setminus A$ with $|Y'|\geq |Y|/2$ such that $Y'$ is contained
in a single connected component of $H_U$
and
in a single connected component of $H[U']$.
\end{lemma}
\begin{proof}
We first argue that there exists some $Y'_1\subseteq Y$, with $|Y'_1|\geq 3 |Y|/4$, such that $Y'_1$ is contained in a single connected component of $H_U$.
Let $\bar{Y}=Y\cap V(H_U)$.
If $\bar{Y}$ is contained in a single connected component of $H_U$, then there is nothing to show.
We may therefore assume w.l.o.g.~that $\bar{Y}$ intersects at least two connected components of $H_U$.
Let $\bar{Y}_1,\ldots,\bar{Y}_d$ be the partition of $\bar{Y}$ into maximal subsets, each contained in a single connected component of $H_U$.
Assume w.l.o.g.~that $|\bar{Y}_1|\geq \ldots\geq |\bar{Y}_d|$.
Suppose first that $|\bar{Y}_1|\leq 7|\bar{Y}|/8$.
We define a bipartition $\bar{Y}=S_1\cup S_2$ as follows:
Initially, we set $S_1=S_2=\emptyset$.
For any $i=1,\ldots,d$, we add $\bar{Y}_i$ to the set $S_j$, $j\in \{1,2\}$, of minimum cardinality.
It is immediate that we arrive at a bipartition with $\min\{|S_1|,|S_2|\}\geq |\bar{Y}|/8$.
Since $Y$ is $c$-cut-linked in $G_U$, it follows that the number of edges that we need to remove in $H_U$ to separate $S_1$ and $S_2$ is at least $\ell=\Omega(1/\log^{13/2} n)\cdot |\bar{Y}|/8
= \Omega(1/\log^{13/2} n)\cdot (|Y|-k)/8
= \Omega(1/\log^{13/2} n)\cdot |Y|/16
 = \Omega(t/\log^{11} n)$.
However, we can separate $S_1$ and $S_2$ in $H_U$ by removing all the edges incident to $A$; the number of such edges is at most $k\cdot \dmax < \ell$, for some sufficiently large universal constant $c>0$, which is a contradiction.
We have thus established that $|\bar{Y}_1|> 7|\bar{Y}|/8$.
Let $Y'_1=\bar{Y}_1$.
We have $|Y'_1| = |\bar{Y}_1| > 7|\bar{Y}|/8 = 7(|Y|-k)/8 > 3|Y|/4$, as required.

Similarly, we have that there exists some $Y'_2\subseteq Y$, with $|Y'_2|\geq 3 |Y|/4$, such that $Y'_2$ is contained in a single connected component of $H[U']$.
Setting $Y'=Y'_1\cap Y'_2$, completes the proof.
\end{proof}

For the remainder of this section let $Y'$ be as in Lemma \ref{lem:Y_prime}.
Since $Y'$ is contained in a single connected component of $H[U']$ and in a single connected component of $H_U$ it follows that there exists some disk ${\cal D}\subset {\cal S}$ such that
 $\phi(H_U) \subset {\cal D}$
and
 $\phi(H[U'] \subset {\cal S} \setminus ({\cal D}\setminus \partial {\cal D}))$.
It follows that
 $\phi(Y')\subset \partial {\cal D}$.
 For every path $P$ sampled from $F$ that has both endpoints in $Y'$ and does not intersect $A$ we have $\phi(P)\subset {\cal D}$.
Let us order $Y'=\{y_1,\ldots,y_{r}\}$ along a clockwise traversal of $\partial {\cal D}$.
For any $i\in \{1,\ldots,4\ell\}$ let
\[
Y'_i = \{y_{1+(i-1)r/(4\ell)}, \ldots, y_{ir/(4\ell)}\}.
\]
Let ${\cal E}_1$ denote the event that for all $i\in \{1,\ldots,\ell\}$ we have
$u_i \in Y_{2i}$
and
$v_i \in Y_{2\ell+2i}$. We prove the following lemma for $\Pr[{\cal E}_1]$.  

\begin{lemma}
$\Pr[{\cal E}_1] \geq (1/96)^{12}$.
\end{lemma}
\begin{proof}
For any $\{1,\ldots,\ell\}$ let ${\cal E}_{1,i}$ denote the event that $u_i\in Y_{2i}$  and  $v_i\in Y_{2\ell+2i}$.
Since by Lemma \ref{lem:Y_prime} we have $|Y'|\geq |Y|/2$ it follows that for each $i\in \{1,\ldots,\ell\}$ we have
\[
\Pr[{\cal E}_{1,i}] \geq \frac{|Y_{2i}| \cdot |Y_{2\ell+2i}|}{|Y|^2} \geq 1/96.
\]
Since the events ${\cal E}_i$ are independent for distinct $i$ it follows that
\[
\Pr[{\cal E}_1] = \Pr[{\cal E}_{1,1}\wedge \ldots \wedge {\cal E}_{1,\ell}] = \prod_{i=1}^{\ell} \Pr[{\cal E}_{1,i}] \geq (1/96)^{12},
\]
which concludes the proof.
\end{proof}

Let ${\cal E}_2$ denote the event that none of the paths $P_1,\ldots,P_{\ell}$ intersects $A$, that is
\[
\left(\bigcup_{i=1}^{\ell} V(P_i)\right) \cap A = \emptyset.
\]

 We prove the following lemma for $\Pr[{\cal E}_2]$.  

\begin{lemma}\label{lem:grid_E2}
$\Pr[{\cal E}_2] \geq 1-O(\Delta k / (t\log^3 n))$.
\end{lemma}

\begin{proof}
The congestion of $F$ is $c=O(|Y| \log^{15/2} n)$.
Thus for any $e\in E(G_U)$, the probability that a path sampled from $F$ traverses $e$ is at most $c/|Y|^2 = O(1/(|Y| \log^{15/2}n)) = O(1/(t\log^3 n))$.
There are at most $k\cdot \Delta$ edges in $G$ that have at least one endpoint in $A$.
By a union bound over all these edges it follows that a random path sampled from $F$ intersects $A$ with probability at most
$O(\Delta k / (t\log^3 n))$.
The assertion follows by a union bound over all paths $P_1,\ldots,P_{\ell}$.
\end{proof}

We now relate  ${\cal E}_1$ and ${\cal E}_2$ to present the following two lemmas. 

\begin{lemma}
Conditioned on ${\cal E}_1\wedge {\cal E}_2$ occurring, the algorithm does not return $\nil$.
\end{lemma}

\begin{proof}
The algorithm can return $\nil$ only in Steps 2 and 3.
Suppose that ${\cal E}_1\wedge {\cal W}_2$ occurs.
Then for any $i\in \{1,\ldots,\ell\}$ we have that $\phi(P_i)$ is a simple curve contains in ${\cal D}$ with endpoints in $\partial {\cal D}$.
Thus for any $i\neq j\in \{1,\ldots,\ell\}$, we have $V(P_i)\cap V(P_j)\neq \emptyset$.
Thus, for any $i\in \{2,\ldots,\ell\}$, we have $V(P_i)\cap V(T_{i-1})\neq \emptyset$.
It follows that the algorithm does not return $\nil$ during Step 2.
Similarly, the graph $S$ constructed in Step 3 can be drawn inside ${\cal D}$ by mapping the cycle $C$ to $\partial {\cal D}$.
It follows that the algorithm does not return $\nil$ during Step 3, which concludes the proof.
\end{proof}

\begin{lemma}\label{lem:large_combed_pair}
There exist a universal constant $\alpha>0$ such that if $t>\alpha k \Delta \log^3 n$ then the following holds.
Conditioned on ${\cal E}_1\wedge {\cal E}_2$ occurring, the algorithm computes a $(r \times r)$-grid minor, for some $r=\Omega(t/(\dmax\log n))$, with probability at least $1-O(1/n)$.
\end{lemma}

\begin{proof}
Consider a path $Q_{u,v}$ sampled at Step 4.
Let us first give the lower bound of the probability that $\{u,v\}\in M_{1,3}$, conditioned on ${\cal E}_1\wedge {\cal E}_2$.
Let $y_{\zeta}, y_{\zeta'}\in Y'$ be the endpoints of $C_{1}$.
and let $y_{\xi}, y_{\xi'}\in Y'$ be the endpoints of $C_{3}$.
Since ${\cal E}_1$ occurs, it follows that $|\zeta-\zeta'|\geq \ell/48$ and $|\xi-\xi'|\geq \ell/48$.
Let $u=y_i$ and $v=y_j$.
Since $|Y'|\geq |Y|/2$, we get
\[
\Pr[\zeta<i<\zeta' \text{ and } \xi<j<\xi'] \geq (1/96)^2.
\]
Since $C_1,\ldots,C_4$ are edge disjoint and $T$ is a tree it follows that for any $i\in \{1,\ldots,4\}$ there exist at most one vertex $b_i\in V(C_i)\setminus V(B_i)$.
Arguing as in Lemma \ref{lem:grid_E2} we have that
 the probability that $Q_{u,v}$  does not intersect $A\cup \{b_0\}$ is at least $1 - O(\Delta k / (t\log^3 n))$.
Thus for sufficiently large constant $\alpha>0$, it follows that there exists some constant $p>0$ such that the probability that
$\zeta<i<\zeta'$,  $\xi<j<\xi'$, and $Q_{u,v}$ does not intersect $A\cup \{b_0\}$ is at least $p$.
When these three events occur we have $\{u,v\}\in M_{1,3}$.
Thus by the Chernoff-Hoeffding Theorem we get
\begin{align}
\Pr[|M_{1,3}| \geq t p /2 ~ | ~ {\cal E}_1\wedge {\cal E}_2] &> 1-O(1/n). \label{eq:grid_M13}
\end{align}
Similarly we obtain that
\begin{align}
\Pr[|M_{2,4}| \geq t p /2 ~ | ~ {\cal E}_1\wedge {\cal E}_2] &> 1-O(1/n). \label{eq:grid_M24}
\end{align}

The maximum number of paths sampled by the algorithm is at most $t$.
Since all the paths are sampled independently, arguing as above, it follows by the Chernoff-Hoeffding Theorem and a union bound over all edges that the maximum number of paths visiting any edge is at most $O(\log n)$ with probability at least $1-O(1/n)$.
Combining with \eqref{eq:grid_M13} and \eqref{eq:grid_M24} we get that, with probability at least $1-O(1/n)$,
the size of the minimum cut in $H_{1,3}$ between $B_1$ and $B_3$ is at least $\Omega(|M_{1/3}|/\log n)$.
Thus by Menger's Theorem we get
\begin{align}
\Pr\left[|{\cal R}_{1,3}| \geq \Omega\left(t p / \Delta \log n\right) ~ | ~ {\cal E}_1\wedge {\cal E}_2\right] &\geq 1-O(1/n). \label{eq:grid_R13}
\end{align}
Similarly we get
\begin{align}
\Pr\left[|{\cal R}_{2,4}| \geq \Omega\left(t p / \Delta \log n\right) ~ | ~ {\cal E}_1\wedge {\cal E}_2\right] &\geq 1-O(1/n). \label{eq:grid_R24}
\end{align}

Let ${\cal I}$ be the conflict graph constructed at Step 6.
Let ${\cal X}\subseteq {\cal R}_{1,3}\cup {\cal R}_{2,4}$ such that ${\cal I} \setminus {\cal X}$ is complete bipartite, and such that $|{\cal X}|$ is minimized.
Every path in ${\cal R}_{1,3}$ that does not intersect $A$ must intersect every path in ${\cal R}_{2,4}$ that does not intersect $A$.
Therefore $|{\cal X}|$ is at most the number of paths in ${\cal R}_{1,3}\cup {\cal R}_{2,4}$ that intersect $A$, which is at most the number of paths sampled throughout the execution of the algorithm that intersect $A$.
Using the Chernoff-Hoeffding Theorem as above we obtain
\[
Pr[|{\cal X}| \geq |{\cal R}_{1,3}|/4 ~ | ~ {\cal E}_1\wedge {\cal E}_2] \leq O(1/n).
\]
and
\[
Pr[|{\cal X}| \geq |{\cal R}_{2,4}|/4 ~ | ~ {\cal E}_1\wedge {\cal E}_2] \leq O(1/n).
\]
Thus, with probability at least $1-O(1/n)$, the algorithm from Lemma  \ref{lem:compute_bipartite} outputs some ${\cal X}'\subset {\cal R}_{1,3}\cup {\cal R}_{2,4}$ with $|{\cal X}'|\leq 2|{\cal X}|\leq |{\cal R}_{1,3}|/2$ and $|{\cal X}'|\leq |{\cal R}_{2,4}|$, and such that ${\cal I}\setminus {\cal X}'$ is complete bipartite.
Since ${\cal W}_{1,3}={\cal R}_{1,3}\setminus {\cal X}'$ and ${\cal W}_{2,4}={\cal R}_{2,4}\setminus {\cal X}'$, we get
\begin{align}
\Pr\left[|{\cal W}_{1,3}| \geq \Omega\left(t p / \Delta \log n\right) ~ | ~ {\cal E}_1\wedge {\cal E}_2\right] &\geq 1-O(1/n),
\end{align}
and
\begin{align}
\Pr\left[|{\cal W}_{2,4}| \geq \Omega\left(t p / \Delta \log n \right) ~ | ~ {\cal E}_1\wedge {\cal E}_2\right] &\geq 1-O(1/n).
\end{align}

Finally, let ${\cal Z}\subseteq {\cal W}_{1,3}$, ${\cal Z}'\subseteq {\cal W}_{2,4}$ such that $({\cal W}_{1,3}\setminus {\cal Z}, {\cal W}_{2,4}\setminus {\cal Z}')$ is combed and $|{\cal Z}|+|{\cal Z}'|$ is minimized.
By setting ${\cal Z}$ (resp.~${\cal Z}'$) to be the set of all paths in ${\cal W}_{1,3}$ (resp.~${\cal W}_{2,4}$) that intersect $A$, we get that $({\cal W}_{1,3}\setminus {\cal Z}, {\cal W}_{2,4}\setminus {\cal Z}')$ is combed.
Therefore $|{\cal Z}|+|{\cal Z}'|$ is at most the number of paths that intersect $A$.
Arguing as above we have
\[
\Pr[|{\cal Z}|+|{\cal Z}'| \geq \min\{|{\cal W}_{1,3}|,|{\cal W}_{2,4}|\}/6 ~ | ~ {\cal E}_1\wedge {\cal E}_2] \leq O(1/n).
\]
Thus for the combed pair $({\cal P}_{1,3}, {\cal P}_{2,4})$ computed by the algorithm in Lemma \ref{lem:combing} we have
\begin{align}
\Pr\left[|{\cal P}_{1,3}| \geq \Omega\left(t p / \Delta \log n \right) ~ | ~ {\cal E}_1\wedge {\cal E}_2\right] &\geq 1-O(1/n),
\end{align}
and
\begin{align}
\Pr\left[|{\cal P}_{2,4}| \geq \Omega\left(t p / \Delta \log n \right) ~ | ~ {\cal E}_1\wedge {\cal E}_2\right] &\geq 1-O(1/n).
\end{align}
The assertion now follows by Lemma \ref{lem:from_combed_to_grid}.
\end{proof}

We are now ready to present the main result for analysis of the algorithm. 

\begin{theorem}\label{thm:k-apex_grid_minor}
There exists a randomized polynomial-time algorithm which given a graph $G$ of maximum degree $\dmax$, and $k,t\in \mathbb{N}$,
with $t>\alpha k \dmax \log^3 n$, for some universal constant $\alpha>0$,
terminates with one of the following outcomes, with high probability:
\begin{description}
\item{(1)}
Correctly decides that $\mvp(G)>k$.
\item{(2)}
Correctly decides that $\tw(G) = O(t \log^{1/2}n)$.
\item{(3)}
Outputs a $(r\times r)$-grid minor in $G$ for some $r=\Omega(t/(\Delta \log n))$.
\end{description}
\end{theorem}

\begin{proof}
Let $t=c' k \dmax \log^{11} n$.
We can pick constant $c'>0$ such that
$t>\max\{c  k \dmax \log^{11} n, \alpha k \dmax \log^3 n\}$, where $c$ and $\alpha$ are as in Lemmas \ref{lem:Y_prime} and \ref{lem:large_combed_pair} respectively.
Using Theorem \ref{thm:tw-pw-approx} we can decide in polynomial time whether $\tw(G) > t$, or $\tw(G) = O(t\log^{1/2} n)$.
In the latter case, the algorithm terminates with outcome (2).
In the former case, we run $c'' \log n$ times the algorithm from Lemma \ref{lem:large_combed_pair}; if at least one of the executions returns a $(r\times r)$-grid minor for some $r=\Omega(t/(\dmax \log n))$, then the algorithm terminates with outcome (3); otherwise we terminate with outcome (1).
The only case where the algorithm can err is if we terminate with outcome (1) while $\mvp(G) \leq k$.
By setting $c''$ to be a sufficiently large constant, the probability of this event can be made at most $1/n^{c'''}$, for an arbitrarily large constant $c'''>0$, concluding the proof.
\end{proof}



We now eliminate the condition on $\dmax$. 
It was observed by Chekuri and Chuzhoy \cite{DBLP:conf/soda/ChekuriC15} that, as a consequence of the cut-matching game of Khandekar, Rao and Vazirani \cite{DBLP:conf/stoc/KhandekarRV06}, any graph of treewidth $t$ contains a subgraph of degree at most $\log^{O(1)} t$ and treewidth at least $t/\log^{O(1)} t$.
The precise statement is as follows.

\begin{theorem}[Chekuri and Chuzhoy \cite{DBLP:conf/soda/ChekuriC15}]\label{thm:tw_degree}
There exists a randomized polynomial-time algorithm which given a graph $G$ outputs, with high probability, some subgraph $H\subseteq G$ of maximum degree at most $O(\log^3 \tw(G))$ and with $\tw(H) = \Omega(\tw(G)/\log^6 \tw(G))$.
\end{theorem}

We remark that Chekuri and Chuzhoy \cite{DBLP:conf/soda/ChekuriC15} have proved that any graph $G$ contains some topological minor $H$ of maximum degree $3$ and with $\tw(H)\geq \tw(G)/\log^{O(1)} \tw(G)$.
Despite the stronger degree bound, this leads to a weaker approximation factor in our application.

\begin{corollary}\label{cor:k-apex_grid_minor_general}
There exists a randomized polynomial-time algorithm which given a graph $G$ and $k,t\in \mathbb{N}$,
with $t>\alpha' k \log^6 n$, for some universal constant $\alpha'>0$,
terminates with one of the following outcomes, with high probability:
\begin{description}
\item{(1)}
Correctly decides that $\mvp(G)>k$.
\item{(2)}
Correctly decides that $\tw(G) = O(t \log^{13/2} n)$.
\item{(3)}
Ouptuts an $(r\times r)$-grid minor in $G$ for some $r=\Omega(t/\log^4 n)$.
\end{description}
\end{corollary}

\begin{proof}
Let $H$ be the subgraph of $G$ computed by the algorithm of Theorem \ref{thm:tw_degree}.
Since $t>\alpha' k \log^6 n$, by setting $\alpha'$ to be some sufficiently large constant, we get $t>\alpha k \dmax(H) \log^3 n$, and thus $t$ satisfies the condition of Lemma \ref{thm:k-apex_grid_minor} for the graph $H$.
Applying the algorithm from Theorem \ref{thm:k-apex_grid_minor} on $H$, we obtain one of the following results:
(1) We correctly decide that $\mvp(H)>k$. Since $H$ is a subgraph of $G$, we have $\mvp(G)\geq \mvp(H) > k$.
(2)
We correctly decide that $\tw(H) = O(t \log^{1/2} n)$.
By the guarantee of Theorem \ref{thm:tw_degree} we deduce that $\tw(G) = O(\tw(H) \log^6 \tw(G)) = O(\tw(H) \log^{6} n) = O(t \log^{13/2} n)$.
(3) A grid $(r\times r)$-grid minor $J$ of $H$ for some $r=\Omega(t/(\log^3 \tw(G) \cdot \log n)) = \Omega(t/\log^4 n)$; since $H$ is a subgraph of $G$, it follows that $J$ is also a grid minor of $G$.
\end{proof}

We remark that the bounds in Corollary \ref{cor:k-apex_grid_minor_general} can be slightly improved:
the proof also yields in case (1) $r=\Omega(r/(\log^3 \tw(G) \cdot \log n))$ and in case (2) $\tw(G) = O(k \log^{23/2} n \cdot (\log^9 k + \log^{O(1)} \log n))$.
In order to simplify the exposition we use the more succinct bounds obtained in Corollary \ref{cor:k-apex_grid_minor_general}.

Finally, we restate Corollary \ref{cor:k-apex_grid_minor_general} so that the output in case (1) is a combed pseudogrid instead of a grid minor.
This is for technical reasons that simplify the remaining algorithms and proofs.
The proof of Corollary \ref{cor:k-apex_grid_minor_general_combed} is identical to the proof of Corollary \ref{cor:k-apex_grid_minor_general}, so it is omitted.

\begin{corollary}\label{cor:k-apex_grid_minor_general_combed}
There exists a randomized polynomial-time algorithm which given a graph $G$ and $k,t\in \mathbb{N}$,
with $t>\alpha' k \log^6 n$, for some universal constant $\alpha'>0$,
terminates with one of the following outcomes, with high probability:
\begin{description}
\item{(1)}
Correctly decides that $\mvp(G)>k$.
\item{(2)}
Outputs some $3/4$-separator $S$ of $G$ with $|S|=O(t \log^{13/2} n)$.
\item{(3)}
Outputs some combed pseudogrid $({\cal P}, {\cal Q})$ in $G$, with $|{\cal P}| = |{\cal Q}| = r$, for some $r=\Omega(t/\log^4 n)$.
\end{description}
\end{corollary}


\section{Computing a partially triangulated grid contraction with a few apices}\label{sec:pt_grid}

In this section, we are given a combed pseudogrid in $G$. 
We shall construct a  partially triangulated grid contraction, after deleting a small set of vertices.
We begin with some notation.

\begin{definition}[Simple combed pseudogrid]
Let $G$ be a graph and let $({\cal P}, {\cal Q})$ be a combed pseudogrid in $G$.
We say that $({\cal P}, {\cal Q})$ is \emph{simple} if for all $P\in {\cal P}$ and $Q\in {\cal Q}$, we have that $P\cap Q$ is a path.
\end{definition}

For the remainder of this section, let $G$ be a graph.
Let $({\cal P}, {\cal Q})$ be a simple combed pseudogrid in $G$, with $|{\cal P}|=|{\cal Q}|=r$, where $r\geq c\cdot k \log n$, some some sufficiently large constant $c>0$ to be determined.
Let
\[
\Psi=\left(\bigcup_{P\in {\cal P}} P\right) \cup \left( \bigcup_{Q\in {\cal Q}} Q\right)
\]
be the induced grid minor in $G$.
Fix a planar drawing $\phi$ of $\Psi$ into the plane.
We may assume w.l.o.g.~that each path in ${\cal P}$ has both endpoints in some paths in ${\cal Q}$, and each path in ${\cal Q}$ has both endpoints in some paths in ${\cal P}$.
It follows that $\Psi$ is the subdivision of some 3-connected planar graph.
Thus all planar embeddings of $\Psi$ are combinatorially equivalent.
Let ${\cal P}=\{P_1,\ldots,P_r\}$, where the paths $P_1,\ldots,P_r$ are the rows in the standard drawing of $\Psi$, and they appear in this order from top to bottom.
Similarly, let ${\cal Q}=\{Q_1,\ldots,Q_r\}$, where $Q_1,\ldots,Q_r$ are the columns of $\Psi$, and they appear in this order from left to right.

For each $i,j\in \{1,\ldots,r-1\}$ let $F_{i,j}\subset \Psi$ be the cycle bounding the unique face in $\phi$ such that $F_{i,j}\subset P_i\cup P_{i+1} \cup Q_j \cup Q_{j+1}$.

We are now ready to present several technical lemmas. 

\subsection{Row/column restrictions}

In this subsection, we shall present ``row/column restrictions'' which will be used 
in the main result of this section. For this purpose, let us give some tools needed here. 

We recall that an instance of the \NWMWC~problem is a pair $(\Gamma,S)$, where $\Gamma$ is a graph with non-negative weights on its vertices and $S\subseteq V(\Gamma)$ is a set of \emph{terminals}.
The goal is to find some $C\subseteq V(\Gamma)\setminus S$, such that for all $s\neq t\in S$, we have that $s$ and $t$ lie in different connected components in $\Gamma\setminus C$; in other words, by deleting all vertices in $C$ we separate all pairs of terminals.
Let $\OPTNWMWC(\Gamma,S)$ denote the minimum total weight of any solution $C$ for the instance $(\Gamma,S)$.
We will use the algorithm for \NWMWC~summarized int he following.

\begin{theorem}[Garg, Vazirani, and Yannakakis \cite{garg2004multiway}]\label{thm:NWMWC-approx}
There exists a polynomial-time $(2-2/|S|)$-approximation algorithm for the \NWMWC~problem.
\end{theorem}


An instance of \NWMWC~is called \emph{uniform} if all non-terminal nodes  have unit weight.

We now present the precise definition of row/column restrictions.
\begin{definition}[Restriction]
Let ${\cal A}=\{A_i\}_{i=1}^\ell$ be a collection of pairwise disjoint subsets of $V(\Psi)$, for some $\ell\geq 1$.
We define the \emph{${\cal A}$-restriction} to be the uniform instance $(\Gamma,S)$ of the \NWMWC~problem with
\[
V(\Gamma) = (V(G) \setminus V(\Psi)) \cup \left( \bigcup_{i=1}^{\ell} \left(\{s_i,s_i'\} \cup A_i \right) \right),
\]
where $s_1',\ldots,s_\ell'$ are new vertices, and
\[
E(\Gamma) = \{\{u,v\}\in E(G) : u,v \in V(\Gamma)\} \cup \left( \bigcup_{i=1}^\ell \bigcup_{v\in A_i} \{\{s_i',v\}\} \right) \cup \left( \bigcup_{i=1}^\ell \{\{s_i,s_i'\}\}\right),
\]
and
\[
S = \{s_1,\ldots,s_\ell\}.
\]
We remark that every terminal vertex in $S$ is adjacent to a unique non-terminal vertex in $\Gamma$.
\end{definition}

\begin{definition}[Row/column restriction]\label{defn:row_col_restriction}
Let $I\subseteq \{1,\ldots,r-1\}$.
Let $F_1 = \bigcup_{j=1}^{r-1} F_{1,j}$,
$F_{r-1} = \bigcup_{j=1}^{r-1} F_{r-1,j}$, and for any $i\in \{2,\ldots,r-2\}$, let $F_i = \bigcup_{j=2}^{r-2} F_{i,j}$.
Let
${\cal A} = \bigcup_{i\in I} \{V(F_i)\} \}$,
and let $(\Gamma, S)$ be the ${\cal A}$-restriction.
Then we say that $(\Gamma, S)$ is the \emph{$I$-row restriction}.
The \emph{$I$-column restriction} is defined analogously, with the only difference being that we set
$F_1 = \bigcup_{i=1}^{r-1} F_{i,1}$,
$F_{r-1} = \bigcup_{i=1}^{r-1} F_{i,r-1}$, and for any $j\in \{2,\ldots,r-2\}$,  $F_j = \bigcup_{i=2}^{r-2} F_{i,j}$.
\end{definition}

For some $I\subseteq \mathbb{N}$, we say that $I$ is \emph{sparse} if it satisfies the following conditions:
\begin{description}
\item{(1)}
For all $i\neq j\in I$, we have $|i-j|\geq 4$.
\item{(2)}
$\left|\{1,r-1\}\cap I\right| \leq 1$.
\end{description}

Figure \ref{fig:row_restriction} depicts an example of a row restriction.

\begin{figure}
\begin{center}
\scalebox{0.85}{\includegraphics{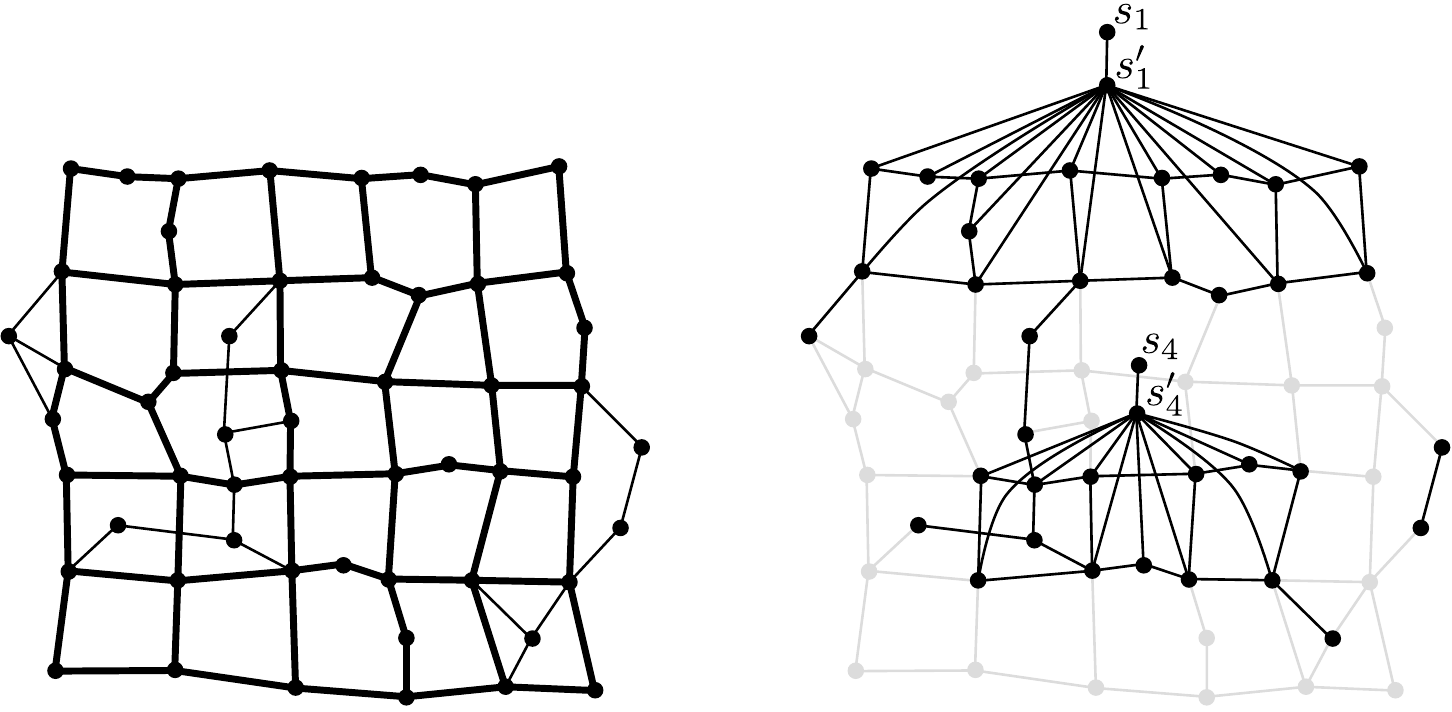}}
\caption{The graph $G$ with the subgraph $\Psi$ depicted in bold (left), and the $\{1,4\}$-row restriction with the deleted part of $G$ depicted in light gray (right).\label{fig:row_restriction}}
\end{center}
\end{figure}

Here is the main lemma concerning row/column restrictions.

\begin{lemma}\label{lem:restriction_cost}
Let $I\subseteq \{1,\ldots,r-1\}$ be sparse.
Let $(\Gamma,S)$ be either the $I$-row restriction or the $I$-column restriction.
Then $\OPTNWMWC(\Gamma,S) \leq 4 \cdot \mvp(G)$.
\end{lemma}

\begin{proof}
Let us assume w.l.o.g.~that $(\Gamma,S)$ is the $I$-row restriction; the case where $(\Gamma,S)$ is the $I$-column restriction is identical.

Let $k=\mvp(G)$.
Let $X\subseteq V(G)$, with $|X|=k$, such that $G\setminus X$ is planar.
Let
\[
J = \{i\in \{2,\ldots,r-2\} : (V(F_{i-1})\cup V(F_i)\cup V(F_{i+1}))\cap X \neq \emptyset\}.
\]
Let
\[
C = X \cup \left(\bigcup_{i\in J} \{s_i'\}\right).
\]

Clearly, we have
\begin{align*}
|C| &\leq |X| + |J|\\
 &\leq |X| + 3|X|\\
 &= 4|X|.
\end{align*}

It remains to show that $C$ is a valid solution for the \NWMWC~instance $(\Gamma,S)$.
Suppose for the sake of contradiction that $C$ is not a valid solution.
It follows that there exist $s_i\neq s_j\in S$, such that $C$ does not separate $s_i$ and $s_j$.
Let $P$ be a path between $s_i$ and $s_j$ in $\Gamma \setminus C$.
Any such path must visit $V(F_i)$ and $V(F_j)$.
Let $P'$ be the subpath of $P$ between some vertex $v_i\in V(F_i)$ and some vertex $v_j\in V(F_j)$.
Let
\[
Z_i = F_{i-1} \cup F_i \cup F_{i+1},
\]
and
\[
Z_j = F_{j-1} \cup F_j \cup F_{j+1},
\]
where we have used the notational convention that $F_{-1}$ and $F_{r+1}$ are the empty graph.
Note that $Z_i$ and $Z_j$ are subdivisions of some grids; let $\partial Z_i$ and $\partial Z_j$ denote their boundary cycles.
Note that there exist at least $r-2$ vertex-disjoint paths between $V(\partial Z_i)$ and $V(\partial Z_j)$ in $\Psi$.
Since $r>k+2$, it follows that there exists some path $Q$ in $\Psi \setminus X$ between some $u_i\in V(\partial Z_i)$ and some $u_j\in V(\partial Z_j)$.
By the choice of $C$, we have that $Z_i\cup Z_j \cup P'\cup Q$ is a subgraph of $G\setminus X$.
Thus $Z_i\cup Z_j\cup P'\cup Q$ must be planar.
Since $Z_i$ and $Z_j$ are subdivisions of $3$-connected graphs, it follows that they admit combinatorially unique planar embeddings.
Fix a planar embedding $\gamma$ of $Z_i\cup Z_j\cup P'\cup Q$.
It follows that $\partial Z_i$ and $\partial Z_j$ are contained in the same face of $\gamma$, and $v_i$ and $v_j$ are in the same face of $\gamma$, which is a contradiction.
It thus follows that $C$ is a valid solution for $(\Gamma, S)$, concluding the proof.
\end{proof}

\subsection{Sparse sets}

We say that some sparse $I\subseteq \{1,\ldots,r-1\}$ is \emph{maximal} if for all sparse $I'\subseteq \{1,\ldots,r-1\}$ with $I\subseteq I'$, we have $I=I'$. 
We now find some probability distribution ${\cal D}$ over subsets $I\subseteq \{1,\ldots,r-1\}$, where $I$ is sparse. 

\begin{lemma}\label{lem:sparse_collection}\label{lem:sample_sparse}
There exists a probability distribution ${\cal D}$ over subsets $I\subseteq \{1,\ldots,r-1\}$, where $I$ is sparse, and such that for all $x,y\in \{1,\ldots,r-1\}$, with $|x-y|\geq 4$, and $\{x,y\}\neq \{1,r-1\}$, we have
\[
\Pr_{I\sim {\cal D}}[\{x,y\}\subseteq I] \geq 1/182.
\]
Moreover, ${\cal D}$ can be sampled in polynomial time.
\end{lemma}

\begin{proof}
We define the desired distribution ${\cal D}$ via the following sampling process.
We construct a random sequence $I^0,\ldots,I^{\tau}$ of sparse subsets of $\{1,\ldots,r-1\}$.
We start by setting $I^0=\emptyset$.
We then proceed inductively as follows.
Suppose that $I^i$ has been constructed for some $i\geq 0$;
if $I^i$ is not maximal, then let
\[
W^i=\{x\in \{1,\ldots,r-1\} : \text{ for all } x'\in I^i \text{ we have } |x-x'| < 4\}.
\]
We sample some $x\in W^i$ uniformly at random, and we set $I^{i+1}=I^i\cup \{x\}$.
Otherwise, if $I^i$ is maximal, then we set $\tau=i$, and terminate the sequence at $I^i$.
Given $I^\tau$, we proceed as follows:
If $|\{1,r-1\}\cap I^i| < 2$, then we set $I=I^\tau$.
Otherwise (i.e.~if $\{1,r-1\}\subseteq I^\tau$), then we randomly choose $I=I^\tau\setminus \{1\}$ with probability $1/2$, or $I=I^\tau\setminus \{r-1\}$ with probability $1/2$.
This completes the process for sampling $I$ from ${\cal D}$.

Clearly, the process runs in polynomial time.
It is also immediate by construction that $I$ is sparse with probability $1$.
It thus remains to show that for all $x,y\in \{1,\ldots,r-1\}$, with $|x-y|\geq 4$, and $\{x,y\}\neq \{1,r-1\}$, we have that $I$ contains both $x$ and $y$ with at least the desired probability.
For any $z\in \{1,\ldots,r-1\}$, let
\[
J_z = \{z-3,\ldots,z+3\}.
\]
As we inductively add elements to the current set $I^i$ during the construction, let ${\cal E}$ denote the event that $x$ is the first element to be chosen from $J_x$, and $y$ is the first element to be chosen from $J_y$.
We have
\begin{align*}
\Pr[{\cal E}] &\geq \frac{1}{|J_x\cup J_y|} \cdot \frac{1}{|J_x\cup J_y|-1} \geq \frac{1}{14}\cdot \frac{1}{13} = \frac{1}{182}.
\end{align*}
Conditioned on ${\cal E}$ occurring, we have $\{x,y\}\subseteq I^\tau$.
Thus
\begin{align*}
\Pr[\{x,y\}\subseteq I^{\tau}] &= \Pr[{\cal E}] \geq 1/182
\end{align*}
Since $\{x,y\}\neq \{1,r-1\}$, it follows that
\begin{align*}
\Pr[\{x,y\}\subseteq I] &= \Pr[\{x,y\}\subseteq I^{\tau}] \geq 1/182,
\end{align*}
concluding the proof.
\end{proof}

Using Lemma \ref{lem:sample_sparse} we obtain the following lemma. 

\begin{lemma}\label{lem:sparse_cover}
There exists ${\cal I}=\{I_1,\ldots,I_{\ell}\}$, satisfying the following conditions:
\begin{description}
\item{(1)}
For all $t\in \{1,\ldots,\ell\}$ we have that $I_t\subseteq \{1,\ldots,r-1\}$, and $I_t$ is sparse.

\item{(2)}
For all $x,y\in \{1,\ldots,r-1\}$, if $|x-y|\geq 4$, and $\{x,y\}\neq \{1,r-1\}$, then there exist some $t\in \{1,\ldots,\ell\}$, such that $\{x,y\}\subseteq I_t$.

\item{(3)}
$\ell=O(\log r)$.
\end{description}
Moreover ${\cal I}$ can be computed in polynomial time, with high probability.
\end{lemma}

\begin{proof}
Let ${\cal D}$ be the probability distribution given by Lemma \ref{lem:sample_sparse}.
Let $\ell= c \cdot \log r$, for some constant $c>0$ to be determined later.
We sample ${\cal I}=\{I_1,\ldots,I_\ell\}$, where each $I_i\in {\cal I}$ is sampled randomly and independently from ${\cal D}$.
Conditions (1) \& (3) clearly hold.
If condition (2) does not hold, then we sample another ${\cal I}$, and we repeat until condition (2) is satisfied.

It remains to bound the running time.
By Lemma \ref{lem:sample_sparse} we have that for all $i\in \{1,\ldots,\ell\}$, and for all $x,y\in \{1,\ldots,r\}$, with $|x-y|\geq 4$ and $\{x,y\}\neq \{1,r\}$, we have that $\Pr[\{x,y\}\subseteq I_i]\geq 1/182$.
Thus,
\[
\Pr[\exists i\in \{1,\ldots,\ell\} \text{ s.t.~} \{x,y\}\subseteq I_i] \geq 1-(1-1/182)^\ell > 1-1/(2r^2),
\]
where the last inequality follows by setting $c=378$.
By taking a union bound over all $\{x,y\}$, we get that condition (2) holds at the end of each iteration with probability at least $1/2$.
It follows that the algorithm terminates after $O(\log n)$ iterations with high probability.
\end{proof}

\subsection{The algorithm for finding a nearly flat partially triangulated grid contraction}

We now present an algorithm for computing a contraction that consists of a large grid and some additional edges such that the ``dual distance'' between their endpoints in the grid is small. To this end, we need the following definition.

\begin{definition}[Purification]\label{defn:purification}
Recall that for all $P\in {\cal P}$, $Q\in {\cal Q}$, we have that $P\cap Q$ is a path.
Let $G'$ be the contraction of $G$ obtained by contracting for all $P\in {\cal P}$, $Q\in {\cal Q}$, the path $P\cap Q$ into a single vertex,
and for each $P\in {\cal P}\cup {\cal Q}$, each subpath of $P$ that is induced in $\Psi$, to a single edge.
Note that the above operation contracts $\Psi$ into some grid $\Gamma$.
Let $I,J\subseteq \{1,\ldots,r-1\}$.
Assume that $I=\{i_1,\ldots,i_a\}$, $J=\{j_1,\ldots,j_b\}$, with $i_1<\ldots<i_a$, and $j_1<\ldots<j_b$.
Let $G''$ be the contraction of $G'$ defined as follows:
If $i_1>1$ then we contract the subpath of each column of $\Gamma$ between row $1$ and row $i_1$ to its endpoint in $i_1$.
Similarly, if $j_1>1$ then we contract the subpath of each row of $\Gamma$ between column $1$ and column $j_1$ to its endpoint in $j_1$.
Next, for all $\iota\in \{1,\ldots,a\}$, we contract the subpath of each column of $\Gamma$ between row $i_\iota$ and row $i_{\iota}-1$ to its endpoint in row $i_\iota$.
Finally, for all $\gamma\in \{1,\ldots,b\}$, we contract the subpath of each row of $\Gamma$ between column $j_\gamma$ and column $j_{\gamma}-1$ to its endpoint in column $j_\gamma$.
We say that $G''$ is the \emph{$(I,J)$-purification of $G$}.
\end{definition}

\begin{lemma}\label{lem:almost_triag_grid}
There exists a randomized polynomial-time algorithm,  which given $G$, ${\cal P}$, and ${\cal Q}$, outputs with high probability some
$A\subset V(G)$,
and some
contraction $\xi:V(W)\to 2^{V(G\setminus A)}$ such that
$V(W)=V(H')$,
$E(W)=E(H')\cup E^*_1\cup E^*_2$, where $H'$ is the $(r'\times r')$-grid for some $r'\geq r/2$, and such that the following conditions are satisfied:
\begin{description}
\item{(1)}
$|A| = O(k \log n)$.

\item{(2)}
For all $\{u,v\}\in E^*_1$, we have $\max\{d_{H'}(u, V(\partial H')), d_{H'}(v, V(\partial H'))\} \leq 3$.

\item{(3)}
For all $\{u,v\}\in E^*_2$, we have $d_{H'}(u,v) \leq 6$.
\end{description}
\end{lemma}

\begin{proof}
Let $I=J=\{1,\ldots,r-1\}$.
Let ${\cal I}=\{I_1,\ldots,I_\ell\}$, ${\cal J}=\{J_1,\ldots,J_\ell\}$ be  collections of sparse subsets of $\{1,\ldots,r-1\}$, for some $\ell=O(\log r)$, computed by Lemma \ref{lem:sparse_cover} (we may assume w.l.o.g.~that $|{\cal I}|=|{\cal J}|$ by padding one of the two sets).
For each $t\in \{1,\ldots,\ell\}$, let $(\Gamma^{\row}_t, S^\row_t)$ be the $I_t$-row restriction, and let $(\Gamma^{\col}_t, S^\col_t)$ be the $J_t$-column restriction.
Let $S^\row_t = \{s^\row_{t,1},\ldots,s^\row_{t,a_t}\}$,
and $S^\col_t = \{s^\col_{t,1},\ldots,s^\col_{t,b_t}\}$.

By Lemma \ref{lem:restriction_cost}, for all $t\in \{1,\ldots,\ell\}$, we have
\[
\OPTNWMWC(\Gamma^\row_t,S^\row_t) \leq 4 k,
\]
and
\[
\OPTNWMWC(\Gamma^\col_t,S^\col_t) \leq 4 k.
\]
Thus, for all $t\in \{1,\ldots,\ell\}$, using Theorem \ref{thm:NWMWC-approx} we can compute in polynomial time some solution $C_t^\row$ for $(\Gamma^\row_t,S^\row_t)$, with
\[
|C^\row_t| \leq 2\cdot \OPTMWC(\Gamma^\row_t,S^\row_t) \leq 8 k,
\]
and some solution $C_t^\col$ for $(\Gamma^\col_t,S^\col_t)$, with
\[
|C^\col_t| \leq 2\cdot \OPTMWC(\Gamma^\col_t,S^\col_t) \leq 8 k.
\]

Let
\[
C=\bigcup_{t=1}^\ell (C^\row_t \cup C^\col_t).
\]
Let
\[
A = C\cap V(G).
\]
Since $A\subseteq C$, we have $|A|\leq |C|\leq 2 \ell 8 k = O(k\log r) = O(k \log n)$, and thus condition (1) holds.

For each $t\in \{1,\ldots,\ell\}$, let $I_t'$ be the set of integers $i\in I_t$, such that the unique vertex adjacent to the terminal $s^\row_{t,i}$ in $\Gamma^\row_t$, is in $C^\row_t$.
Similarly, for each $t\in \{1,\ldots,\ell\}$, let $J'_t$ be the set of integers $j\in J_t$, such that the unique vertex incident to the terminal $s^\col_{t,j}$ in $\Gamma^\col_t$, is in $C$.
Let $I'=\bigcup_{t=1}^\ell I'_t$, and $J'=\bigcup_{t=1}^\ell J'_t$.
Let
\[
I^* = \bigcup_{i\in I'} \{i,i+1\}
\text{ ~~ and ~~ }
J^* = \bigcup_{j\in J'} \{j,j+1\}.
\]
Let
\[
{\cal P}' = {\cal P} \setminus \bigcup_{i\in I^*} \{P_i\}
\text{ ~~ and ~~ }
{\cal Q}' = {\cal Q} \setminus \bigcup_{j\in J^*} \{Q_j\}.
\]
We first derive a lower bound on $|{\cal P}'|$ and $|{\cal Q}'|$.
We have
\begin{align*}
|{\cal P}'| &\geq |{\cal P}| - 2|I'| \geq |{\cal P}| - 16 k \ell \geq |{\cal P}|/2,
\end{align*}
where the last inequality follows when $|{\cal P}| \geq c k \log n$, for some sufficiently large constant $c>0$.
Similarly, we get $|{\cal Q'}|\geq|{\cal Q}|/2$.

Let
\[
\Psi'=\left(\bigcup_{P\in {\cal P}'} P\right) \cup \left(\bigcup_{Q\in {\cal Q}'} Q\right).
\]
Note that $\Psi'\subseteq \Psi \subseteq G$.
Let $W$ be the $(I\setminus I^*, J\setminus J^*)$-purification of $G\setminus A$ with contraction mapping $\xi : V(W) \to 2^{V(G\setminus A)}$.
Note that $H'=\xi^{-1}(\Psi')\subseteq W$ is the $(r_1,r_2)$-grid, where $r_1=|I\setminus I^*|=|{\cal P}'| \geq r/2$, and $r_2=|J\setminus J^*|=|{\cal Q}'| \geq r/2$.
It remains to show that for all $\{u',v'\}\in E(W)\setminus E(H')$, we have
\[
\max\{d_{H'}(u', V(\partial H')), d_{H'}(v', V(\partial H'))\} \leq 3,
\text{ or }
d_{H'}(u',v') \leq 6.
\]
To that end, suppose for the sake of contradiction that there exists $\{u',v'\}\in E(W)\setminus E(H')$, such that
\[
\max\{d_{H'}(u', V(\partial H')), d_{H'}(v', V(\partial H'))\} > 3,
\text{ and }
d_{H'}(u',v') > 6.
\]
It follows that there exist $u\in \xi(u')$, $v\in \xi(v')$, and some path $Z$ in $G\setminus A$ between $u$ and $v$, that interesects $\Psi'$ only on $u$ and $v$.
Suppose that $u'$ is in row $i_1'$ and column $j_1'$ of $H'$, and that $v'$ is in row $i_2'$ and column $j_2'$ of $H'$.
We may assume w.l.o.g.~that $|i_1'-i_2'|\geq |j_1'-j_2'|$, since the case $|j_1'-j_2'|\geq |i_1'-i_2'|$ is identical after exchanging the rows and columns of $\Psi$.
We have
\[
|i_1'-i_2'| \geq d_{H'}(u',v')/2 > 3.
\]
We may also assume w.l.o.g.~that $i_1<i_2$, since otherwise we may simply swap $u'$ and $v;$.
We may further assume w.l.o.g.~that $d_{H'}(u',V(\partial H')) \geq d_{H'}(v',V(\partial H'))$; the complementary case can be treated in a similar fashion.
Thus we have
\[
3 < i_1' < i_2' + 3.
\]
For any $i\in \{1,\ldots,r-1\}$, let $F_i$ be as in the definition of a row restriction (Definition \ref{defn:row_col_restriction}).
Let $i_1,i_2\in \{1,\ldots,r-1\}$ such that $u\in F_{i_1}$ and $v\in F_{i_2}$.
It follows by the definition of purification (Definition \ref{defn:purification}) that
$i_1\geq i_1'-1> 2$ and $i_1-i_2\geq i_1'-i_2' > 3$.
By Lemma \ref{lem:sparse_cover}, there exists $I_\iota\in {\cal I}$ such that $\{i_1,i_2\}\subseteq I_\iota$.
It follows that $V(Z)\cap C_\iota^{\row} \neq \emptyset$, and thus $A\cap V(Z)\neq \emptyset$, which contradicts the fact that $Z$ is a path in $G\setminus A$, and concludes the proof.
\end{proof}

We are now ready to obtain the main result of this section, which is the algorithm for computing a partially triangulated grid contraction.

\begin{lemma}[Computing a partially triangulated grid contraction with a few apices]\label{lem:pt_grid}
There exists a randomized polynomial-time algorithm,  which given some graph $G$, some $k\in \mathbb{N}$, and some simple combed pseudogrid $({\cal P}, {\cal Q})$ in $G$, with $|{\cal P}|=|{\cal Q}|=r$, where $r\geq ck  \log n$,
for some universal constant $c>0$,
 terminates with one of the following outcomes, with high probability:
\begin{description}
\item{(1)}
Correctly decides that $\mvp(G)>k$.
\item{(2)}
Outputs some
$A\subset V(G)$,
with $|A|=O(k \log n)$,
some $(r''\times r'')$-partially triangulated grid $W''$, for some $r''\geq  r/8$,
and some
contraction mapping $\xi':V(W'')\to 2^{V(G\setminus A)}$.
\end{description}
\end{lemma}

\begin{proof}
Let $A, W, \xi, H'$ be as in Lemma
\ref{lem:almost_triag_grid}, where $H'$ is the $(r'\times r')$-grid.
Let $E'=E(W)\setminus E(H')$.
Let also $E''$ be the set of edges in $E'$ that do not have both endpoints in the same face of $H'$.
Let $\widetilde{I}=\widetilde{J}=\{1,\ldots,r'\}$.
Let $\widetilde{I}'$ be the set of all $i\in \widetilde{I}$ such that there exists some edge in $E''$ with at least one endpoint in row $i$ of $H'$.
Similarly, let $\widetilde{J}'$ be the set of all $j\in \widetilde{J}$ such that there exists some edge in $E''$ with at least one endpoint in column $j$ of $H'$.
For each $t\in \{0,\ldots,7\}$ let $\widetilde{I}'_t$ be the set of integers $i\in \widetilde{I}'$ such that there exists some edge in $E''$ with one endpoint in some row $i$, where $i\mymod 8=t$, and one endpoint in row $i'$ of $H'$, for some $i'>i$.
Let
\[
t^*=\text{argmax}_{t\in \{0,\ldots,7\}} \{|\widetilde{I}'_t|\}.
\]
For each $i\in \widetilde{I}'_{t^*}$, the subgraph $W_i$ of $W$ induced on the rows $i-1,\ldots,i+7$ is nonplanar; this follows by the fact that the subgrid $H'_i$ of $H'$ induced on rows $i-1,\ldots,i+7$ is 3-connected, and there exists, by the choice of $i$, some edge in $W'_i$ with endpoints in different faces of $H'_i$.
Let also $G_i=G[\xi(V(W_i)]$; note that $G_i$ contains $W_i$ as a minor, and thus $G_i$ is also non-planar.
By construction, we have that for all $i\neq i'\in \widetilde{I}'_{t^*}$, the graphs $G_i$ and $G_{i'}$ are vertex-disjoint subgraphs of $G$.
Thus, if $\widetilde{I}'_{t^*}>k$, we may conclude that $\mvp(G)>k$.
Otherwise, we have
$|\widetilde{I}'|\leq 8 |\widetilde{I}'_{t^*}| \leq 8k$.
Similarly, if $|\widetilde{J}'| > 8k$, we may conclude that $\mvp(G)>k$, and otherwise we have that $|\widetilde{J}'| \leq 8k$.
Let
\[
\widetilde{I}^* = \bigcup_{i\in \widetilde{I}'} \{i-1,\ldots,i+7\}
\text{ ~~ and ~~ }
\widetilde{J}^* = \bigcup_{j\in \widetilde{J}'} \{j-1,\ldots,j+7\}
\]

Recall that $W$ is the $(I\setminus I^*, J\setminus J^*)$-purification of $G$, where $I=J=\{1,\ldots,r-1\}$, $I^*\subseteq $, and $J^*\subseteq J$.
Every row $i$ (resp.~column $j$) of $H'$ corresponds to some row $\tau(i)$  (resp.~column $\tau(j)$) of $\Psi$.
Let $I^{**} = \bigcup_{i\in \widetilde{I}^*} \{\tau(i)\}$,
and
$J^{**} = \bigcup_{j\in \widetilde{J}^*} \{\tau(j)\}$.
Let $W'$ be the $(I\setminus (I^*\cup I^{**}), J\setminus (J^*\cup J^{**}))$-purification of $G$.
It follows by construction that $W'=\Gamma \cup E^*$, where $\Gamma$ is the $(r'_1\times r_2')$-grid, for some $r_1'\geq r/2-|\widetilde{I}^{**}| \geq r/2-8|\widetilde{I}'|\geq r/2-64 k \geq r/4$, and, similarly, for some $r_2'\geq r/4$.
Moreover, for every $e\in E^*$, there exists some face of $\Gamma$ that contains both endpoints of $e$.

In order to obtain a partially triangulated grid, it suffices to ensure that there are no crossings between pairs of edges with endpoints in the same face of $\Gamma$.
To that end, let $R$ be the set of all faces $F$ of $\Gamma$ such that there exist $\{u,v\},\{u',v'\}\in E^*$, with $\{u,v,u',v'\}\in V(F)$, and with all the vertices $u$, $v$, $u'$, $v'$ being distinct, and such that along a clockwise traversal of $F$ we visit the vertices $u,u',v,v'$ in this order.
For each $F\in R$, let $\Gamma^F$ be the subgraph of $\Gamma$ induced by the $(5\times 5)$-subgrid of $\Gamma$ having $F$ as its central face.
We claim that $|R|\leq 25k$.
To that end, if $|R|>25k$, then we can find some $R'\subset R$ with $|R'|\geq |F|/25 > k$, such that for all $F,F'\in R'$, with $F\neq F'$, we have $V(\Gamma^F)\cap V(\Gamma^{F'})=\emptyset$.
However, the graphs $W'[V(\Gamma^F)]$ and $W'[V(\Gamma^{F'})]$ are vertex disjoint, which implies that $\mvp(G)>j$, which is a contradiction.
Thus, we have established that $|R|\leq 25k$.
Let $\bar{I}$ (resp.~$\bar{J})$ be the set of rows (resp.~columns) of $\Gamma$ that intersect the faces in $R$.
We have $|\bar{I}|\leq 50k$ and $|\bar{J}|\leq 50k$.
For each $i\in \bar{I}$, row $i$ of $\Gamma$ corresponds to some row $\tau'(i)$ of $G$.
Similarly, for each $j\in \bar{J}$, column $j$ of $\Gamma$ corresponds to some column $\tau'(j)$ of $G$.
Let $I^{***}=\bigcup_{i\in \bar{I}}\{\tau'(i)\}$, and $J^{***}=\bigcup_{j\in \bar{J}}\{\tau'(j)\}$.
Let $W''$ be the $(I\setminus (I^*\cup I^{**}\cup I^{***}), J\setminus (J^*\cup J^{**}\cup J^{***}))$-purification of $G$.
It is immediate by the above construction that $W''$ is a $(r_1'',r_2'')$-partially triangulated grid, where $r_1''\geq r_1'-|\bar{I}|\geq r/4-50k \geq r/8$.
We can ensure that $r_1''=r_2''$ by adding some arbitrary entries to either $I^{***}$ or $J^{***}$.
This concludes the proof.
\end{proof}


\section{Computing semi-universal vertex sets}\label{sec:semi-universal}

In this Section we present our algorithm for computing semi-universal vertex sets.

\begin{definition}[Semi-universality]
Let $G$ be a graph and let $Y\subset V(G)$.
We say that $Y$ is \emph{semi-universal}  (w.r.t.~$G$) if
for all $S\subseteq V(G)$, with $|S|=\mvp(G)$, such that $G\setminus S$ is planar, we have
\[
|Y\cap S|\geq (2/3)\cdot |Y|.
\]
\end{definition}

The following is an easy consequence. 

\begin{lemma}
Let $G$ be a graph and let $Y\subseteq V(G)$ be semi-universal.
Then
\begin{align*}
\mvp(G\setminus Y) &\leq \mvp(G) - (2/3)\cdot |Y|.
\end{align*}
\end{lemma}

\begin{proof}
Let $S\subseteq V(G)$, with $|S|=\mvp(G)$, such that $G \setminus S$ is planar.
We have
\begin{align*}
\mvp(G\setminus Y) &\leq |S\cap (V(G)\setminus Y)| = |S| - |S\cap Y| \leq \mvp(G) - (2/3)\cdot |Y|,
\end{align*}
concluding the proof.
\end{proof}

We now describe an algorithm for computing a semi-universal set.
We first introduce some notation.
For the remainder of this Section, let $G$ be a graph, and let $k\in \mathbb{N}$.
Let $X\subset V(G)$,
and let $\mu:V(H)\to 2^{V(G)\setminus X}$ be a contraction of $G\setminus X$, where $H$ is some $(r\times r)$-partially triangulated grid, for some $r>0$.
Let $\mu':V(H')\to 2^{V(G)}$ be the contraction of $G$ induced by $\mu$, where $V(H')=V(H)\cup X$, and $\mu'$ is the identity on $X$.
Let
\[
L = |N_{H'}(X)| = |\{v\in V(H):d_H(v,\partial H)\geq 3 \text{ and } N_G(X)\cap \mu'(v)\neq \emptyset\}|.
\]
Suppose that
\[
L>(144|X|+1296k)c_{FHL}\log^{3/2}n,
\]
where $\cFHL>0$ is the universal constant in Theorem \ref{thm:vertex-separators-approx}.

For any $i,j\in \{1,\ldots,r\}$, let $v_{i,j}\in V(H)$ be the vertex in row $i$ and column $j$ in $H$.
For any $i,j\in \{0,\ldots,3\}$ let
\[
V_{i,j} = \{v_{i',j'} : i'\mymod 4 = i \text{ and } j'\mymod 4 = j\},
\]
and
\[
L_{i,j} = |N_{H'}(X) \cap V_{i,j}|.
\]
Since $L=\sum_{i=0}^{3} \sum_{j=0}^{3} L_{i,j}$, it follows that there exist $i^*\in \{0,\ldots,3\}$, and $j^*\in \{0,\ldots,3\}$, such that
\[
L_{i^*,j^*} \geq L/16.
\]
Let $V^* = V_{i^*,j^*}$.
For any $x\in X$, let
$N^*(x) = N_{H'}(x) \cap V^*$.
For any $X'\subseteq X$, let also $N^*(X) = \bigcup_{x\in X'} N^*(x)$.
For any $U\subseteq V(H)$, we define
\[
\cover(U) = \{x\in X : N^*(x) \subseteq U\}.
\]

We are now ready to present our algorithm. 

\subsection{The algorithm for computing a semi-universal set}
The algorithm proceeds in the following steps:
\begin{description}
\item{\textbf{Step 1.}}
The algorithm computes a sequence ${\cal U}=U_0,U_1,\ldots,U_{\ell}$ of disjoint subsets of $V(H)$ as follows.
We set $U_0=\empty$, and for any $i\geq 1$, using the algorithm from Lemma \ref{lem:semi-universal-induction}, we obtain in polynomial time one of the following two outcomes:
\begin{description}
\item{(I)}
We compute some $U_{i+1}\subseteq V(H) \setminus (U_0\cup \ldots\cup U_i)$, such that
$|U_{i+1}| \leq \frac{L}{48 \log n}$,
and
\[
|X\setminus \cover(U_0\cup \ldots \cup U_{i+1})| \leq (3/4) \cdot |X\setminus \cover(U_0\cup \ldots \cup U_{i})|.
\]

\item{(II)}
We correctly decide that for all $U\subseteq V(H) \setminus (U_0\cup \ldots\cup U_i)$,
such that
$|U| \leq \frac{L}{48 c_{FHL} \log^{3/2} n}$,
we have
\[
|X\setminus \cover(U_0\cup \ldots \cup U_{i}\cup U)| > (2/3) \cdot |X\setminus \cover(U_0\cup \ldots \cup U_{i})|.
\]
\end{description}
If outcome (I) occurs, then we obtain $U_{i+1}$; if outcome (II) occurs, then we set $\ell=i$ and we terminate the sequence ${\cal U}$ at $U_i$.

\item{\textbf{Step 2.}}
We output the set $Y = X \setminus \cover(U_0\cup\ldots U_\ell)$.
\end{description}
This concludes the description of the algorithm for computing a semi-universal set.

\subsection{Analysis}
We now analyze the above algorithm.
We first show that the set $Y$ computed by the algorithm is non-empty and semi-universal.

\begin{lemma}\label{lem:Y_nonempty}
$Y\neq \emptyset$.
\end{lemma}

\begin{proof}
By induction on $i$, we have that $|X\setminus \cover(U_0\cup \ldots \cup U_i)|\leq (3/4)^i \cdot |X|$.
It follows that $\ell\leq \log_{4/3} |X| <  3\log n$.
Since $U_0=\emptyset$, and for all $i\in \{1,\ldots,\ell\}$, we have $|U_i|\leq L/(48 \log n)$, it follows that
\begin{align*}
|U_0\cup \ldots\cup U_\ell| &\leq \ell L / (48 \log n)\\
 &< (3 \log n) L / (48 \log n)\\
 &= L / 16\\
 &\leq L_{i^*,j^*} \\
 &= |N^*(X)|
\end{align*}
This implies that $N^*(X)\setminus (U_0\cup \ldots \cup U_\ell) \neq \emptyset$, and thus $Y=X\setminus \cover(U_0\cup \ldots \cup U_\ell)\neq \emptyset$.
\end{proof}

We now show that $Y$ is ``semi-universal''. 
\begin{lemma}\label{lem:Y_semiuniversal}
If $\mvp(G)\leq k$, then
$Y$ is semi-universal.
\end{lemma}

\begin{proof}
Assume, for the sake of contradiction, that $\mvp(G)\leq k$, and $Y$ is not semi-universal.
It follows that there exists some $S\subseteq V(G)$, with $|S|\leq k$, such that $G\setminus S$ is planar, and $|Y\cap S| < (2/3)\cdot |Y|$.
Let
\[
S' = \{v\in V(H) : S\cap \mu(v)\neq \emptyset\}.
\]
Let also
\[
S''=S' \setminus (U_0\cup \ldots \cup U_\ell).
\]
We have $|S''|\leq |S'|\leq |S| \leq k$.
Thus, by the termination outcome (II) in the construction of ${\cal U}$, we have  that
\begin{align*}
|Y\setminus \cover(U_0\cup \ldots\cup U_\ell\cup S'')| &=|X\setminus \cover(U_0\cup \ldots\cup U_\ell\cup S'')|\\
 &> (2/3) \cdot |X\setminus \cover(U_0\cup \ldots \cup U_{\ell})|\\
 &= (2/3)\cdot |Y|.
\end{align*}
Let $U=N^{*}(Y\setminus S) \setminus (U_0\cup \ldots \cup U_{\ell})$.
We have
$Y\setminus S \subseteq \cover(U_0\cup\ldots\cup U_\ell \cup U)$.
Thus
\begin{align*}
|X\setminus \cover(U_0\cup\ldots\cup U_\ell \cup U)| &\leq |Y\setminus \cover(U_0\cup\ldots\cup U_\ell \cup U)|\\
 &\leq |Y\cap S|\\
 &< (2/3)\cdot |Y|.
\end{align*}
It follows by applying the condition of outcome (II) on $U$ that
\begin{align}
|N^{*}(Y \setminus S)| & \geq |N^{*}(Y\setminus S) \setminus (U_0\cup \ldots \cup U_{\ell})| = |U| \geq \frac{L}{48 c_{FHL} \log^{3/2} n}. \label{eq:univ1}
\end{align}

Let
\[
V^{**}=\{v\in V^* : d_H(v,S')\geq 3\}.
\]
For any $x\in X$, let $N^{**}(x) = N^*(x)\cap V^{**}$, and for any $X'\subseteq X$, let $N^{**}(X')=\bigcup_{x\in X'} N^{**}(x)$.
From \eqref{eq:univ1} we get
\begin{align}
|N^{**}(Y\setminus S)| &\geq |N^{*}(Y\setminus S)| - |V^{*}\setminus V^{**}| \notag \\
 &\geq |N^{*}(Y\setminus S)| - |S'| \cdot 25 \notag \\
 &\geq \frac{L}{48 c_{FHL} \log^{3/2} n} - 25 k \label{eq:univ1**}
\end{align}

Let ${\cal C}$ be the set of connected components of $H\setminus S'$.
We argue that for all $x\in X\setminus S$, and for all $C\in {\cal C}$, we have $|N^{**}(x)\cap C| \leq 1$.
Suppose, for the sake of contradiction, that there exists $C\in {\cal C}$, and $x\in X\setminus S$, such that $|N^{**}(x)\cap C| \geq 2$.
Let $u, v \in N^{**}(x)\cap C$ be distinct vertices.
Since $d_{H}(u, S') > 2$,
it follows that there exists some $(3\times 3)$-subgrid $\Gamma_u\subset C$, such that $u$ is the central vertex of $\Gamma_u$.
Similarly, there exists some $(3\times 3)$-subgrid $\Gamma_v\subset C$, such that $v$ is the central vertex of $\Gamma_v$.
By the construction of $V^*$, we have that $d_{H}(u,v)\geq 4$,
which implies that $V(\Gamma_u)\cap V(\Gamma_v)=\emptyset$.
Since $\Gamma_u$ and $\Gamma_v$ are vertex-disjoint subgraphs of the connected component $C$, it follows that there exists path $P$ in $C$ between some vertex $u'\in V(\partial H_u)$, and some vertex $v'\in V(\partial H_v)$, that intersects $\Gamma_u\cup \Gamma_v$ only on $\{u',v'\}$.
It follows that the graph $\Gamma = \Gamma_u\cup \Gamma_v\cup P\cup \{x\}$ is non-planar.
Since $\Gamma$ is a minor of $G\setminus S$, this contradicts the fact that $G\setminus S$ is planar.
We have thus established that for all $x\in X\setminus S$, and for all $C\in {\cal C}$, we have $|N^{**}(x)\cap C| \leq 1$.

Let $J$ be the bipartite graph with
\[
V(J)=(Y\setminus S) \cup {\cal C},
\]
and
\[
E(J)=\{\{x,C\} : \text{ there exist } x\in Y\setminus S, C\in {\cal C}, \text{ s.t.~}N^{**}(x)\cap V(C) \neq \emptyset\}.
\]
By contracting each $C\in {\cal C}$ into a single vertex in $G$, we obtain $J$.
Since for all $x\in Y\setminus S$, and for all $C\in {\cal C}$, we have $|N^{**}(x)\cap C| \leq 1$, it follows that
\[
|E(J)| = |E_{H'}(Y\setminus S, V^{**})|.
\]
Let $E'$ be the subset of $E(Y\setminus S,V(H))$ obtained by deleting, for all $C\in {\cal C}$, for all $y\in C$, all but at most one of the edges in $E(Y\setminus S, \{y\})$.
We have
\[
|E'|=|N^{**}(Y\setminus S)|.
\]
Let $J'=((Y\setminus S) \cup {\cal C}, E')$.
Since $J'$ is a subgraph of $J$, and $J$ is a minor of $G\setminus S$, it follows that $J'$ is planar.
By the formula for Euler's characteristic on planar graphs, we have
\[
|E(J')| \leq 3|V(J')|-6.
\]
Thus
\begin{align}
|N^{**}(Y\setminus S)| \leq 3 (|Y\setminus S|+|{\cal C}|) \leq 3(|X| + |{\cal C}|) \label{eq:univ2}
\end{align}
Since $H$ is a connected graph of maximum degree $4$, it follows that
\begin{align}
|{\cal C}| \leq 1+3|S'| \leq 4 |S'| \leq 4|S| \leq 4k. \label{eq:univ3}
\end{align}
By \eqref{eq:univ2} and \eqref{eq:univ3}, we get
\[
|N^{**}(Y\setminus S)| \leq 3 |X| + 12k,
\]
which contradicts \eqref{eq:univ1**}, since, by assumption, $L>(144|X|+1296k)c_{FHL}\log^{3/2}n$.
We have thus established that $Y$ is semi-universal, which concludes the proof.
\end{proof}

It remains to show how Step 1 of the algorithm can be performed in polynomial time.

\begin{lemma}\label{lem:semi-universal-induction}
Let $U_0,\ldots,U_i \subset V(H)$ be disjoint subsets of $V(H)$.
There exists a polynomial-time algorithm which terminates with one of the two following outcomes:
\begin{description}
\item{(I)}
Computes some $U_{i+1}\subseteq V(H) \setminus (U_0\cup \ldots\cup U_i)$, such that
$|U_{i+1}| \leq \frac{L}{48 \log n}$,
and
\[
|X\setminus \cover(U_0\cup \ldots \cup U_{i+1})| \leq (3/4) \cdot |X\setminus \cover(U_0\cup \ldots \cup U_{i})|.
\]

\item{(II)}
Correctly decides that for all $U\subseteq V(H) \setminus (U_0\cup \ldots\cup U_i)$,
such that
$|U| \leq \frac{L}{48 c_{FHL} \log^{3/2} n}$,
we have
\[
|X\setminus \cover(U_0\cup \ldots \cup U_{i} \cup U)| > (3/2)\cdot |X\setminus \cover(U_0\cup \ldots \cup U_{i})|.
\]
\end{description}
\end{lemma}

\begin{proof}
Let $V_1 = X\setminus A(U_0\cup \ldots \cup U_i)$,
and
$V_2 = V(H) \setminus (U_0\cup \ldots \cup U_{\ell})$.
Let $K$ be the graph with $V(K)=V_1 \cup V_2$, and
\[
E(K) = E_G(V_1,V_2\cap N^*(V_1)) \cup {V_2 \choose 2}.
\]
That is, $K$ contains all edges in $G$ between $V_1$ and $V_2\cap N^*(V_1)$, and a clique on $V_2$.
Let $W=V_1$.
For any $v\in V(K)$, let
\[
c(v) = \left\{\begin{array}{ll}
\infty & \text{ if } v\in V_1\\
1 & \text{ if } v\in V_2
\end{array}\right.
\]
We remark that it is sufficient for this proof to set $c(v)=\Theta(n^2)$, for all $v\in V_1$; however we use the convention $c(v)=\infty$ in order to simplify the exposition.
Using Theorem \ref{thm:vertex-separators-approx} we compute some $3/4$-balanced vertex separator (w.r.t.~$W$) $S$ of $G$.
Since $c(v)=\infty$ for all $v\in V_1$, we may assume w.l.o.g.~that $S\subseteq V_2$, and thus $c(S)=|S|$.
We consider the following two cases:
\begin{description}
\item{Case 1:}
Suppose that $c(S)\leq \frac{L}{48 \log n}$.
Since $K$ contains a clique on $V_2$, it follows that there exists some connected component $C$ of $K\setminus S$ that contains all the vertices in $V_2\setminus S$.
Let $X'$ be the set of vertices in $V_1$ that are not in $C$.
Since $S$ is a $3/4$-balanced vertex separator w.r.t.~$W$, it follows that $|(C\setminus S)\cap W| \leq (3/4)\cdot |W|$.
Thus, $|X'|\geq |W|/4$.
Since $X'\cap C=\emptyset$, it follows that $X'\subseteq \cover(S)$.
Thus,
\begin{align*}
|X\setminus \cover(U_0\cup \ldots \cup U_i \cup S| &\leq |X\setminus \cover(U_0\cup\ldots\cup U_i)| - |X'|\\
 &\leq |X\setminus \cover(U_0\cup\ldots\cup U_i)| - |V_1|/4\\
 &= |X\setminus \cover(U_0\cup\ldots\cup U_i)| - |\cover(U_0\cup \ldots \cup U_i)|/4\\
 &= (3/4) \cdot |X\setminus \cover(U_0\cup\ldots\cup U_i)|
\end{align*}
We may thus set $U_{i+1}=S$, and terminate with outcome (I).

\item{Case 2:}
Suppose that $c(S)>  \frac{L}{48 \log n}$.
We claim that the conclusion in outcome (II) holds; namely, for all $U\subseteq V(H)\setminus (U_0 \cup \ldots \cup U_\ell)$, such that $|U|\leq \frac{L}{48 c_{FHL} \log^{3/2}n}$, we have $|X\setminus \cover(U_0\cup \ldots \cup U_i\cup U)| > (3/2)\cdot |X\setminus \cover(U_0 \cup \ldots \cup U_i)|$.
Suppose, for the sake of contradiction, that there exists some $U^*\subseteq V(H) \setminus (U_0\cup \ldots \cup U_i)$, such that
\[
|U^*| \leq \frac{L}{48 c_{FHL} \log^{3/2}n},
\]
and
\[
|X\setminus \cover(U_0\cup \ldots \cup U_i \cup U^*)| \leq (3/2)\cdot |X\setminus \cover(U_0\cup \ldots \cup U_i)|.
\]
By removing $U^*$ from $K$, every $x\in \cover(U_0\cup \ldots \cup U_i \cup U^*) \setminus \cover(U_0\cup \ldots \cup U_i)$ becomes isolated (since, by definition, $N^*(x)\subseteq \cover(U_0\cup \ldots \cup U_i\cup U^*$)).
Thus, the maximum size of any connected component in $K\setminus U^*$ is at most
$|X\setminus \cover(U_0\cup \ldots \cup U_i \cup U^*)| \leq (3/2) \cdot |X \setminus \cover(U_0\cup \ldots \cup U_i)| = |W|$.
Thus, $U_i$ is a $3/2$-balanced vertex separator (w.r.t.~$K$) of $K$.
Since $|U^*|\leq \frac{L}{48 c_{FHL} \log^{3/2}n}$, it follows that the algorithm from Theorem \ref{thm:vertex-separators-approx} outputs some $S$ with $c(S)\leq \frac{L}{48 \log n}$, which contradicts the assumption.
Thus, we have established that the conclusion in outcome (II) holds.
\end{description}
In summary, when Case 1 occurs, we terminate with outcome (I), and when Case 2 occurs, we terminate with outcome (II).
This concludes the proof.
\end{proof}

The following summarizes the main result of this section.

\begin{lemma}[Computing a semi-universal set]\label{lem:semi-universal}
Let $G$ be a graph, and let $k\in \mathbb{N}$.
Let $X\subset V(G)$,
and let $\mu:V(H)\to 2^{V(G)\setminus X}$ be a contraction of $G\setminus X$, where $H$ is the $(r\times r)$-partially triangulated grid, for some $r>0$.
Let $\mu':V(H')\to 2^{V(G)}$ be the contraction of $G$ induced by $\mu$, where $V(H')=V(H)\cup X$, and $\mu'$ is the identity on $X$.
Let
\[
L = |N_{H'}(X)| = |\{v\in V(H):d_H(v,\partial H)\geq 3 \text{ and } N_G(X)\cap \mu'(v)\neq \emptyset\}|.
\]
Suppose that
\[
L>(144|X|+1296k)c_{FHL}\log^{3/2}n.
\]
Then there exists a polynomial-time algorithm which given $G$, $X$, $H$, and $\mu$, computes some non-empty $Y\subseteq X$, satisfying the following property:
If $\mvp(G)\leq k$, then $Y$ is semi-universal (w.r.t.~$G$).
\end{lemma}

\begin{proof}
The bound on the running time follows by the fact that the algorithm in Lemma \ref{lem:semi-universal-induction} runs in polynomial time.
The fact that $Y$ is non-empty follows by Lemma \ref{lem:Y_nonempty}, and the fact that $Y$ is semi-universal follows by Lemma \ref{lem:Y_semiuniversal}.
\end{proof}


\section{Computing irrelevant vertex sets}\label{sec:irrelevant}


In this Section we present our algorithm for computing irrelevant vertices. 

Let $G$ be some $n$-vertex graph, and let $k\in \mathbb{N}$.
Let $X\subseteq V(G)$,
with $|X|\leq c k \log n$, for some universal constant $c>0$.
Let $H$ be some $(r\times r)$-partially triangulated grid, where $r\geq c' k  \sqrt{\approxfactor} \log^{5/2}n$,
for some universal constant $c'>0$ to be determined later.
Suppose that $H$ is a contraction of $G\setminus X$, with contraction mapping $\mu:V(H)\to 2^{V(G\setminus X)}$.
Let $\mu':V(H')\to 2^{V(G)}$ be the contraction of $G$ induced by $\mu$, where $V(H')=V(H)\cup X$, and $\mu'$ is the identity on $X$; that is for all $v\in V(H')$, we have
\[
\mu'(v)=\left\{\begin{array}{ll}
  \mu(v)  & \text{ if } v\notin X\\
  \{x\}   & \text{ if } v\in X
\end{array}\right.
\]
Let
\[
L = |N_{H'}(X)| = |\{v\in V(H):d_H(v,\partial H)\geq 3 \text{ and } N_G(X)\cap \mu'(v)\neq \emptyset\}|.
\]
Suppose that
\[
L \leq (144|X|+1296k)c_{FHL}\log^{3/2}n,
\]
where $\cFHL>0$ is the universal constant in Theorem \ref{thm:vertex-separators-approx}.

For each $i,j\in \{1,\ldots,r\}$, let $v_{i,j}$ be the vertex in the $i$-th row and $j$-th column of $H$.
For each $i,j,\ell\in \{1,\ldots,r\}$, let
\[
S(i,j,\ell) = \bigcup_{i'=\max\{1,i-\ell-1\}}^{\min\{r,i+\ell\}} \bigcup_{j'=\max\{1,j-\ell-1\}}^{\min\{r,j+\ell\}} \{v_{i',j'}\}.
\]
Note that $H[S(i,j,\ell)]$ is a $(2\ell\times 2\ell)$-partially triangulated grid.
We define
\[
\width(S(i,j,\ell)) = 2\ell.
\]
For each $t\in \{0,\ldots,\log r\}$ we define a partition ${\cal S}_{t,1}$ and ${\cal S}_{t,2}$ of $V(H)$, with
\[
{\cal S}_{t,1} = \bigcup_{i=0}^{\lfloor r/2^{t+1}\rfloor} \bigcup_{j=0}^{\lfloor r/2^{t+1} \rfloor} \{S(i 2^{t+1},j 2^{t+1},2^{t})\},
\]
and
\[
{\cal S}_{t,2} = \bigcup_{i=0}^{\lfloor r/2^{t+1}\rfloor} \bigcup_{j=0}^{\lfloor r/2^{t+1} \rfloor} \{S(i 2^{t+1} + 2^t,j 2^{t+1} + 2^t,2^{t})\}.
\]
Let also
\[
{\cal S} = \bigcup_{t=0}^{\log r} \bigcup_{i=1}^2 {\cal S}_{t,i}.
\]

For each $U\subseteq V(H)$, let $\gamma(U)$ be the number of vertices $v\in U$ such that there exists at least some neighbor of $X$ in $\mu(v)$; that is
\[
\gamma(U) = |N_H'(X) \cap U|.
\]

\begin{definition}[Dominating weight function]
Let $w:{\cal S}\to \mathbb{N}$.
We say that $w$ is \emph{dominating} if for all $S\in {\cal S}$, we have
\[
w(S) \geq \mvp(G[\mu(S)]) + \gamma(S).
\]
We say that some $S\in {\cal S}$ is \emph{active} (w.r.t.~$w$) if $w(S)\geq \width(S) / 2000$.
We say that some $v\in V(H)$ is \emph{forgotten} (w.r.t.~$w$) if $d_H(v, \partial H)>123 k$,
and
$v$ is not contained it any active $S\in {\cal S}$; that is
$v\notin \bigcup_{S\in {\cal S}:S \text{ is active}} S$.
We also say that some $U\subseteq V(H)$ is forgotten, if every $v\in U$ is forgotten.
Note that if $U$ is forgotten, then $N_{H'}(X)\cap U=\emptyset$, and thus $N_G(X)\cap \mu(U)=\emptyset$ (since otherwise, if $v\in N_{H'}(X)\cap U$, then $\{v\}$ is active, and thus $U$ cannot be forgotten).
\end{definition}

For some $v\in V(G)$, we say that $v$ is \emph{irrelevant} if for any $X\subset V(G)$, with $|X|=\mvp(G)$, such that $G\setminus X$ is planar, we have $v\notin X$.
Similarly, we say that some $U\subset V(G)$ is \emph{irrelevant} if for all $v\in V(G)$, $v$ is irrelevant.

In the next lemma whose proof is technical and lengthly, we show which vertices are irrelevant. 

\begin{lemma}[Forgotten vertices are irrelevant]\label{lem:forgotten_irrelevant}
Let $w:{\cal S}\to \mathbb{N}$ be \emph{dominating},
and
let $v\in V(H)$ be forgotten (w.r.t.~$w$).
If $\mvp(G)\leq k$, then $\mu(v)$ is irrelevant.
\end{lemma}

\begin{proof}
Suppose for the sake of contradiction that the assertion is false.
It follows that there exists some $v\in V(H)\setminus \left(\bigcup_{S\in {\cal S}'} S\right)$, with $d_H(v, \partial H) > 123 k$, such that $\mu(v)$ is not irrelevant.
Thus there exists $u\in \mu(v)$ that is not irrelevant.
It follows that there exists some $X_{\OPT}\subseteq V(G)$, with $|X_{\OPT}|=k$, such that $G\setminus X_\OPT$ is planar, and $u\in X_\OPT$.

We say that some $S\in {\cal S}$ is \emph{dense} if
\[
|X_{\OPT} \cap \mu(S)| \geq \width(S) / 49,
\]
and otherwise we say that it is \emph{sparse}.
Trivially, $\{v\}$ is a dense square that contains $v$.
Pick $i,j,\ell\in \mathbb{N}$ such that $S(i,j,\ell)\in {\cal S}$ is dense, with $v\in S(i,j,\ell)$, and such that $\ell$ is maximized.
Since $|X_\OPT|=k$, it follows that $2\ell = \width(S) \leq 49 k$, thus
$\ell\leq 49 k/2$.
Since $d_H(v, \partial H) > 123 k$, it follows by the triangle inequality that
\begin{align*}
d_H(S(i,j,\ell), \partial H) &= d_H(\partial H[S(i,j,\ell)], \partial H) \\
 &\geq d_H(v, \partial H) - d_H(v, \partial H[S(i,j,\ell)]) \\
 &> 123 k - \ell \\
 &\geq 4\ell.
\end{align*}
It follows that there exist $S(i',j',2\ell)\in {\cal S}$, for some $i',j'$, such that
\[
S(i,j,\ell) \subset S^{\father},
\]
and
\[
d_H(S(i,j,\ell), \partial H[S^\father]) \geq \ell,
\]
where
\[
S^\father = S_1 \cup S_2 \cup S_3 \cup S_4,
\]
\[
S_1 = S(i',j',2\ell), ~~ S_2 = S(i'+4\ell,j',2\ell), ~~ S_3 = S(i',j'+4\ell,2\ell), ~~ S_4 = S(i'+4\ell,j'+4\ell,2\ell).
\]
By the construction of ${\cal S}$, we have $S^\father\in {\cal S}$.
Let
\[
S^+ = S^\father \setminus S(i,j,\ell).
\]
For each $\ell'\in \{1,\ldots,\ell\}$, let
$C^H_{\ell'}=\partial H[S(i,j,\ell+\ell')]$;
note that $C^H_{\ell'}$ is a cycle in $H$, and the images of the cycles $C^H_1,\ldots,C^H_{\ell}$ are nested in the standard drawing of $H$ into the plane, with $C^H_1$ being the innermost cycle and $C^H_\ell$ being the outermost cycle, and with $C_1^H,\ldots,C_{\ell}^H \subset S^+$ (see Figure \ref{fig:S_father}).

\begin{figure}
\begin{center}
\includegraphics{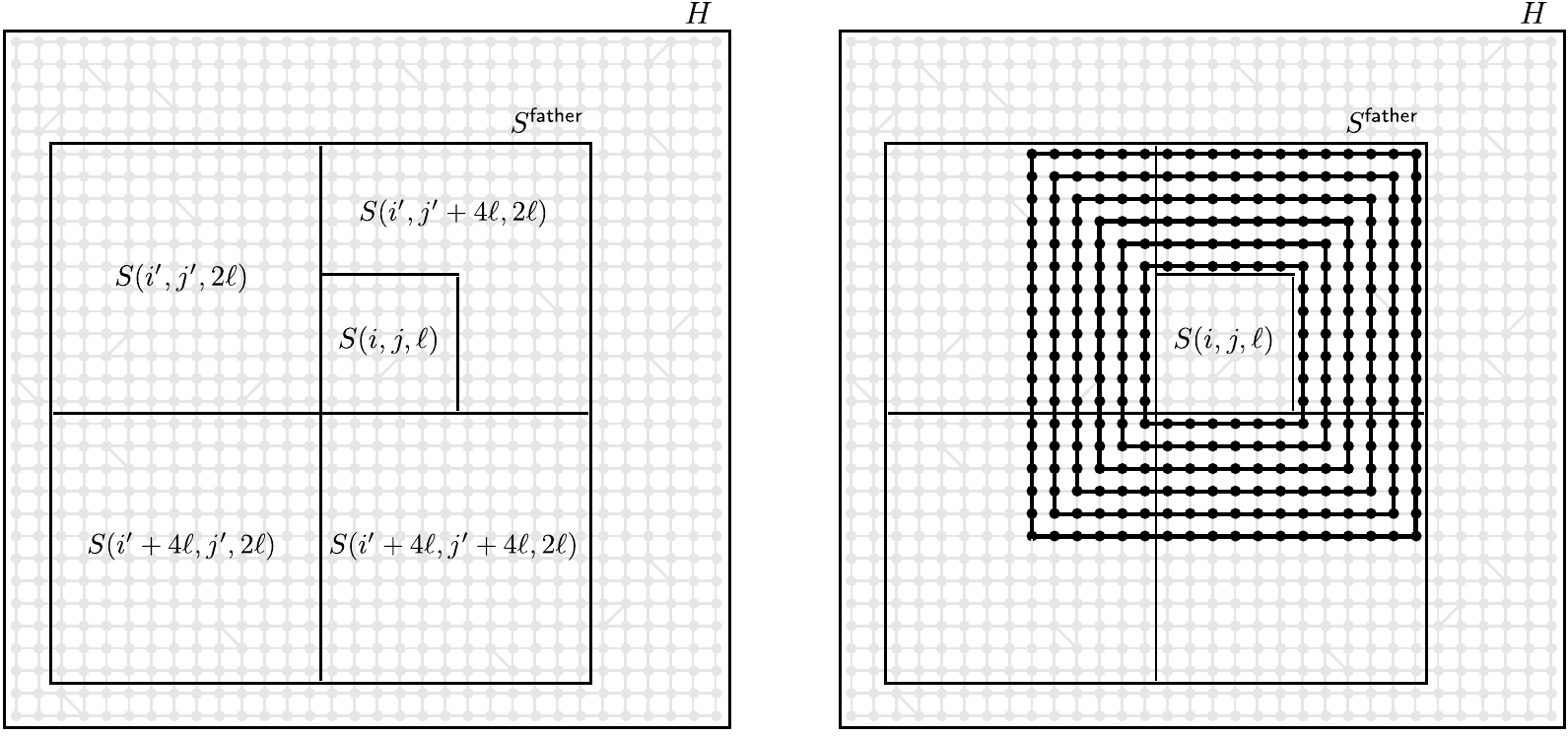}
\caption{The subsets $S(i,j,\ell)$, $S(i',j',2\ell)$, $S(i'+4\ell,j',2\ell)$, $S(i',j'+4\ell,2\ell)$, $S(i'+4\ell,j'+4\ell,2\ell)$, and $S^\father$ of $V(H)$ (left), and the cycles $C_1^H,\ldots,C_{\ell}^H$ depicted in bold (right).\label{fig:S_father}}
\end{center}
\end{figure}

Let
\[
Y_1 = N_{H'}(X) \cap \mu(S^\father).
\]
By the choice of $v$, we have that $S^\father$ is not active.
Since $w$ is dominating, and $S^\father$ is not active, it follows that there exists some $Y_2 \subset \mu(S^\father)$, such that $G[\mu(S^\father)] \setminus Y_2$ is planar,
with
\begin{align}
|Y_1\cup Y_2| \leq w(S^\father) \leq \width(S^\father) / 200 = 8 \ell / 200. \label{eq:Y_father}
\end{align}
Let $Y=Y_1\cup Y_2$, and $Z=G[\mu(S^\father)]\setminus Y_2$.

Let $X^\father_\OPT=X_\OPT \cap \mu(S^\father)$.
By the choice of $\ell$, we have that $S^\father$ is sparse, and thus
\begin{align}
|X_\OPT^\father| &= |X_{\OPT} \cap \mu(S^\father)| < \width(S^\father)/49 = 8\ell / 49. \label{eq:father_X_OPT_plus}
\end{align}
By \eqref{eq:Y_father} and \eqref{eq:father_X_OPT_plus}
we get
\begin{align*}
|Y\cup X_\OPT^\father| &\leq |Y| + |X_\OPT^\father| < \ell / 3.
\end{align*}
Thus there exists $t\in \{2,\ldots,\ell-1\}$, such that $(V(C_{t-1}^H) \cup V(C_t^H) \cup V(C_{t+1}^H)) \cap (Y\cup X_\OPT^\father)=\emptyset$.

Let $V^\myin$ be the set of vertices $v\in V(G)\setminus X$, such that any path betweeh $v$ and $\mu(\partial H)$ in $G\setminus X$ intersects $C_t^G$.
Let $V^\myout = (V(G) \setminus X) \setminus V^\myin$.
Thus, $V(G) = V^\myin \cup V^\myout \cup X$ is a partition of $V(G)$.
We define a partition $X_\OPT=X_\OPT^\myin \cup X_\OPT^\myout \cup X_\OPT^\apex$, where
$X_\OPT^\myin = X_\OPT \cap V^\myin$,
$X_\OPT^\myout = X_\OPT \cap V^\myout$, and
$X_\OPT^\apex = X_\OPT \cap X$.

Fix some embedding $\phi$ of $G\setminus X_\OPT$ into $\mathbb{S}^2$ (i.e.~the 2-sphere), and some embedding $\psi$ of $Z$ into $\mathbb{S}^2$.
For each $i\in \{1,\ldots,\ell\}$, there exists some cycle $C_i^G$ in $H[\mu(V(C_i^H))]$ such that $C_i^H$ is a contraction of $C_i^G$.
Since $H[V(C_{t-1}^H)\cup C_{t}^H\cup C_{t+1}^H)]$ is the subdivision of some 3-connected planar graph, it follows that it admits a unique planar drawing.
Since $H[\mu(V(C_i^H))]\subset G\setminus X_\OPT$, this implies that $\phi(C_{t-1}^G)$ and $\phi(C_{t+1}^G)$ are contained in distinct components of $\mathbb{S}^2 \setminus \phi(C_{t-1}^G)$.
Similarly, since $H[\mu(V(C_i^H))]\subset Z$, it follows that $\psi(C_{t-1}^G)$ and $\psi(C_{t+1}^G)$ are contained in distinct components of $\mathbb{S}^2 \setminus \psi(C_{t-1}^G)$.

Let ${\cal D}_\phi^\myin$ (resp.~${\cal D}_\phi^\myout$) be the topological disk in $\mathbb{S}^2$ with boundary $\phi(C_{t}^G)$ that contains $\phi(C_{t-1}^G)$ (resp.~$\phi(C^G_{t+1})$).
Similarly, let ${\cal D}_\psi^\myin$ (resp.~${\cal D}_\psi^\myout$) be the topological disk in $\mathbb{S}^2$ with boundary $\psi(C_{t}^G)$ that contains $\psi(C_{t-1}^G)$ (resp.~$\psi(C^G_{t+1})$).

Let
\begin{align*}
V_\phi^{\myout\to\myin} &= \{v\in V^{\myout} : \phi(v)\in {\cal D}_\phi^\myin\}\\
V_\phi^{\apex\to\myin} &= \{v\in X : \phi(v)\in {\cal D}_\phi^\myin\}\\
V_\phi^{\myin\to\myin} &= \{v\in V^{\myin} : \phi(v)\in {\cal D}_\phi^\myin\}\\
V_\phi^{\myout\to\myout} &= \{v\in V^{\myout} : \phi(v)\in {\cal D}_\phi^\myout\}\\
V_\phi^{\apex\to\myout} &= \{v\in X : \phi(v)\in {\cal D}_\phi^\myout\}\\
V_\phi^{\myin\to\myout} &= \{v\in V^{\myin} : \phi(v)\in {\cal D}_\phi^\myout\}\\
V_\psi^{\myout\to\myin} &= \{v\in V^{\myout} : \psi(v)\in {\cal D}_\psi^\myin\}\\
V_\psi^{\apex\to\myin} &= \{v\in X : \psi(v)\in {\cal D}_\psi^\myin\}\\
V_\psi^{\myin\to\myin} &= \{v\in V^{\myin} : \psi(v)\in {\cal D}_\psi^\myin\}\\
V_\psi^{\myout\to\myout} &= \{v\in V^{\myout} : \psi(v)\in {\cal D}_\psi^\myout\}\\
V_\psi^{\apex\to\myout} &= \{v\in X : \psi(v)\in {\cal D}_\psi^\myout\}\\
V_\psi^{\myin\to\myout} &= \{v\in V^{\myin} : \psi(v)\in {\cal D}_\psi^\myout\}
\end{align*}
Figure \ref{fig:irrelevant_surgery} depicts the above sets in the embeddings $\phi$ and $\psi$.
Since $H[C_{t-1}^H\cup C_t^H \cup C_{t+1}^H]$ is a subdivision of a 3-connected graph, it follows that there are no edges between $V_\phi^{\myout\to\myin}$ and $V_\phi^{\myin\to \myin}$.
Similarly, there are no edges between
$V_\phi^{\myout\to\myout}$ and $V_\phi^{\myin\to \myout}$,
between
$V_\psi^{\myout\to\myin}$ and $V_\psi^{\myin\to \myin}$,
and between
$V_\psi^{\myout\to\myout}$ and $V_\psi^{\myin\to \myout}$.
Thus, there can only be edges
between
$V_\phi^{\myout\to\myin}$ and $V_\phi^{\apex\to\myin}$,
between
$V_\phi^{\myin\to\myin}$ and $V_\phi^{\apex\to\myin}$,
between
$V_\phi^{\myout\to\myout}$ and $V_\phi^{\apex\to\myout}$,
between
$V_\phi^{\myin\to\myout}$ and $V_\phi^{\apex\to\myout}$,
between
$V_\phi^{\myin\to\myin}$ and $C_t^G$,
and between
$V_\phi^{\myout\to\myout}$ and $C_t^G$.
Similarly,
there can only be edges
between
$V_\psi^{\myout\to\myin}$ and $V_\psi^{\apex\to\myin}$,
between
$V_\psi^{\myin\to\myin}$ and $V_\psi^{\apex\to\myin}$,
between
$V_\psi^{\myout\to\myout}$ and $V_\psi^{\apex\to\myout}$,
between
$V_\psi^{\myin\to\myout}$ and $V_\psi^{\apex\to\myout}$,
between
$V_\psi^{\myin\to\myin}$ and $C_t^G$,
and between
$V_\psi^{\myout\to\myout}$ and $C_t^G$.
Figure \ref{fig:irrelevant_surgery} depicts all possible combinations of pairs that can have edges between them.

We now combine the embeddings $\phi$ and $\psi$ to obtain a new embedding $\psi'$ (see Figure \ref{fig:irrelevant_surgery}).
We start by setting $\phi'=\phi$.
We delete all the vertices in $V_\phi^{\myin\to\myin} \cup V_\phi^{\myin\to\myout}$.
By the above discussion, there are no edges between $V_\phi^{\myout\to\myin} \cup V_{\phi}^{\apex\to\myin}$ and $C_t^G$.
Thus, in the current embedding $\phi'$, there exists some disk ${\cal W}_\phi\subset {\cal D}_\phi^\myin$ that contains only the image of $V_\phi^{\myout\to\myin} \cup V_{\phi}^{\apex\to\myin}$.
We move ${\cal W}_{\phi}$ inside some face outside ${\cal D}_\phi^\myin$.
By restricting $\psi$ on $V_\psi^{\myin\to\myout}$ we obtain an embedding into some disk ${\cal W}_\psi$.
We place a copy of ${\cal W}_\psi$ inside some face outside ${\cal D}_\phi$.
Finally, we restrict $\psi$ on $C_t^G\cup (V_\psi^{\myin\to\myin}\setminus Y_1)$ to obtain an embedding into some disk ${\cal W}'$ with boundary $C_t^G$.
We replace, in $\phi'$, the disk ${\cal D}_\phi$ by ${\cal W}'$.
Finally, we restrict the resulting embedding on the graph
\[
G'=G\setminus X',
\]
where
\[
X'=(X_\OPT^\myout \cup X_\OPT^\apex \cup Y).
\]
It follows that $G'$ is planar.

\begin{figure}
\begin{center}
\scalebox{0.9}{\includegraphics{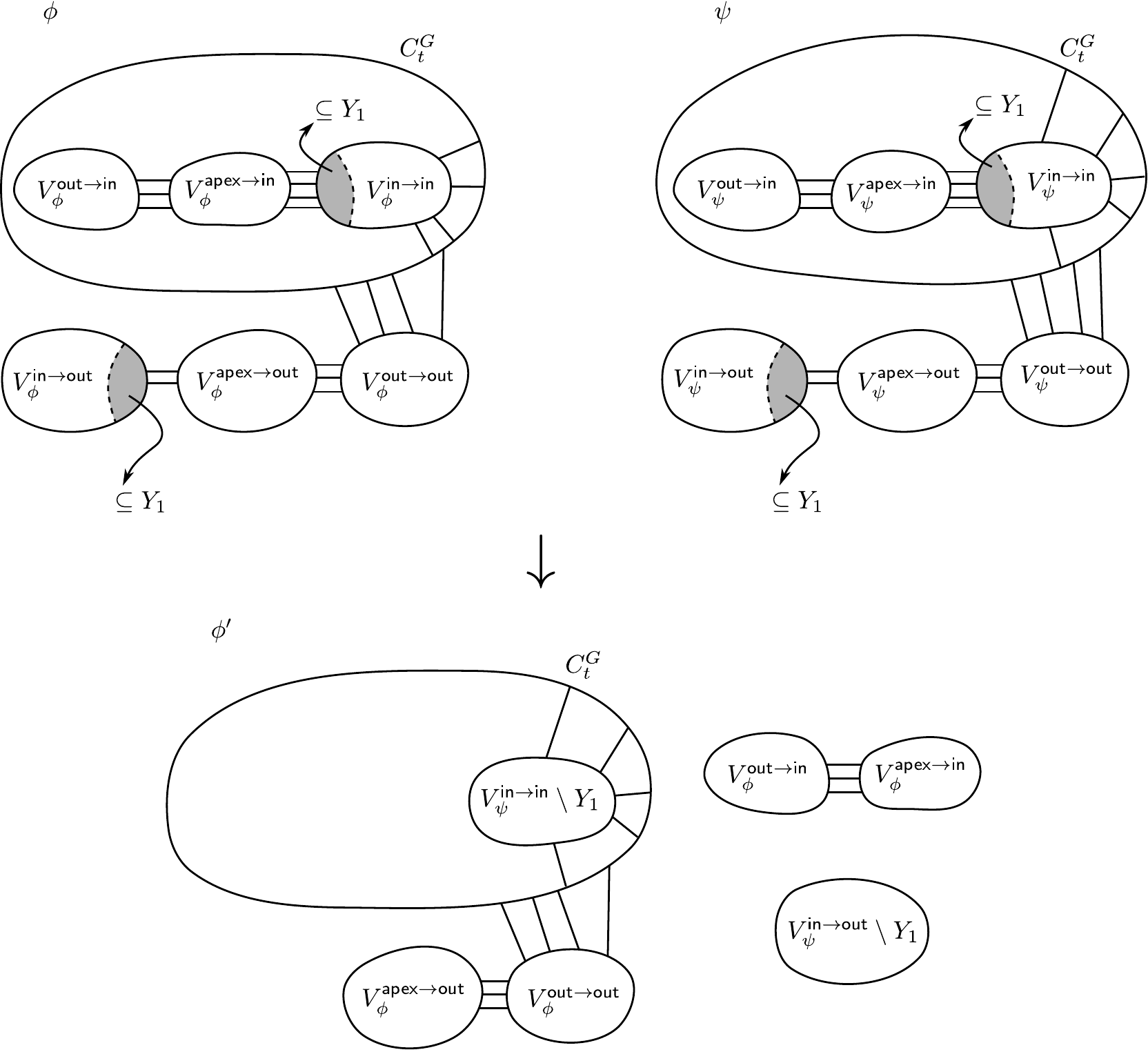}}
\caption{The embeddings $\phi$ (top left), $\psi$ (top right), and $\phi'$ (bottom).\label{fig:irrelevant_surgery}}
\end{center}
\end{figure}

Since $S(i,j,\ell)$ is dense, we get
\begin{align}
|X_\OPT^\myin| &\geq |X_\OPT\cap \mu(S(i,j,\ell))| \geq \width(S(i,j,\ell))/49 = 2\ell/49 \label{eq:X_OPT_father_large}
\end{align}
By \eqref{eq:Y_father} and \eqref{eq:X_OPT_father_large}, we obtain
\begin{align*}
|X'| &\leq |X_\OPT| - |X_\OPT^\myin| + |Y| \leq |X_\OPT| - 2\ell/49 + 8 \ell / 200 < |X_{OPT}|.
\end{align*}
This contradicts the choice of $X_\OPT$, and concludes the proof.
\end{proof}

Using Lemma \ref{lem:forgotten_irrelevant}, we show how to compute irrelevant vertices. 

\begin{lemma}[Computing an irrelevant subgrid]\label{lem:computing_dominating}
Let $G$ be an $n$-vertex graph, and let $k\in \mathbb{N}$, and let $\rho>0$.
Let $X\subset V(G)$, with $|X|\leq c k \log n$, for some universal constant $c>0$.
Let $H$ be the $(r\times r)$-partially triangulated grid, for some $r\geq c' \cdot (\log^{7/2} n) \cdot \sqrt{\approxfactor} \cdot k$, for some universal constant $c'>0$.
Suppose that $H$ is a contraction of $G\setminus X$, with contraction mapping $\mu:V(H)\to 2^{V(G\setminus X)}$.
Let $\mu':V(H')\to 2^{V(G)}$ be the contraction of $G$ induced by $\mu$, where $V(H')=V(H)\cup X$, and $\mu'$ is the identity on $X$.
Let
\[
L = |N_{H'}(X)| = |\{v\in V(H):d_H(v,\partial H)\geq 3 \text{ and } N_G(X)\cap \mu'(v)\neq \emptyset\}|.
\]
Suppose that
\[
L \leq (144|X|+1296k)c_{FHL}\log^{3/2}n,
\]
Suppose that there exists an algorithm $\AlgApprox$ which for all $n'\in \mathbb{N}$, with $n'<n$, given an $n'$-vertex graph $G'$, outputs some $S'\subset V(G')$, such that $G'\setminus S'$ is planar, with $|S'|\leq \approxfactor \cdot \mvp(G')$,
for some $\approxfactor \geq 2\rho$,
in time $T_\myapprox(n')$, where $T_\myapprox:\mathbb{N}\to\mathbb{N}$ is increasing and convex.
Then, there exists an algorithm $\AlgIrrelevant$ which given $G, k, X, \mu, H$, and $L$, terminates with one of the following outcomes:
\begin{description}
\item{(1)}
Correctly decides that $\mvp(G)>k$.
\item{(2)}
Outputs some $(6\log n \times 6\log n)$-partially triangulated subgrid $J$ of $H$,
such that if $\mvp(G)\leq k$, then $\mu(V(J))$ is irrelevant.
\end{description}
Moreover, the running time of $\AlgIrrelevant$ is at most $T_\irrelevant(n) \leq n^{O(1)} + T_\myapprox(n/\rho) 2 \rho \log n$.
\end{lemma}

\begin{proof}
The algorithm $\AlgIrrelevant$ proceeds by computing some dominating weight function $w:{\cal S}\to \mathbb{N}$,
and some $(2\log n\times 2\log n)$-partially triangulated subgrid $J$ of $H$, such that $V(J)$ is forgotten (w.r.t.~$w$).
By Lemma \ref{lem:forgotten_irrelevant}, it follows that $\mu(V(J)$ is irrelevant.
It thus suffices to show how to compute $w$ and $J$ within the desired time bound.

More specifically, the algorithm $\AlgIrrelevant$ proceeds as follows.
For each $S\in {\cal S}$, if $|\mu(S)|\leq n/\rho$, then we run $\AlgApprox$ on $G[\mu(S)]$ and we obtain some $X_S\subseteq \mu(S)$, such that $G[\mu(S)]\setminus X_S$ is planar, and with
\[
|X_S| \leq \approxfactor \mvp(G[\mu(S)]).
\]
For all $S\in {\cal S}$ we set
\[
w(S) = \left\{\begin{array}{ll}
\gamma(S) + \min\{k, |X_S|\} & \text{ if } |\mu(S)|\leq n/\rho\\
\gamma(S) + k & \text{ if } |\mu(S)|> n/\rho
\end{array}\right.
\]

Let $V^{\inner}=\{v\in V(H) : d_H(v,\partial H) > 123 k\}$.
Let $H^{\inner} = H[V^{\inner}]$; that is, $H^{\inner}$ is the $(r'\times r')$-partially triangulated subgrid of $H$, consisting of the central $r'$ rows and central $r'$ columns, where $r'=r-2 \cdot 123 k\geq r/2$.

Let $i\in \{0,\ldots,\log r\}$, and let $j\in \{1,2\}$.
For each $S\in {\cal S}$, in order for $S$ to be active (w.r.t.~$w$), it must be that $w(S)\geq \width(S) / 200 = 2^{i+1}/200$.
Thus, if $2^{i+1} > 200 (k+L)$, then there are no active $S\in {\cal S}_{i,j}$.
Otherwise, if $2^{i+1} \leq 200 (k+L)$, then each active $S\in {\cal S}_{i,j}$ must satisfy either $|\mu(S)|\geq n/\rho$, or
$\mvp(S) \geq 2^{i+1}/(\approxfactor \cdot 200)$.
Thus, the number of active $S\in {\cal S}_{i,j}$ is at most $\rho + k \cdot \approxfactor \cdot 200 / 2^{i+1}$.
Let $i^*=\lfloor \log (123(k+L)/2) \rfloor$.
Since for each $S\in {\cal S}_{i,j}$, we have $|S|\leq 2^{2i+2}$, we get
\begin{align*}
\left| \bigcup_{S\in {\cal S} : S \text{ is active}} S \right| &= \left| \bigcup_{i=0}^{i^*} \bigcup_{j=1}^2 \bigcup_{S\in {\cal S}_{i,j} : S \text{ is active}} S \right| \\
 &\leq  \sum_{i=0}^{i^*} \sum_{j=1}^2 \left| \bigcup_{S\in {\cal S}_{i,j} : S \text{ is active}} S \right| \\
  &= \sum_{i=0}^{i^*} \sum_{j=1}^2 |\{S\in {\cal S}_{i,j} : S \text{ is active}\}| \\
  &\leq \sum_{i=0}^{i^*} \sum_{j=1}^2 2^{2i+2} \left( \rho + k \approxfactor 200 / 2^{i+1}\right) \\
  &< 8 (200^2) \approxfactor (k+L)^2\\
  &< (r')^2/(6\log n)^2,
\end{align*}
where the last inequality holds for some constant
$c' = \Theta(c_{FHL} \cdot c)$.
It follows that there exists some $(6\log n \times 6\log n)$-partially triangulated subgrid $J$ of $H^{\inner}$, such that $V(J)$ does not intersect any active squares, and thus $V(J)$ is forgotten (w.r.t.~$w$).

It thus remains to bound the running time.
The algorithm $\AlgIrrelevant$ recursively calls the algorithm $\AlgApprox$ on some subgraphs $G[\mu(S)]$, for some $S\in {\cal S}$.
Let $i\in \{0,\ldots,\log r\}$, and let $j\in \{1,2\}$.
Suppose that there are recursive calls on the subgraphs $G[\mu(S_{i,j,1})],\ldots,G[\mu(S_{i,j,t})]$, for some $S_{i,j,1},\ldots,S_{i,j,t}\in {\cal S}_{i,j}$.
For each $t'\in \{1,\ldots,t\}$, let $n_{i,j,t'}=|\mu(S_{i,j,t'})|$.
Note that we only recurse on subgraphs with at most $n/\rho$ vertices.
Thus, by the fact that $T_\myapprox$ is increasing and convex, we obtain
$\sum_{t'=1}^t T_\myapprox(n_{i,j,t'}) \leq T_\myapprox(\sum_{t'=1}^t n_{i,j,t'}) \leq  T_\myapprox(n/\rho) \rho$.
Thus, the running time of $\AlgIrrelevant$ is
\begin{align*}
T_\irrelevant(n) &\leq n^{O(1)} + \sum_{i=1}^{\log n} \sum_{j=1}^2 T_\myapprox(n/\rho) \rho\\
 &= n^{O(1)} + T_\myapprox(n/\rho) 2 \rho \log n,
\end{align*}
which concludes the proof.
\end{proof}


\section{Patches and frames}\label{sec:patch}

In this Section we use the algorithm for computing irrelevant vertices to compute a patch that can be used in the computation of a pruning sequence.
A key desired property is that the framing of the patch must be smaller than the graph before the framing.




\begin{lemma}[Computing a patch]\label{lem:patch}
Let $G, n, \rho, k, X, \mu, H$, and $L$ be as in Lemma \ref{lem:computing_dominating}.
Suppose that there exists an algorithm $\AlgApprox$ which for all $n'\in \mathbb{N}$, with $n'<n$, given an $n'$-vertex graph $G'$, outputs some $S'\subset V(G')$, such that $G'\setminus S'$ is planar, with $|S'|\leq \approxfactor \cdot \mvp(G')$,
for some $\approxfactor \geq 2\rho$,
in time $T_\myapprox(n')$, where $T_\myapprox:\mathbb{N}\to\mathbb{N}$ is increasing and convex.
Then, there exists an algorithm $\AlgPatch$,
which given $G, k, X, \mu, H$, and $L$,
terminates with one of the following outcomes:
\begin{description}
\item{(1)}
Correctly decides that $\mvp(G)>k$.

\item{(2)}
Computes some patch $(\Gamma, C)$ of $G$, satisfying the following conditions:
Let $G^\framed$ be the $(\Gamma, C)$-framing of $G$.
Then $|V(G^\framed)| < |V(G)|$.
Moreover, if $\mvp(G)\leq k$, then $\mvp(G^\framed)\leq \mvp(G)$.
\end{description}
Moreover, the running time of $\AlgPatch$ is at most $T_\patch(n) \leq n^{O(1)} + T_\myapprox(n/\rho)  2 \rho \log n$.
\end{lemma}

\begin{proof}
We first run the algorithm $\AlgIrrelevant$ from Lemma \ref{lem:computing_dominating}.
If $\AlgIrrelevant$ decides that $\mvp(G)>k$, then we terminate with outcome (1).
Otherwise,
 we obtain some  $(6\log n \times 6\log n)$-partially triangulated subgrid $J$ of $H$, such that if $\mvp(G)\leq k$, then $\mu(V(\Gamma))$ is irrelevant (w.r.t.~$G$).
The partially triangulated grid $J$ contains $\ell=3\log n$ nested cycles $K^H_1,\ldots,K^H_{\ell}$; that is, for each $i\in \{2,\ldots,\ell-1\}$, we have that $K^H_i$ separates $K^H_{1}\cup \ldots \cup K^H_{i-1}$ from $K^H_{i+1}\cup \ldots\cup K^H_\ell \cup (H\setminus J)$ in $H$.
For each $i\in \{1,\ldots,\ell\}$ we compute some cycle $K_i^G$ in $G$ such that $K_i^H$ is a contraction of $K^G_i$; this can clearly be done in polynomial time by computing, for each $v\in V(K_i^H)$, an arbitrary path in $G[\mu(v)]$ between the endpoints in $\mu(v)$ of the two edges that are incident to $v$ in $K_i^H$.
For each $i\in \{1,\ldots,\ell-1\}$ let
\[
U_i = \{v\in \mu(V(\Gamma)) : K_{i}^G \text{ separates } v \text{ and } \mu(V(K_{\ell}^G)) \text{ in } G\},
\]
and
\[
\Gamma_i = G[U_i].
\]
Since $V(\Gamma)$ is forgotten, we have that $N_{G}(X) \cap \mu(V(\Gamma))=\emptyset$.
It follows that for all $i\in \{1,\ldots,\ell-1\}$, $(\Gamma_i, K_i^G)$ is a patch in $G$.
Since $\ell=3\log n$ and $V(G)=n$ it follows that there exists some $i^*\in \{2,\ldots,\ell-1\}$, such that
\begin{align}
|V(K_{i^*}^G)| < \frac{1}{3}(|V(K_1^G)| + \ldots + |V(K_{i^*-1}^G)|) \label{eq:Gamma_size}
\end{align}
Let $\Gamma=\Gamma_{i^*}$ and $C=K_{i^*}^G$.
Let $G^\framed$ be the $(\Gamma, C)$-framing of $G$.
It follows that by \eqref{eq:Gamma_size} that $|V(\Gamma)|>4|V(C)|$.
Thus $|V(G^\framed)|=|V(G)| + 4|V(C)| - |V(\Gamma)| < |V(G)|$, and thus condition (1) holds.

Let $X_\OPT\subset V(G)$, with $|X_\OPT|=\mvp(G)$, such that $G\setminus X_\OPT$ is planar.
If $\mvp(G)\leq k$, then $\mu(V(J))$ is irrelevant in $G$.
It follows that $X_\OPT \cap \mu(V(J))=\emptyset$.
Thus $X_\OPT \cap (V(K_{i^*}^G) \cup V(K_{i^*+1}^G) = \emptyset$.
Since $H[V(K_{i^*}^H) \cup V(K_{i^*+1}^H)]$ is the subdivision of some 3-connected graph,
it follows that for any planar embedding $\phi$ of $G$,
there exists an embedding $\phi'$ of $G$ obtained from $\phi$ by a sequence of zero or more Whitney flips,
such that
the cycle $C$ bounds a disk ${\cal D}$ with $\phi'(V(G)) \cap {\cal D} = \phi'(V(\Gamma))$.
Let $\phi''$ be the embedding obtained by restricting $\phi'$ on $G \setminus (V(\Gamma)\setminus V(C))$.
It follows that $C$ bounds a face in $\phi''$ (i.e.~it is a facial cycle).
Thus $\phi''$ can be extended to a planar drawing of $G^\framed$.
We thus obtain that $G^\framed\setminus X_{\OPT}$ is planar, and thus $\mvp(G^\framed) \leq \mvp(G)$, which establishes condition (2).

By Lemma \ref{lem:computing_dominating} we get $T_\patch(n) \leq n^{O(1)} + T_{\irrelevant}(n) \leq n^{O(1)} + T_{\myapprox}(n/\rho) 2\rho \log n$, which concludes the proof.
\end{proof}



\section{Computing a pruning sequence}\label{sec:pruning}

In the previous section, we find irrelevant vertices. This allows us to define the 
following ``pruning sequence'' in our algorithm.

The main result of this section is to compute a pruning sequence, which will be presented in the next subsection. 

\subsection{The algorithm for computing a pruning sequence}

We now describe an algorithm for computing a pruning sequence for a given graph $G$.
We give two algorithms: $\AlgPruning$ and $\AlgPruningDecision$.
The algorithm $\AlgPruning$ gets as input some graph $G$ and outputs a pruning sequence for $G$.
The algorithm $\AlgPruningDecision$ gets as input some graph $G$ and some $k\in \mathbb{N}$, and either returns $\nil$, which indicates the fact that $\mvp(G)>k$, or outputs some pruning sequence for $G$.
The algorithm $\AlgPruning$ calls recursively algorithm $\AlgPruningDecision$, and $\AlgPruningDecision$ calls recursively either itself or $\AlgPruning$.

The algorithms use a parameter $\rho>0$, which is as in Lemmas \ref{lem:computing_dominating} and \ref{lem:patch}.
The parameter $\rho$ which allows us to obtain a trade-off between the approximation ratio and the running time.
In particular, the approximation ratio of the final algorithm increases and the running time decreases when $\rho$ increases.
We also use a parameter $\approxfactor>0$ which denotes the target approximation factor.

The formal description of algorithm $\AlgPruning$ is as follows.

\begin{flushleft}
\textbf{Algorithm} $\AlgPruning(G)$:
\end{flushleft}
\begin{description}
\item{\textbf{Step 1: The main loop.}}
We consider all values $k=1,2,4,\ldots,2^i,\ldots,n$, which, intuitively, are used as ``approximate guesses'' for $\mvp(G)$.
For each such value of $k$, we execute
\[
\AlgPruningDecision(G,k).
\]
If the execution returns $\nil$ then repeat Step 1 for the next value of $k$.
Otherwise, we output the resulting pruning sequence found.
\end{description}

This completes the description of the algorithm $\AlgPruning$.
We next give the formal description of the algorithm $\AlgPruningDecision$.

\begin{flushleft}
\textbf{Algorithm} $\AlgPruningDecision(G,k)$:
\end{flushleft}
\begin{description}
\item{\textbf{Step 1:}}
Let $t=c'' \cdot \log^{15/2} n \cdot \sqrt{\approxfactor} \cdot k$, for some sufficiently large constant $c''>0$ to be determined.
We have $t > \alpha' k \log^6 n$, where $\alpha'>0$ is the universal constant in Corrolary \ref{cor:k-apex_grid_minor_general_combed}.
Using Corrolary \ref{cor:k-apex_grid_minor_general_combed}, in polynomial time, we obtain one of the following outcomes:
\begin{description}
\item{\textbf{Case 1:}}
We correctly decide that $\mvp(G)>k$. In this case, we return $\nil$.

\item{\textbf{Case 2: Computing a small balanced separator.}}
We compute a $3/4$-balanced separator $S$ of $G$, with $|S| = O(t \log^{13/2} n)$.
In this case we partition $G\setminus S$ into two vertex-disjoint subgraphs $G_1$ and $G_2$ such that there are no edges between $V(G_1)$ and $V(G_2)$, and with $|V(G_1)|\leq 3n/4$, and $|V(G_2)|\leq 3n/4$.
We recursively call $\AlgPruning(G_1)$ and $\AlgPruning(G_2)$, and we obtain pruning sequences
${\cal G}_1=(G_{1,0},A_{1,0}), \ldots (G_{1,\ell},A_{1,\ell})$ for $G_1$,
and
${\cal G}_2=(G_{2,0},A_{2,0}), \ldots (G_{2,\ell'},A_{2,\ell'})$ for $G_2$.
We output
\[
{\cal G}=
(G_{1,0}\cup G_{2,0},A_{1,0}),
\ldots,
(G_{1,\ell}\cup G_{2,0},A_{1,\ell}),
(G_{1,\ell}\cup G_{2,1},A_{2,0}),
\ldots,
(G_{1,\ell}\cup G_{2,\ell'}, A_{2,\ell'}).
\]
Since $G_{1,\ell}$ and $G_{2,\ell'}$ are planar graphs, it follows that $G_{1,\ell}\cup G_{2,\ell'}$ is planar, and thus ${\cal G}$ is a pruning sequence for $G$.
If $\cost({\cal G}) > c''' k \sqrt{\approxfactor} \log^{15} n$, for some sufficiently large constant $c'''>0$ to be specified, then we return $\nil$, and otherwise we return ${\cal G}$.

\item{\textbf{Case 3: Computing a large grid minor.}}
We compute some combed pseudogrid $({\cal P}, {\cal Q})$ in $G$, with $|{\cal P}|=|{\cal Q}|=r$, for some $r=\Omega(t/\log^4 n)$.
Setting $c''$ to be some sufficiently large constant, we get that $r\geq c k \log n$, where $c$ is the universal constant from Lemma \ref{lem:pt_grid}, and thus the conditions of Lemma \ref{lem:pt_grid} are satisfied.
Using Lemma \ref{lem:pt_grid}, in polynomial time, we obtain one of the following outcomes:
\begin{description}
\item{\textbf{Case 3.1:}}
We can correctly decide that $\mvp(G)>k$. In this case, we return $\nil$.

\item{\textbf{Case 3.2: Computing a partially triangulated grid contraction with a few apices.}}
We compute some $X\subset V(G)$, with $|X|=O(k \log n)$, some $(r'\times r')$-partially triangulated grid $H$, for some $r'\geq r/8$, and some contraction mapping $\mu:V(H)\to 2^{V(G\setminus X)}$.
Let $\mu':V(H')\to 2^{V(G)}$ be the contraction mapping of $G$ induced by $\mu$, where $V(H')=V(H)\cup X$, and $\mu'$ is the identity on $X$.
Let
\[
L = |N_{H'}(X)| = |\{v\in V(H) : d_H(v,\partial H) \geq 3 \text{ and } N_G(X)\cap \mu'(v)\}|.
\]
We consider the following two cases:
\begin{description}
\item{\textbf{Case 3.2.1: Computing a semi-universal set.}}
Suppose that
\[
L>(144|X|+1296k)c_{FHL}\log^{3/2}n.
\]
Then using Lemma \ref{lem:semi-universal} we can compute, in polynomial time, some non-empty $Y\subseteq X$, such that if $\mvp(G)\leq k$, then $Y$ is semi-universal.
We recursively call
\[
\AlgPruningDecision(G\setminus Y,k).
\]
If the recursive call returns $\nil$ then we return $\nil$.
Otherwise, we obtain a pruning sequence
\[
(G_0,A_0),\ldots,(G_\ell,A_\ell)
\]
for $G\setminus Y$.
We define the pruning sequence
\[
{\cal G} = (G, Y),(G_0,A_0),\ldots,(G_\ell,A_\ell)
\]
for $G$.
If $\cost({\cal G}) > c''' k \sqrt{\approxfactor} \log^{15} n$, for some sufficiently large constant $c'''>0$ to be specified, then we return $\nil$, and otherwise we return ${\cal G}$.

\item{\textbf{Case 3.2.2: Computing an irrelevant patch.}}
Suppose that
\[
L \leq (144|X|+1296k)c_{FHL}\log^{3/2}n.
\]
By setting $c''$ to be a sufficiently large constant, we get that
$r'\geq c' (\log^{7/2} n) \sqrt{\approxfactor} k$,
where $c'$ is the universal constant from Lemma \ref{lem:patch},
 and thus the conditions of Lemma \ref{lem:patch} are satisfied.
Using the algorithm from Lemma \ref{lem:patch}, we obtain one of the following two outcomes:
(i) We either decide that $\mvp(G)>k$; in this case, we go to Step 1, and consider the next value for $k$.
(ii) We compute some patch $(\Gamma, C)$ of $G$, satisfying the following conditions:
Let $G^\framed$ be the $(\Gamma, C)$-framing of $G$.
Then $|V(G^\framed)| < |V(G)|$.
Moreover, if $\mvp(G)\leq k$, then $\mvp(G^\framed) \leq \mvp(G)$.
We recursively call
\[
\AlgPruningDecision(G^\framed,k).
\]
If the recursive call returns $\nil$, then we return $\nil$.
Otherwise, we obtain a pruning sequence
\[
(G_0,A_0),\ldots,(G_\ell,A_\ell)
\]
for $G^\framed$.
We return the pruning sequence
\[
{\cal G} = (G, (\Gamma,C)),(G_0,A_0),\ldots,(G_\ell,A_\ell)
\]
for $G$.
\end{description}
\end{description}
\end{description}
\end{description}
This completes the description of the algorithm $\AlgPruningDecision$.

\subsection{Analysis}

We now give analysis of the above algorithm.

\begin{proof}[Proof of Lemma \ref{lem:pruning}]
The fact that the algorithm always outputs some pruning sequence follows by the fact that the algorithms in Corollary \ref{cor:k-apex_grid_minor_general_combed} and Lemmas \ref{lem:pt_grid}, \ref{lem:semi-universal}, and \ref{lem:patch} always terminate with an outcome different than deciding that $\mvp(G)>k$, when $\mvp(G)\geq k$; thus, since the algorithm $\AlgPruning$ eventually considers a large enough value of $k$, all these algorithms output the desired structure, and thus $\AlgPruning$ outputs some pruning sequence.

Next, we bound the cost of the pruning sequence computed by the algorithm $\AlgPruning$.
The proof is via double induction on $|V(G)|$ and $\mvp(G)$.
That is, we assume that the assertion holds for all graphs $G'$ with $|V(G')|\leq |V(G)$, $\mvp(G')\leq \mvp(G)$, and either $|V(G)|<|V(G')|$ or $\mvp(G')<\mvp(G)$, and we prove the assertion for $G$.
The cost of the pruning sequence can increase only in Case 2 and Case 3.2.1.
In Case 2, we compute some $3/4$-balanced vertex separator $S$, with $|S|=O(t \log^{13/2} n) = O(\mvp(G)  \sqrt{\approxfactor} \log^{14} n)$.
The cost of the pruning sequence increases by $|S|$, and we can charge this increase to the $\mvp(G)$ vertices in some optimal solution for $G$, where each vertex receives $O(\sqrt{\approxfactor} \log^{14} n)$ units of charge.
Since we recurse on $G_1$ and $G_2$, the size of the graph decreases by a factor of $3/4$, and thus each vertex can be charged at most $O(\log n)$ times throughout the recursion.
Thus, each vertex of some optimal planarizing set receives at most $O(\sqrt{\approxfactor} \log^{15} n)$ units of charge.
Thus, the total increase in the cost of the resulting pruning sequence due to vertices in balanced separators computed in Case 3, over all recursive executions, is at most $O(\mvp(G) \sqrt{\approxfactor} \log^{15} n)$.
The increase in the cost in Case 3.2.1 is due to the removal of some set $Y$.
By the definition of a semi-universal set, we have that, if $\mvp(G)\leq k$, then by deleting $Y$, decreases the cost of the optimum solution by at least $(2/3)\cdot |Y|$.
Thus, the total increase in the cost the pruning sequence due to vertices deleted in Case 3.2.1, throughout all recursive calls, is at most $(3/2)\cdot \mvp(G)$.
We conclude that the total cost of the resulting pruning sequence is at most $O(\mvp(G) \sqrt{\approxfactor} \log^{15} n) + (3/2) \mvp(G) = O(\mvp(G) \sqrt{\approxfactor} \log^{15} n)$.

Finally, we bound the running time.
Let $T_\pruning(n,k)$ denote the running time of $\AlgPruningDecision$, and let $T_\pruning(n)$ denote the running time of $\AlgPruning$.
We have
\[
T_\pruning(n) \leq n^{O(1)} + \sum_{i=0}^{\log n} T_\pruning(n, 2^i) \leq n^{O(1)} + T_\pruning(n,k) \log n.
\]
In Case 2, we get
\[
T_\pruning(n,k) \leq n^{O(1)} + T_\pruning(n/4)+T_\pruning(3n/4).
\]
In Case 3.2.1, we get
\[
T_\pruning(n,k) \leq n^{O(1)} + T_\pruning(n-1,k).
\]
In Case 3.2.2, we get
\[
T_\pruning(n,k) \leq n^{O(1)} + T_\myapprox(n/\rho)2\rho \log n + T_\pruning(n-1,k).
\]
Thus, in all cases, we get
\[
T_\pruning(n,k) \leq n^{O(1)} + \max\{T_\pruning(n/4)+T_\pruning(3n/4), T_\myapprox(n/\rho)2\rho \log n + T_\pruning(n-1,k)\}.
\]
which concludes the proof.
\end{proof}


\section{Embedding into a higher genus surface}\label{sec:embedding}

We are now given a pruning sequence for $G$. We then add a sequence of handles and crosscaps to obtain a surface $S$ of Euler genus $O(\cost({\cal G}))$ so that $G$ can be embedded into $S$. More precisely we prove the following.

\begin{proof}[Proof of Lemma \ref{lem:embedding}]
Let ${\cal G} = (G_0,A_0),\ldots,(G_\ell,A_\ell)$.
Let $D=\{i\in \{1,\ldots,\ell-1\} : i \text{ is a deletion step}\}$,
and $P=\{i\in \{1,\ldots,\ell-1\} : i \text{ is a patching step}\}$.
Let
\[
X = \bigcup_{i\in D} A_i.
\]
Clearly, we have $|X|=\cost({\cal G})$.
It thus remains to construct an embedding for $G\setminus X$.

For all $i\in \{0,\ldots,\ell\}$, let $G_i'=G_i\setminus X$.
We construct an embedding for $G\setminus X$ via reverse induction on the pruning sequence ${\cal G}$.
Specifically,
for each $i\in \{0,\ldots,\ell\}$ we construct an embedding $\phi_i$ of $G_i'$ into some surface ${\cal S}_i$.
Recall that by the definition of a pruning sequence, $G_{\ell'}=G_\ell\setminus X=G_\ell$, which is a planar graph.
Let $\phi_{\ell}$ be an arbitrary planar embedding of $G_\ell'$.
Let $i\in \{0,\ldots,\ell-1\}$ and suppose we have already constructed the embedding $\phi_{i+1}$; we proceed to construct $\phi_i$.
We consider the following two cases:
\begin{description}
\item{Case 1:}
Suppose that $i$ is a deletion step of ${\cal G}$.
We have $G_{i+1}'=G_{i+1}\setminus X=(G_{i}\setminus A_i)\setminus X = G_i\setminus X = G_i'$.
Thus we may set $\phi_i=\phi_{i+1}$.
\item{Case 2:}
Suppose that $i$ is a framing step of ${\cal G}$.
We have that $A_i=(\Gamma_i,C_i)$ is some patch of $G_i$, and $G_{i+1}$ is the $(\Gamma_i,C_i)$-framing of $G_i$.
By the definition of a framing, we have that $V(G_{i+1})=V(G_i) \setminus (V(\Gamma_i) \setminus V(C_i)) \cup \bigcup_{i=1^3} V(C_i^j)$, where for each $j\in \{1,2,3\}$, $C_i^j$ is a new cycle added to $G_{i+1}$ with $|V(C_i^j)|=|V(C_i)|$.
We define $C_i^0 = C_i$.
For each $j\in \{0,\ldots,3\}$ let
$V(C_i^j)=\{v_1^j,\ldots,v_t^j\}$, where the numbering of the indices agrees with a clockwise traversal of $C_i^j$, and for all $t'\in \{1,\ldots,t\}$, and for all $j\in \{0,\ldots,2\}$, we have $\{v_{t'}^j,v_{t'}^j\}\in E(G_{i+1})$.
Let $F_i$ be the $(4\times t)$-cylinder that is a subgraph of $G_{i+1}$ on the vertex set $V(C_i^0\cup \ldots \cup C_i^3)$.
Let also $F'_i$ be the $(3\times t)$-cylinder subgraph of $G_{i+1}$ on the vertex set $V(C_i^1\cup \ldots C_i^3)$.
We partition $F'_i$ into a collection ${\cal F}_i$ of vertex-disjoint subgraphs as follows.
If $V(F'_i)\cap X=\emptyset$ then we set ${\cal F}_i = \{F'_i\}$.
Otherwise, let ${\cal F'}_i$ be the set of columns of $F'_i$ that intersect $X$.
Let ${\cal C}_i$ be the set of connected components of $F'_i$ obtained after deleting all columns in ${\cal F}_i$.
Each component $C\in {\cal C}_i$ is a $(3\times a)$-grid, for some $a\geq 1$.
If $a\leq 2$, then we further partition  $C$ to $a$ disjoint columns, and we add them to the set ${\cal F}''_i$.
Let ${\cal C}'_i$ be the set of connected components of $F'_i$ after we delete all the rows in ${\cal F'}_i\cap {\cal F}''_i$.
We set ${\cal F}_i = {\cal F}'_i\cup {\cal F}''_i\cup {\cal C}'_i$.
It is immediate that $n_i=|{\cal F}_i| \leq 3|F_i'\cap X|$.
By the definition of a $(\Gamma,C)$-patching, it follows that for all $i\neq i'\in P$, the graphs $F_i'$ and $F_{i'}'$ are vertex-disjoint.
Thus
\[
\sum_{i \in P} n_i \leq 3 \cdot \cost({\cal G}).
\]
We now modify $\phi_{i+1}$ to obtain an embedding $\phi_i$ of $G_i$.
If $V(F_i')\cap X=\emptyset$ then, by reverse induction on $i$, there exists a face $f_i$ of $\phi_{i+1}$ that is bounded by $C_i^3$.
By the definition of a patch, there exists a planar embedding $\psi_i$ of $\Gamma_i$ into some disk ${\cal D}_i$ with $\partial {\cal D}_i=\psi_i(C_i)$.
By identifying $C_i$ in $\psi_i$ with $C_i^3$ in $\phi_i$, we can past the disk ${\cal D}_i$ inside the face $f_i$, thus extending $\phi_i$ to $\Gamma_i(\setminus C_i\setminus X)$.
Finally, for each $v\in C_i\setminus X$, we contract the column of $F_i$ containing $v$ into $v$.
It remains to consider the case $V(F_i')\cap X \neq \emptyset$.
As above, fix some planar drawing $\psi_i$ of $\Gamma_i$ into some disk ${\cal D}_i$ with $\partial {\cal D}_i=\psi_i(C_i)$.
By construction, each graph in $J\in {\cal F}_i$ is either a path, or some $(3\times c)$-grid, for some $a\geq 3$.
Let $K=J\cap C_i^3$.
It follows that, by reverse induction on $i$, there exist a face of $f$ of $\phi_{i+1}$ that contains $K$ as a subpath.
We may thus identify the image of $K$ in $\psi$ with its image in $\psi_{i+1}$ by adding either a handle (if the orientations of the two images are the same) or an antihandle (if the orientations are opposite).
Finally, for each $v\in C_i\setminus X$, we contract the column of $F'$ containing $v$ into $v$; for each $v\in C_i\cap X$, we delete that column.
\end{description}
It is immediate by induction that $\phi_0$ is an embedding of $G\setminus X$ into some surface ${\cal S}_0$.
For each $i\in D$, we have ${\cal S}_i={\cal S}_{i+1}$.
For each $i\in P$ with $F_i'\cap X=\emptyset$ we do not add any handles or antihandles.
Moreover, for each $i\in P$ with $F_i'\cap X\neq\emptyset$ we add at most one handle or antihandle for each graph in ${\cal F}_i$.
Each handle or antihandle added to the embedding increases the Euler genus of the underlying surface by at most 2.
Thus
\[
\eg({\cal S}_0)\leq 2 \sum_{i\in P} n_i \leq 6 \cdot \cost({\cal G}),
\]
which concludes the proof.
\end{proof}


\section{Planarizing a surface-embedded graph}\label{sec:surface_planarization}

It remains to consider the case when a graph $G$ is embedded in a surface of Euler genus $O(\cost({\cal G}))$, where $G$ has a vertex set $X'$ such that $G-X'$ is planar. 
We shall find some planarizing set $X$ for $G$, with $|X| = O(g \log n + \mvp(G) \log^2 n)$. 
To this end, we need some definitions.

For some graph $G$ and some embedding $\phi$ of $G$ into some surface ${\cal S}$, a \emph{$\phi$-noose} is a non-separating cycle in ${\cal S}$ that intersects the image of $G$ only on $\phi(V(G))$.
We recall the following result concerning embeddings of planar graphs into non-planar surfaces \cite{mohar1996planar}.

\begin{lemma}[Mohar \cite{mohar1996planar}]\label{lem:noose2}
Let $\phi$ be an embedding of some planar graph into some non-planar surface.
Then there exists a $\phi$-noose of length at most two.
\end{lemma}

The proof of the following Lemma is similar to an argument used in \cite{DBLP:journals/dcg/EricksonH04}.

\begin{lemma}[Existence of a short noose]\label{lem:short_noose}
Let $G$ be a graph, and let $\phi$ be a cellular embedding of $G$ into some surface ${\cal S}$ of Euler genus $g>0$.
Then there exists some $\phi$-noose of length $O((1+\mvp(G) /g)\log(k+g))$.
\end{lemma}

\begin{proof}
Let $k=\mvp(G)$.
Let $X\subset V(G)$, with $|X|=k$, such that $G\setminus X$ is planar.
We inductively define a sequence of graphs $G_0,\ldots,G_t$, for some $t=O(g)$, with $G_0\supset G_1 \supset \ldots \supset G_t$, and for each $i\in \{0,\ldots,t\}$, we define some embedding of $G_i$ into some surface ${\cal S}_i$.
We will inductively maintain the invariant that for all $i\in \{0,\ldots,t\}$, the surface ${\cal S}_i$ has zero or more punctures, and is obtained from ${\cal S}_{i-1}$ by cutting along some curves.
We also define ${\cal S}_i'$ to be the surface obtained from ${\cal S}_i$ by attaching a disk on each puncture (i.e.~identifying the boundary of each puncture with the boundary of some disk).
Thus, ${\cal S}_i'$ is a surface with no punctures.
We also define $\phi_i'$ to be the embedding of $G_i$ into ${\cal S}_i'$ induced by inclusion.
Moreover, for all $i\in \{1,\ldots,t\}$, $G_i=G_{i-1}\setminus A_{i-1}$, for some $A_{i-1}\subset V(G_{i-1})$.
For the basis of the induction, let $\phi_0$ be the restriction of $\phi$ on $G_0=G\setminus X$, and let ${\cal S}_0={\cal S}$.
Next, let $i\in \{1,\ldots,\ell\}$.
Since $G\setminus X$ is planar, and $G_{i-1}\subset G_0=G\setminus X$, it follows that $G_{i-1}$ is also planar.
By Lemma \ref{lem:noose2} we get that there exists some non-contractible $\phi_{i-1}'$-noose $\gamma_{i-1}'$ of length at most two.
Let $A_{i-1}$ be the set of vertices of $G_{i-1}$ that $\gamma_{i-1}'$ intersects.
Let $\gamma_{i-1}$ be the arrangement of curves obtained by restricting $\gamma_{i-1}'$ on ${\cal S}_{i-1}$.
Let $\delta_{i-1}$ be the arrangement of curves obtained by restricting $\gamma_{i-1}$ on ${\cal S}$.
The curves in $\delta_{i-1}$ can intersect the image of $G$ on edges.
Every such edge must be incident to some vertex in $X\cup A_0\cup \ldots \cup A_{i-1}$.
Thus, by locally modifying $\gamma'_{i-1}$ we can ensure that each such intersection is in $X\cup A_0\cup \ldots \cup A_{i-1}$.
We define ${\cal S}_i$ to be the surface obtained by cutting ${\cal S}_{i-1}$ along $\gamma_{i-1}$.
We also define $G_i=G_{i-1}\setminus A_{i-1}$, and $\phi_i$ the corresponding induced embedding of $G_i$ on ${\cal S}_i$.
Since cutting a surface along a non-contractible curve either increases the number of its non-planar connected components, or it decreases the Euler genus by at most two, it follows that after $t=O(g)$ steps, the surface ${\cal S}_t$ is some punctured sphere.
By induction we have that cutting ${\cal S}$ along $\Gamma=\delta_0\cup\ldots\cup\delta_{t-1}$ we obtain a punctured sphere ${\cal S}_t$.

Let $Y=X\cup A_0\cup \ldots\cup A_{t-1}$.
Since for all $i\in \{0,\ldots,t-1\}$, $|A_i|\leq 2$, it follows that $|Y|\leq k+O(g)$.
For every face $f$ of $\phi$ such that $V(f)\cap Y \neq \emptyset$, we pick some ``representative'' vertex $v(f)\in V(f)\cap Y$.
Each segment in $\Gamma$ that is contained in some face of $\phi$ connects some pair of vertices $x$ and $y$ in $Y$; we replace this segment by a segment between $x$ and $v(f)$ and a segment between $y$ and $v(f)$.
Let $\Gamma'$ be the resulting arrangement of curves.
It is immediate that $\Gamma'$ also cuts ${\cal S}$ into some punctured sphere.
Let $J$ be the graph with $V(J)=Y$ and with each edge of $J$ corresponding to a segment in $\Gamma'$ contained in some face of $\phi$ and connecting two vertices in $Y$.
Let $J'$ be a minimal subgraph of $J$ such that cutting ${\cal S}$ along $J'$ results into a punctured sphere.
It follows that each cycle in $J'$ is non-contractible in ${\cal S}$.
By the bound on Euler's characteristic, we have that $|E(J)| \leq 3 |V(J)| - 6 + 6g = O(k+g)$.
Let $J''$ be the graph obtained from $J'$ by replacing each maximal induced path by a single edge.
We also delete from $J''$ all vertices of degree 1.
Finally, we replace each vertex of $J''$ of degree greater than 3 by a tree where each vertex has degree exactly 3.
Since $|E(J)|=O(k+g)$, it follows that $|V(J'')|=O(k+g)$.

The graph $J''$ has at least $g$ edges; this is because cutting ${\cal S}$ along $J''$ results into a punctured sphere, which is a surface of genus 0, and glueing along a single edge of $J''$ can increase the genus by at most 1.
Arguing as in \cite{DBLP:journals/dcg/EricksonH04}, we have that each graph trivalent graph of girth $\ell$ must have at least $2^{\Omega(\ell)}$ vertices.
It follows that there exists some cycle $C$ in $J''$ containing at most $O(\log (k+g))$ edges.
We can delete from $J''$ all the vertices in $C$ and all their incident edges; this creates $|V(C)|$ paths of length 2 which we contract into single edges.
The resulting graph has at least $|E(J'')| - 4|E(C)|\geq |E(J'')| - O(\log(k+g))$ edges.
Thus we may repeat this process at least $\Omega(g/\log(k+g))$ times, thus obtaining a set of $\Omega(g/\log(k+g))$ disjoint cycles in $J''$.
Thus the shortest cycle $C^*$ in this set contains at most $O(|E(J'')| / (g/\log(k+g))) = O((1+k/g)\log(k+g))$ edges.
The cycle $C^*$ corresponds to a $\phi$-noose with at most $O((1+k/g)\log(k+g))$ vertices, which concludes the proof.
\end{proof}

We are now ready to obtain the main result of this section using Lemma \ref{lem:short_noose}. 

\begin{proof}[Proof of Lemma \ref{lem:planarize_surface}]
We inductively compute a sequence of graphs $G_0,\ldots,G_{t}$, with $G_0\supset \ldots \supset G_t$, and for each $i\in \{0,\ldots,t\}$, some embedding $\phi_i$ of $G_i$ into some surface ${\cal S}_i$.
We will inductively maintain the invariant that $\eg({\cal S}_i)\leq \eg({\cal S}_{i-1})-1$.
We set $G_0=G$, ${\cal S}_0={\cal S}$, and $\phi_0=\phi$.
Let $i\in \{1,\ldots,t\}$, and suppose we have already computed $G_{i-1}$ and $\phi_{i-1}$.

Using the algorithm from \cite{DBLP:conf/compgeom/CabelloVL10a},
we compute the shortest $\phi_{i-1}$-noose $\gamma_{i-1}$.
Let $X_{i-1}$ be the set of vertices of $G_{i-1}$ that $\gamma_{i-1}$ intersects.
By Lemma \ref{lem:short_noose} we have
\begin{align*}
|X_{i-1}| &= O\left(\left(1+\frac{\mvp(G_{i-1})}{\eg({\cal S}_{i-1})}\right) \log(\mvp(G_{i-1}) + \eg({\cal S}_{i-1}))\right) \\
 &= O\left(\left(1+\frac{\mvp(G)}{\eg({\cal S}_{i-1})}\right) \log n\right)
\end{align*}
We set $G_i=G_{i-1}\setminus X_{i-1}$, we let ${\cal S}_i$ be the surface obtained by cutting ${\cal S}_{i-1}$ along $\gamma_{i-1}$, and we set $\phi_i$ be the induced embedding of $G_i$ into ${\cal S}_i$.
If ${\cal S}_i$ is planar, then we terminate the sequence at $t=i$.
This completes the definition of the sequence $G_0,\ldots,G_t$.

We let $X=X_0\cup \ldots \cup X_{t-1}$.
Since ${\cal S}_t$ is planar, and $G_t$ is embedded into ${\cal S}_t$, it follows that $G_t$ is planar.
This $X$ is a valid planarizing set for $G$.
It remains to bound $|X|$.
Since for all $i\in \{1,\ldots,t\}$ we have $\eg({\cal S}_i)\leq \eg({\cal S}_{i-1})$, it follows that $t\leq g$.
We have
\begin{align*}
|X| &= \sum_{i=0}^{t-1} |X_i| \\
 &= \sum_{i=0}^{t-1} O\left(\left(1+\frac{\mvp(G)}{\eg({\cal S}_{i})}\right) \log n\right)\\
 &\leq \sum_{i=0}^{g-1} O\left(\left(1+\frac{\mvp(G)}{g-i}\right) \log n\right)\\
 &= O(g \log n + \mvp(G) \log^2 n),
\end{align*}
which concludes the proof.
\end{proof}

\bibliographystyle{alpha}
\bibliography{bibfile-tasos}

\end{document}